%% file: main.tex
\PassOptionsToPackage{dvipsnames}{xcolor}
\documentclass[runningheads]{llncs}

\usepackage{geometry}
\usepackage{babel}
\usepackage[sort&compress,square,comma,authoryear]{natbib}
\usepackage{hyperref}
\usepackage{amsmath}
\usepackage{gensymb}
\usepackage{latexsym}
\usepackage{amssymb}
\usepackage{mathpartir}
\usepackage{tikz}
\usepackage{xspace}
\usepackage{listings}
\usepackage{stmaryrd}
\usepackage{subcaption}
\usepackage{cleveref}

\crefname{appendix}{supplement}{supplements}

\input{macros}

\begin{document}
\title{Parameterized Algebraic Protocols}

\author{Andreia Mordido\inst{1} \and
Janek Spaderna\inst{2} \and
Peter Thiemann\inst{2} \and
Vasco T. Vasconcelos\inst{1}
}

\authorrunning{A. Mordido, J. Spaderna, P. Thiemann, and V.T. Vasconcelos}

\institute{LASIGE, Faculdade de Ciências, Universidade de Lisboa, Lisbon, Portugal \and
University of Freiburg, Freiburg, Germany}

\maketitle

\input{abstract}
\input{introduction}
\input{motivation}
% \input{action}
\input{types}
\input{processes}
\input{metatheory}
%\input{algorithmic} % distribute between expressions and processes
%\input{cfst-vs-algst}
\input{implementation}
\input{related}

\input{conclusion}

\section*{Acknowledgements}
We thank the reviewers for their suggestions that greatly
contributed to a more solid version of the paper. Support for this research was provided by the
Fundação para a Ciência e a Tecnologia through project SafeSessions, ref.\
PTDC/CCI-COM/6453/2020, and the LASIGE Research Unit, ref.\ UIDB/00408/2020 and
ref.\ UIDP/00408/2020, and by the COST Action EuroProofNet (CA20111).

\input{artifact}

% \input{questions}
\bibliographystyle{plainnat}
\bibliography{biblio}

\newpage
\appendix
\input{appendix-faq}
\input{appendix-motivation}
\input{appendix-types}
\input{appendix-procs}
\input{appendix-cfst}
\end{document}

%% file: macros.tex
\newcommand{\kindc}[1]{{\color{ForestGreen}{#1}}} % c for colour
\newcommand{\typec}[1]{{\color{RoyalBlue}{#1}}}
\newcommand{\termc}[1]{{\color{RedOrange}{#1}}}
\newcommand{\blk}[1]{{\color{black}{#1}}}
\newcommand{\gray}[1]{{\color{gray}{#1}}}

\newcommand\algst{AlgST\xspace}
\newcommand\fuse{FuSe$^{\{\}}$\xspace}

\newcommand{\grmeq}{\; ::= \;\;}
\newcommand{\grmor}{\;\mid\;}

\newcommand{\syntax}{\texttt}

\newcommand{\kindstyle}[1]{\kindc{\textbf{\textsc{#1}}}}
\newcommand{\kindm}{\kindstyle m}
\newcommand{\kindp}{\kindstyle p}
\newcommand{\kinds}{\kindstyle s}

\newcommand{\kindt}{\kindstyle t}
\newcommand{\kind}{\kindc{\kappa}}

\newcommand{\kinda}[2]{\kindc{\overline{#1}}\rightarrow\kindc{#2}}
\newcommand{\kindas}[2]{\kindc{#1}\rightarrow\kindc{#2}} % simple, no overline

\newcommand\lin{\kindstyle{lin}}
\newcommand\un{\kindstyle{un}}

\newcommand{\kindtu}{\kindt^\un}
\newcommand{\kindtl}{\kindt^\lin}

\newcommand{\operator}[1]{\blk{\operatorname{#1}}}

\newcommand{\polarizedparam}[2]{\blk{#1}\blk(\typec{#2}\blk)}
\newcommand{\tosessionop}{\mathsection}
\newcommand{\tosession}[2]{\blk{\tosessionop(}\typec{#1}\blk{).}\typec{#2}} % imaginative syntax
\newcommand\pneg[1]{\typec{-#1}} % was {#1^-}

\newcommand{\types}[1]{\typec{\syntax{#1}}} % s for syntax
\newcommand\TEndW{\ensuremath{\types{End}_{\typec ?}}} % wait
\newcommand\TEndT{\ensuremath{\types{End}_{\typec !}}} % terminate
\newcommand{\TG}[2]{\typec{{#1}.{#2}}} % DEPRECATED, no guards anymore
\newcommand{\TIn}[2]{\TG{?#1}{#2}}
\newcommand{\TOut}[2]{\TG{!#1}{#2}}
\newcommand{\TGIn}[1]{\typec{?#1}}
\newcommand{\TGOut}[1]{\typec{!#1}}
\newcommand{\TSkip}{\types{Skip}}
\newcommand{\TSeq}[2]{\typec{#1;#2}}
\newcommand{\TChoice}[2]{\typec{#1\oplus#2}}
\newcommand{\TBranch}[2]{\typec{#1\mathbin{\&} #2}}
\newcommand{\TOne}{\types{1}}
\newcommand\TFIn[2]{\tosession{\polarizedparam - {#1}}{#2}}
\newcommand\TFOut[2]{\tosession{\polarizedparam + {#1}}{#2}}
\newcommand\TFMinus[1]{\blk{-(}\typec{#1}\blk)}
\newcommand\protocolk{\types{protocol}}
\newcommand\datak{\types{data}}

\newcommand{\TUnit}{\types{Unit}}
\newcommand{\TInt}{\types{Int}}

\newcommand\TFun[3][\mult]{\typec{#2 \to_{#1} #3}}
\newcommand{\TPair}[2]{\typec{#1 \otimes #2}}

\newcommand{\dualk}{\types{Dual}}
\newcommand{\TDual}[1]{\typec{\dualk~{#1}}}
\newcommand{\TMinus}[1]{\typec{-{#1}}}
\newcommand{\tbind}[2]{\typec{{#1}\blk\colon\kindc{#2}}}
\newcommand{\TForall}[3]{\typec{\forall\tbind{#1}{#2}{.#3}}}
\newcommand{\TPApp}[2]{\typec{{#1}~{#2}}}
\newcommand{\matchin}[1][i]{\TFIn   {\overline{T_{#1}}\tsubs{\overline U}{\overline\alpha}}{S}}
\newcommand{\matchout}[1][i]{\TFOut {\overline{T_{#1}}\tsubs{\overline U}{\overline\alpha}}{S'}}

\newcommand{\protocolsfull}{\protocolk\ \typec{\proto\ \overline{\alpha}}~\typec{=}~\typec{\{{\termc{\constr_i}\ \overline{T_i}}\}_{i\in I}}
}

\newcommand{\exps}[1]{\termc{\syntax{#1}}}

\newcommand{\sendk}{\exps{send}}
\newcommand{\receivek}{\exps{receive}}

\newcommand{\letk}{\exps{let}}
\newcommand{\ink}{\exps{in}}
\newcommand{\newk}{\exps{new}}
\newcommand{\reck}{\exps{rec}}
\newcommand{\forkk}{\exps{fork}}
\newcommand{\matchk}{\exps{match}}
\newcommand{\withk}{\exps{with}}
\newcommand{\terminatek}{\exps{terminate}}
\newcommand{\waitk}{\exps{wait}}
\newcommand{\unitk}{\exps{*}} % {\exps{unit}}
\newcommand{\selectk}{\exps{select}}

\newcommand{\constr}{\termc C} % constructor
\newcommand{\proto}{\typec{\ensuremath{\rho}}}%{\typec P} % protocol
\newcommand{\abse}[3]{\termc{\lambda{#1}\colon{\typec #2}.{#3}}}
\newcommand{\absne}[3]{\termc{\lambda{#1}.{#3}}}
\newcommand{\tabse}[3]{\termc{\Lambda}\,\typec{{#1}\colon{#2}}\termc{. #3}}
\newcommand{\rece}[3]{\termc{\reck\ {#1}\colon{\typec #2}.{#3}}}
\newcommand{\appe}[2]{\termc{#1\,#2}}
\newcommand{\letu}[2]{\termc{\letk\,\unitk = #1\,\ink\,#2}}
\newcommand{\lete}[4]{\termc{\letk\,\langle #1, #2\rangle = #3\,\ink\,{#4}}}

\newcommand{\paire}[2]{\termc{\langle #1, #2\rangle}}

\newcommand{\matche}[2]{\matchk\;\termc #1\;\withk\;\termc{#2}}
\newcommand{\matchwithe}[6]{\termc{\matche{#1}{\{{#2}_{#5}\;{#3}_{#5}\rightarrow{#4}_{#5}\}_{{#5}\in{#6}}}}}
\newcommand{\selecte}[1]{\selectk\,\termc{#1}} % Use C_i alone

\newcommand{\closee}[1]{\terminatek\;\termc{#1}}
\newcommand{\waite}[1]{\waitk\;\termc{#1}}
\newcommand{\tappe}[2]{\termc{#1[}\typec{#2}\termc{]}}

\newcommand{\ebind}[2]{\termc{{#1}\blk\colon\typec{#2}}}
\newcommand{\eubind}[2]{\termc{{#1}\blk{\colon^{\!\!\!\star}}\typec{#2}}} % unrestricted

\newcommand\Normal[2][\pm]{\operator{nrm}^{\blk #1}\blk(\typec{#2}\blk)}

\newcommand\FVT[1]{\operator{fv} ({#1})} % in types
\newcommand\FVE[1]{\FVT{#1}} % in terms (exps and procs)
\newcommand\FVL[1]{\FVT{#1}} % in labels

\newcommand{\expp}[1]{\termc{\langle{#1}\rangle}}
\newcommand{\PAR}{\mid}
\newcommand{\parp}[2]{\termc{#1 \PAR #2}}
\newcommand{\newp}[3]{\termc{(\nu{#1}{#2})\,{#3}}}

\newcommand{\lts}[1]{\xrightarrow{#1}}
\newcommand{\trans}[3]{\termc{#1}\lts{#2}\termc{#3}}
\newcommand{\transCtx}[3]{{#1}\lts{#2}{#3}}

\newcommand{\forkl}[1]{\syntax{fork}\,{#1}}
\newcommand{\sendl}[2]{{#1}!{#2}}
\newcommand{\receivel}[2]{{#1}?{#2}}
\newcommand{\openl}[1]{{#1}!}
\newcommand{\closel}[1]{{#1?}}
\newcommand{\newl}[3]{(\nu{#1}{#2})\colon{#3}}
\newcommand{\scopel}[4]{(\nu{#1}{#2}){#3}!{#4}}
\newcommand{\parl}[2]{{#1}\PAR{#2}}

\newcommand{\cfsttoalgst}[1]{\blk{\llbracket {#1} \rrbracket}}
\newcommand{\cfsttoprot}[1]{\blk{\llparenthesis {#1} \rrparenthesis}}
\newcommand{\algsttocfst}[1]{\blk{\Lbag {#1} \Rbag}}
\newcommand{\emptyseq}[0]{\varepsilon}
\newcommand{\iso}[1]{\gray{{#1}}}
\newcommand{\isodecl}[1]{\iso{#1}\gray{\colon}}

\newcommand\isConv[4][\Delta]{#1 \vdash \typec{#2} \equiv \typec{#3} : {\kindc
    {#4}}}
\newcommand{\isEquiv}[2]{\typec{#1}\equiv_A\typec{#2}} % Abbreviated version

\newcommand\isProc[2][\Gamma]{{#1} \vdash \termc{#2}}

\newcommand\isCtx[3][\Delta]{{#1} \vdash {#2}  \Leftarrow {#3}}
\newcommand\twoCtx[2]{#1 \mid #2}
\newcommand\typeSynth[3][\Delta]{#1 \vdash {\typec{#2}} \Rightarrow {#3}}
\newcommand\typeAgainst[3][\Delta]{#1 \vdash {\typec{#2}} \Leftarrow {#3}}

\newcommand\expSynth[5][\Delta]{\twoCtx{#1}{#2} \vdash \termc{#3} \Rightarrow
  \twoCtx{\typec{#4}}{#5}}
\newcommand\expAgainst[5][\Delta]{\twoCtx{#1}{#2} \vdash {\termc{#3}} \Leftarrow
  \twoCtx{\typec{#4}}{#5}}

\newcommand{\tsubs}[2]{\blk[\typec{#1}\blk/\typec{#2}\blk]} % term substitution
\newcommand{\esubs}[2]{\blk[\termc{#1}\blk/\termc{#2}\blk]} % expression substitution

\newcommand{\declrel}[2]{\emph{#1}\hfill\fbox{{#2}}}
\newcommand{\Empty}{\cdot}

\newcommand\isSubKind[2]{\kindc{#1} \le \kindc{#2}}
\newcommand{\typeofop}{\operator{typeof}}
\newcommand\typeof[1]{\typeofop\operator{(}\termc{#1}\blk{)}}

\newcommand{\infrule}[3]{\inferrule[#1]{#2}{#3}}
\newcommand{\axiom}[2]{\infrule{#1}{}{#2}}

\newcommand{\rulename}[2]{\textsc{{#1}-{#2}}\xspace}

\newcommand{\rulenametype}[1]{\rulename{T}{#1}}  % T- for type formation
\newcommand{\rulenametconv}[1]{\rulename{C}{#1}}  % C- for type conversion
\newcommand{\rulenameexp}[1]{\rulename{E}{#1}}   % E- for expressions

\newcommand{\rulenametypeEndW}{\rulenametype{End?}}
\newcommand{\rulenametypeEndT}{\rulenametype{End!}}

\newcommand{\rulenametypeDotIn}{\rulenametype{In}}
\newcommand{\rulenametypeDotOut}{\rulenametype{Out}}

\newcommand{\rulenametypeMsgNeg}{\rulenametype{MsgNeg}}

\newcommand{\rulenametypeDualS}{\rulenametype{Dual}}
\newcommand{\rulenametypeDual}{\rulenametype{Dual}}

\newcommand{\rulenametypeUnit}{\rulenametype{Unit}}

\newcommand{\rulenametypeArrow}{\rulenametype{Arrow}}

\newcommand{\rulenametypeProtocol}{\rulenametype{Protocol}}
\newcommand{\rulenametypeVar}{\rulenametype{Var}}
\newcommand{\rulenametypePoly}{\rulenametype{Poly}}
\newcommand{\rulenametypeSub}{\rulenametype{Sub}}
\newcommand{\rulenametypePair}{\rulenametype{Pair}}

\newcommand{\rulenametconvReflex}{\rulenametconv{Refl}}
\newcommand{\rulenametconvSymm}{\rulenametconv{Sym}}
\newcommand{\rulenametconvTrans}{\rulenametconv{Trans}}
\newcommand{\rulenametconvSub}{\rulenametconv{Sub}}

\newcommand{\rulenametconvDoubleS}{\rulenametconv{DualInv}}

\newcommand{\rulenametconvDualEndW}{\rulenametconv{DualEnd?}}
\newcommand{\rulenametconvDualEndT}{\rulenametconv{DualEnd!}}

\newcommand{\rulenametconvDualOutM}{\rulenametconv{DualOut}}

\newcommand{\rulenametconvDualInM}{\rulenametconv{DualIn}}

\newcommand{\rulenametconvNegOut}{\rulenametconv{NegOut}}
\newcommand{\rulenametconvNegIn}{\rulenametconv{NegIn}}

\newcommand{\rulenametconvNegInv}{\rulenametconv{NegInv}}

\newcommand{\rulenameexpConst}{\rulenameexp{Const}}
\newcommand{\rulenameexpUVar}{\rulenameexp{Var$\star$}}
\newcommand{\rulenameexpVar}{\rulenameexp{Var}}
\newcommand{\rulenameexpAbs}{\rulenameexp{Abs}}
\newcommand{\rulenameexpAbsn}{\rulenameexp{Abs'}}
\newcommand{\rulenameexpRec}{\rulenameexp{Rec}}
\newcommand{\rulenameexpApp}{\rulenameexp{App}}
\newcommand{\rulenameexpAppn}{\rulenameexp{App'}}
\newcommand{\rulenameexpTAbs}{\rulenameexp{TAbs}}
\newcommand{\rulenameexpTApp}{\rulenameexp{TApp}}

\newcommand{\rulenameexpPair}{\rulenameexp{Pair}}
\newcommand{\rulenameexpLetUnit}{\rulenameexp{Let*}}
\newcommand{\rulenameexpLet}{\rulenameexp{Let}}

\newcommand{\rulenameexpMatch}{\rulenameexp{Match}}
\newcommand{\rulenameexpCheck}{\rulenameexp{Check}}

\newcommand{\cf}{cf.\xspace}

\mprset{lineskip=0ex}

%% file: abstract.tex
\begin{abstract}
  % Current session type systems rely almost exclusively on structural
  % typing with equirecursion to model protocols. The presence of
  % equirecursion requires superlinear algorithms for
  % checking type equivalence. For traditional recursive session types, type
  % equivalence can be reduced to the equivalence of finite automata. For
  % context-free and nested session types, it amounts to the equivalence of
  % deterministic pushdown automata.

  % We propose algebraic protocols that enable the definition of protocol templates and
  % session types analogous to the definition of
  % domain-specific types with algebraic datatypes. 
  % % We propose algebraic protocols that adapt the well-established concept of
  % % algebraic datatypes to the definition of protocol templates and
  % % session types.
  % % This shift of perspective replaces expensive algorithms for checking
  % % the equivalence of session types with a simple linear-time
  % % check.
  % Parameterized algebraic protocols subsume all regular as well as most
  % context-free and nested session types 
  % and, at the same time, replace the expensive superlinear algorithms
  % for type checking
  % % structural type equivalence of equirecursive session types
  % by a nominal check that runs in linear time. 
  % % 
  % Algebraic protocols in combination with polymorphism increase
  % expressiveness and modularity by facilitating new ways
  % of parameterizing and composing session types.
  We propose algebraic protocols that enable the definition of 
  protocol templates and session types analogous to the definition 
  of domain-specific types with algebraic datatypes. Parameterized 
  algebraic protocols subsume all regular as well as most 
  context-free and nested session types and, at the same time, 
  replace the expensive superlinear algorithms for type checking by 
  a nominal check that runs in linear time. Algebraic protocols in 
  combination with polymorphism increase expressiveness and modularity 
  by facilitating new ways of parameterizing and composing session types.
\end{abstract}

%%% Local Variables:
%%% mode: latex
%%% TeX-master: "main"
%%% End:

%% file: introduction.tex
\section{Introduction}
\label{sec:introduction}

The advantages of modeling data with algebraic datatypes (ADTs) 
have been recognised since the 80s 
\cite{Burstall1977a,DBLP:conf/lfp/BurstallMS80}. 
% ADTs are easy to define, 
% promote abstraction, and are less prone to errors.
ADTs are easy to define and use, they enable the construction of
domain-specific types, and are amenable to efficient nominal type checking.
Functional programming languages quickly embraced ADTs along
with the concept of pattern matching which is still gaining momentum.  

Algebraic datatypes integrate recursive types, sum types, product
types, and generativity (each ADT definition creates a new type) in a
single declaration. The smooth integration of different typing concepts
makes ADTs easy to learn and use. It also simplifies the 
implementation of type checkers because they can elide superlinear algorithms for checking
equivalence of equirecursive types (where recursive types are
considered equal to their unfolding) in favor of a linear comparison of nominal types.
%  and efficient to handle in a language
% implementation. ADTs are popular among programmers as they have to learn just a single mechanism
% for type definition.

We show that the ADT approach for defining recursive
datatypes can be carried over to 
\emph{algebraic protocols} (AP), a novel method for defining a large class
of recursive protocols. Analogous to ADTs, APs integrate 
concepts from the theory of session types
\cite{DBLP:conf/concur/Honda93,DBLP:conf/parle/TakeuchiHK94}
(recursive types, choice types, and sequence types) and generativity
in a single declaration. Analogous to ADTs, APs are easy to use and
are amenable to efficient type checking. 

We develop APs by analogy to ADTs. Consider the following ADT for lists of
integers.
\begin{lstlisting}
  data IntList = Nil | Cons Int IntList
\end{lstlisting}
The datatype \lstinline+IntList+  comes with two data constructors,
\lstinline+Nil+  to construct an 
empty list and \lstinline+Cons+ to construct a list from an element
of type \lstinline+Int+ and a list of type \lstinline+IntList+. The
alternative between the constructors corresponds to a sum type and the
two arguments of \lstinline+Cons+ give rise to a product
type. Clearly, \lstinline+IntList+ is a recursive type. Moreover, its
definition introduces \lstinline+IntList+ as a new type, different
from any other. Type equivalence is simple: we just compare ADTs
by name.

Although \lstinline+IntList+ specifies a (static) data type, 
it carries an implicit dynamics and it should be 
quite natural for the reader to imagine elements of this type as
finite sequences of messages between two communication partners: if
there are  further elements, the sender selects the  
constructor \lstinline+Cons+, sends a value of type 
\lstinline+Int+ and then recurses. If there are no further elements,
the sender selects \lstinline+Nil+ and stops the transmission. 
On the other end, the receiver awaits a constructor. On receiving a
\lstinline+Cons+ it receives an \lstinline+Int+ and recurses. It stops
on receiving a \lstinline+Nil+. 

% perhaps, we could rather \emph{receive} a value of type \lstinline+a+ 
% and then receive the rest of the stream. In this paper we 
% interpret algebraic datatypes as protocols and reformulate 
% the theory of session types, now presented as a natural 
% extension of algebraic datatypes -- materialised in our
% theory of \emph{algebraic session types}. 

This view has striking parallels to the theory of binary session
types~\cite{DBLP:conf/concur/Honda93,DBLP:conf/esop/HondaVK98,DBLP:conf/parle/TakeuchiHK94}.
Session types provide a rich type discipline for communication channels that
statically guarantees protocol fidelity as well as type soundness across
communication channels. Here is a session type for the protocol associated to
the \lstinline+IntList+ ADT, written in the syntax of existing
systems~\cite{DBLP:journals/corr/abs-2106-06658,lindley17:_light_funct_session_types,DBLP:journals/jfp/GayV10}
and as seen from the sender's perspective:\footnote{Writing
  \lstinline+IntListS+ for the session type corresponding to
  \lstinline+IntList+.}
\begin{lstlisting}
  type IntListS = mualpha. oplus{Cons: !Int.alpha, Nil: EndT}
\end{lstlisting}
The leading $\mu\alpha\dots$ constructs an equirecursive type. The body is
an internal choice type $\oplus$ with two alternatives indicated by the tags
\lstinline+Cons+ and \lstinline+Nil+. If the sender selects
\lstinline+Cons+, it must continue according to
\lstinline+!Int.alpha+, which means it must send an integer and then
recurse. If the sender selects \lstinline+Nil+, then the protocol
comes to an end as indicated by the type \lstinline+EndT+.
Type \lstinline+EndT+ represents a channel endpoint ready to be closed. The dual
type, \lstinline+EndW+, represents the other endpoint, waiting to be closed.
As the name
suggests, a session type describes a full session of interactions
between the communication partners up to the end.

Like this example, most session type systems 
rely on structural typing and equirecursion; 
this way, unfolding does not give rise to communication.
Traditional session type systems restrict 
recursion to the tail position (as in \lstinline/IntListS/), so that types are effectively finite
automata with regular trace languages. Consequently, type equivalence
and subtyping for such systems are decidable in polynomial time.  Recently, richer forms of recursion
have been considered that lift the restriction to tail recursion, so
that types are akin to, e.g.,  deterministic context-free grammars
\cite{DBLP:conf/icfp/ThiemannV16,DBLP:conf/esop/Padovani17,DBLP:journals/toplas/Padovani19,DBLP:journals/corr/abs-2106-06658,
  DBLP:conf/esop/DasDMP21}.  
Unfortunately, when types are endowed with richer forms 
of recursion, type equivalence gets more complex 
(doubly-exponential \cite{DBLP:conf/esop/DasDMP21}). 
The implementation of a type equivalence algorithm for these systems is nontrivial
\cite{DBLP:conf/tacas/AlmeidaMV20}; incomplete algorithms have been
used to speed up type equivalence \cite{DBLP:conf/esop/DasDMP21};
annotations have been proposed~\cite{DBLP:journals/toplas/Padovani19}
that make type checking effective. But the algorithms for
type equivalence remain asymptotically superlinear in all these works.
% type checking by exploiting programmer annotations and
% showed that subtyping is undecidable. 

The sad truth is that structural session types are quite verbose and
too complex for the average programmer, thus hampering their adoption. While equirecursive types are
expressive, understanding their equivalence is hard for programmers
and inefficient in language implementations.

What if programming with expressive recursive protocols were as easy as
programming with ADTs? The theory and practice of algebraic protocols
that we propose in this paper delivers exactly such an
approach. Enhanced with parameters and polymorphism in the style of
System~F, APs provide a uniquely simple and efficient approach to modularity and abstraction
when defining and implementing communication protocols.

%Continuing the \lstinline+IntList+ example,
Here is the declaration of an algebraic
protocol corresponding to the session type \lstinline+IntListS+:\footnote{Writing
  \lstinline+IntListP+ for the protocol type corresponding to
  \lstinline+IntList+.}
\begin{lstlisting}
  protocol IntListP = Nil | Cons Int IntListP
\end{lstlisting}
% In this paper, we propose a paradigm shift. We revisit the 
% guidelines that constituted the basis of the functional languages 
% that we use today 
% \cite{Burstall1977a,DBLP:conf/lfp/BurstallMS80}
% and reformulate the theory of session types 
% by favouring modularity, code abstraction, the extensive use of 
% user-defined types and overloaded operators.
% Algebraic session types now enable the definition of protocol 
% templates that are \emph{not} committed to a direction (of communication)
% until they are \emph{materialized}. The protocol template for 
% streams reads like:
% \begin{lstlisting}
%     protocol StreamP a = Cons a (StreamP a)
% \end{lstlisting}
The keyword \lstinline+protocol+, in place of \lstinline+data+, denotes an
algebraic \emph{protocol} type rather than an algebraic \emph{data} type. A
protocol type is like a template that describes a composable segment of a
session. It neither describes a full session, nor is it committed to a
particular direction of communication. Rather, a protocol type has to be
\emph{materialized} into a session type by assigning a direction and specifying
a continuation session: we write \lstinline+!IntListP.S+ to send a sequence of
integers and continue with session type \lstinline+S+ or \lstinline+?IntListP.S+
to receive a sequence of integers and continue with \lstinline+S+, so that the
session type \lstinline/IntListS/ corresponds to \lstinline/!IntListP.EndT/.
Hence, protocol types relate to session types just like composable continuations
relate to continuations \cite{DBLP:conf/popl/Felleisen88,DanvyFilinski1992}: the
former describe a segment of interaction (execution) while the latter describe
interaction (execution) to the end.

Just like ADTs, we can equip APs with parameters as in
\lstinline+ListP a+:
\begin{lstlisting}
  protocol ListP a = Nil | Cons a ListP
\end{lstlisting}
While parameters range over types in ADT definitions, they
range over \emph{protocols} in AP definitions! This facility is key to
the modularity of APs. For example, the protocol
\lstinline+ListP Q+ repeats the protocol \lstinline+Q+. It is
also possible to flip the direction of \lstinline+Q+ by considering the protocol
\lstinline+ListP -Q+: repeat the protocol \lstinline+Q+ but with reversed roles
of sender and receiver (more in \cref{sec:simple-protocols-add}). With this
modular design, we can define finite streams of
anything expressible as an algebraic protocol. The
external and internal choices in previous session type systems are
captured by the sum-of-product nature of algebraic types---the sum
type replaced by internal/external choice and the product type
replaced by sequencing as in the linear logic
interpretation of session types \cite{DBLP:journals/mscs/CairesPT16}.

Algebraic protocol types are nominal and isorecursive. 
This combination of features reduces the complexity of type equivalence from
asymptotically doubly-exponential \cite{DBLP:conf/tacas/AlmeidaMV20,DBLP:conf/esop/DasDMP21} 
to linear. Notably, the expressivity of algebraic protocols 
is on par with that of session types with non-regular recursion
\cite{DBLP:conf/icfp/ThiemannV16,DBLP:conf/esop/DasDMP21}, 
so all protocols expressible in such theories can be adapted 
to algebraic protocols (without paying the price of hosting complex or 
incomplete algorithms in the compiler). Interestingly, types 
are now also iso-associative, thus saving us from defining 
hard-coded associativity laws, due to the basic properties 
of sequential composition 
\cite{DBLP:conf/icfp/ThiemannV16,DBLP:journals/corr/abs-2106-06658}.

\subsection*{Contributions}
\label{sec:contributions}

\begin{itemize}
\item Algebraic protocols (AP). We propose a new approach for defining
  recursive protocols compositionally inspired by parameterized algebraic
  datatypes. This approach combines naming, recursion, and
  choice-of-sequence types analogously to algebraic
  datatypes. It adds a novel notion of direction to the component
  protocols of a tag. Checking type equivalence becomes simple and efficient
  (worst case linear in the size of the types) as APs are nominal. 
\item Algebraic session types (\algst). We define a core calculus that is based
  on linear System~F with (non-linear) recursion and features session types
  based on APs. Type checking for \algst is specified by bidirectional typing
  rules that have a direct algorithmic reading.
  % 
  % The operational semantics is presented as a labelled transition relation,
  % inspired by Fowler \etal and Montesi and
  % Peressotti~\cite{DBLP:conf/concur/0001KDLM21,DBLP:journals/corr/abs-2105-08996,DBLP:journals/corr/abs-2106-11818}
  % for double binders and the open/close scheme for the free output of the
  % $\pi$-calculus. %~\cite{DBLP:journals/iandc/MilnerPW92b}.
  We prove type soundness via preservation and
  progress.% , where our approach offers a
  % rather elegant statement (\cref{thm:progress} on \cpageref{thm:progress}).
\item Full implementation as a publicly available artifact
  \cite{artifact}, including practical extensions (cf.\ \cref{sec:implementation}).
\end{itemize}

\subsection*{Overview}
\label{sec:overview}

\Cref{sec:algebr-sess-types} contains an informal introduction of algebraic
protocols with examples; we illustrate the gain in
expressivity and modularity with respect to other work in
\cref{sec:param-prot}. \Cref{sec:types} discusses the type structure
of \algst and establishes key results for efficient type checking: correctness of type normalization and
worst case complexity of type equivalence. \Cref{sec:processes} defines
statics and dynamics of expressions 
and processes % . \Cref{sec:metatheory}
followed by a presentation of the metatheoretical results.
% \Cref{sec:cfst-vs-algst} gives an informal outline of
% an embedding of context-free session types to \algst.
\Cref{sec:implementation} discusses the implementation; since all rules are
algorithmic, their appropriate reading leads directly to a type checker and an interpreter.
\Cref{sec:related,sec:conclusion} discuss related work and conclude.
Proofs, auxiliary results, further examples, an informal outline of an embedding
of context-free session types to \algst, and further discussion are provided as supplementary material. The
implementation is publicly available \cite{artifact}.

%%% Local Variables:
%%% mode: latex
%%% TeX-master: "main"
%%% End:

%% file: motivation.tex
\section{Algebraic Protocols and session types}
\label{sec:algebr-sess-types}

This section provides a programmer's view of algebraic session types (\algst). It is driven by examples and
introduces the concepts of algebraic protocols and session types. Our
syntax is inspired by Haskell, as implemented in our artifact. Examples make
liberal use of practical extensions of the formal system (cf.\ \cref{sec:implementation}).
%
% follows.
% \begin{lstlisting}
% |> : forall(a:T).forall(b:T).a -> (a -> b) -> b
% x |> f = f x
% \end{lstlisting}

% The design of \algst is inspired by two assumptions about communication in distributed
% systems.
% \begin{itemize}
% \item The creation of a connection is expensive. It is desirable to limit the number of connections
%   and to reuse existing connections as much as possible.
% \item  The built-in  communication operations are limited to primitive types, the serialization of which is
%   hardcoded. Communication of any other type has to be encoded as a protocol.
% \end{itemize}
% We finish this section discussing some frequently asked questions in subsection~\ref{sec:freq-asked-quest}.

\subsection{Algebraic datatypes as protocols}
\label{sec:algebr-datatyp-as}

Suppose we want to transmit abstract syntax trees for a language with integers
and addition. Their type might be defined by an algebraic datatype definition as
in Haskell or ML.
\begin{lstlisting}
data Ast = Con Int | Add Ast Ast
\end{lstlisting}
This declaration defines the type \lstinline+Ast+ along with the types of the constructor
functions, which can be used to construct \lstinline+Ast+ values as well as for pattern matching
on them:
\begin{lstlisting}
  Con :        Int -> Ast
  Add : Ast -> Ast -> Ast
\end{lstlisting}
The type \lstinline+Ast+ is recursive and its \lstinline+Add+
constructor takes two \lstinline+Ast+ values as arguments.
To define a protocol for transmitting values of type \lstinline+Ast+, we merely change the keyword
\lstinline+data+ to \lstinline+protocol+ as in
\begin{lstlisting}
protocol AstP = ConP Int | AddP AstP AstP
\end{lstlisting}
This declaration defines the protocol type \lstinline+AstP+ along with tags \lstinline+ConP+
and \lstinline+AddP+, which serve to introduce and eliminate branching in the protocol.
% Tags are encoded by small numbers for transmission.
The resulting \lstinline+AstP+ protocol can be used in
two directions, \lstinline+!AstP+ for sending and \lstinline+?AstP+ for receiving
data.\footnote{The implementation allows us to reuse a
  \lstinline+data+ declaration as a \lstinline+protocol+ declaration,
  thus overloading the constructors.} 
Protocol types only describe behavior, they do not classify run-time values. They inhabit a
dedicated kind \lstinline+KP+.

We illustrate the \lstinline+AstP+ protocol with a function to send an
\lstinline+Ast+ value. We write functions
in \emph{channel passing style}, which takes a channel as its argument
and returns it after processing some prefix of the channel's session
type. The function's type is universally quantified over the continuation
type \lstinline+s+ of the unexplored part of the channel.
% This quantification must be introduced explicitly using
% \lstinline+forall(s:S)+ where \lstinline+KS+ is the kind of session
% types.
The argument
type adds some prefix to \lstinline+s+ that describes the part of the
session performed by the function. In our example the prefix is \lstinline+!AstP+.
\begin{lstlisting}
sendAst : Ast -> forall(s:S). !AstP.s -> s
sendAst t [s] c = case t of {
    Con x   -> select ConP [s] c |> sendInt [s] x,
    Add l r -> select AddP [s] c |> sendAst l [!AstP.s] |> sendAst r [s] }
\end{lstlisting}

We use the \lstinline+|>+ operator for reverse function
application. Defined by \lstinline+x |> f = f x+, the operator associates
to the left and features low priority, so that \lstinline+x |> f |> g+ means \lstinline+g (f x)+.
The square brackets indicate type abstraction (on the left) and type
application (in expressions).
% By abstracting over a type variable \lstinline+s+ of kind \lstinline+S+
% (introduced by \lstinline+forall(s:S)+), we require \lstinline+s+
% to be later replaced by a session type.
%
The predefined \lstinline+sendInt+ operation has type
\lstinline+forall(s:S). Int -> !Int.s -> s+.
The use of pattern matching with  \lstinline+case+ on algebraic
datatypes is standard.
As usual in session type systems, channel values of session type are linear
and the \lstinline+select+ operation makes an internal choice between alternatives. Selecting a protocol
constructor sends an encoding of the corresponding tag and ``pushes'' the transmission of the
fields on the channel type. That is:
\begin{lstlisting}
  select ConP [s] : !AstP.s -> !Int.s
  select AddP [s] : !AstP.s -> !AstP.!AstP.s
\end{lstlisting}
A function to process the receiver end of this protocol must take an argument of type \lstinline+?AstP.s+.
\begin{lstlisting}
recvAst : forall(s:S). ?AstP.s -> (Ast, s)
recvAst [s] c = match c with {
    ConP c -> let (x, c)  = receiveInt [s] c    in (Con x, c),
    AddP c -> let (tl, c) = recvAst [?AstP.s] c in
              let (tr, c) = recvAst [s] c       in (Add tl tr, c) }
\end{lstlisting}
The \lstinline+match+ on a session-typed channel receives
a tag and branches according to it. Branching also ``pushes'' transmission of the fields on the
channel type analogous to the \lstinline+select+ operation:
\begin{lstlisting}
c : ?AstP.s |- match c with { ConP c -> ...   -- c : ?Int.s
                              AddP c -> ... } -- c : ?AstP.?AstP.s
\end{lstlisting}
Left of the turnstile $\vdash$, we write the typing of
\lstinline+c+ before the \matchk. In the body of the \matchk, we write
the typing of \lstinline+c+ after matching against the selectors \lstinline+ConP+ and \lstinline+AddP+,
respectively.
The \lstinline+match c+ (with the type on the left) consumes the linear
binding for channel \lstinline+c+. The pattern \lstinline+ConP c+
reintroduces variable \lstinline+c+, binds it with the channel
after processing the selection and gives it an updated type (on the right) for the body of the
corresponding branch. The consideration for pattern \lstinline+AddP c+
is analogous. The choice of the branch is external and depends
on the received selector tag.

The protocol in this subsection goes beyond the reach of most session type systems because
they only support tail recursive session types.
Thus, functions like \lstinline+sendAst+ and \lstinline+recvAst+ have
no meaningful type in traditional
session type systems \cite{DBLP:journals/jfp/GayV10,lindley17:_light_funct_session_types}.
However, they can be written using context-free or nested session
types \cite{DBLP:journals/corr/abs-2106-06658,DBLP:conf/icfp/ThiemannV16,DBLP:conf/esop/DasDMP21}. We
discuss the differences in \cref{sec:related}.

\subsection{Simple protocols with polarities}
\label{sec:simple-protocols-add}

The protocol type \lstinline+AstP+ from
\cref{sec:algebr-datatyp-as}  describes unidirectional
communication. This section extends protocol types with
\emph{polarities} to model bidirectional communication.
The protocol type \lstinline+Arith+ of an arithmetic server (cf.\ 
\cite{DBLP:conf/esop/GayH99}) serves as an example. The server takes a command, \lstinline+Neg+ or
\lstinline+Add+, accepts one or two inputs to the commands and produces its output.
\begin{lstlisting}
protocol Arith = Neg Int -Int | Add Int Int -Int
\end{lstlisting}
The new element in this declaration is the \emph{polarity} ``\lstinline+-+'' which indicates a
reversal of the direction of communication. A client would use this protocol in ``forward
direction'', that is, in a session type of the form \lstinline+!Arith.s+ whereas the server would
implement \lstinline+?Arith.t+, for some session type \lstinline+t+ dual to  \lstinline+s+. On \lstinline+!Arith.s+,
selection of a tag works as in \cref{sec:algebr-datatyp-as}, but negative polarities
flip the direction:
\begin{lstlisting}
  select Neg [s] : !Arith.s -> !Int.?Int.s
  select Add [s] : !Arith.s -> !Int.!Int.?Int.s
\end{lstlisting}
Branching operates analogously on types in the assumption:
\begin{lstlisting}
c : ?Arith.s |- match c with { Neg c -> ...   -- c : ?Int.!Int.s
                               Add c -> ... } -- c : ?Int.?Int.!Int.s
\end{lstlisting}
The implementation of the server is straightforward. (We leave the client code to the reader.)
\begin{lstlisting}
serveArith : forall(s:S). ?Arith.s -> s
serveArith [s] c = match c with {
    Neg c -> let (x, c) = receiveInt [!Int.s] c in
             sendInt [s] (0-x) c,
    Add c -> let (x, c) = receiveInt [?Int.!Int.s] c in
             let (y, c) = receiveInt [!Int.s] c in
             sendInt [s] (x+y) c }
\end{lstlisting}

% This is a good point to reflect on the meaning of the arguments to the
% tags in a protocol declaration.
The tag declaration \lstinline+Neg Int -Int+ shows that the tag arguments are protocols themselves, rather than
plain types. This fact is witnessed by the use of \lstinline+AstP+ in a tag argument for
\lstinline+AstP+ as well as by the use of \lstinline+-Int+ in a tag argument for
\lstinline+Arith+.
Ordinary types like \lstinline+Int+ are promoted to protocol types.

\subsection{Parameterized protocols and modularity}
\label{sec:param-prot}

Algebraic protocols can have protocol parameters,  just like algebraic data
types have type parameters. This facility provides a powerful and versatile approach for declaring
modular protocols.
As a first example, we declare a protocol that repeatedly runs
another protocol:
\begin{lstlisting}
protocol Stream a = Next a (Stream a)
\end{lstlisting}
Using this protocol we define a process that sends infinitely many ones on a channel:
\begin{lstlisting}
ones : !Stream Int.EndT -> Unit
ones c = select Next [Int,EndT] c |> sendInt [!Stream Int.EndT] 1 |> ones
\end{lstlisting}
When using the \lstinline+Stream+ protocol, the continuation session
type \lstinline+EndT+ does not matter because the protocol
never reaches the continuation. For a parameterized protocol,
\lstinline+select+ takes type parameters corresponding to the protocol
parameters (\lstinline+Int+) as well as the continuation session (\lstinline+EndT+).
Algebraic protocols sidestep issues with contractiveness of
recursive definitions and thus guardedness of the protocols: we always
have to make explicit progress by selecting the \lstinline+Next+ tag.

Up to this point, the examples we presented are definable with some
effort in previous work on session types. The following examples
exploit features that make \algst stand out:
\begin{itemize}
\item Traditional session types lack composability in the sense that
  abstraction over behavior is only possible in the tail
  position. Polymorphic context-free session types
  \cite{DBLP:journals/corr/abs-2106-06658} lift this restriction, but
  at the price of a type equivalence
  algorithm that takes doubly exponential time in the
  worst-case. \algst allows abstraction of any part of a protocol
  while providing a linear-time algorithm for type equivalence.
\item The protocol kind is a unique feature of \algst that enables
  abstraction over entire subprotocols and their orientation at any place in a
  protocol. This facility enables us to define a toolbox of generic
  protocol operators, inspired by regular language operations, and
  flip between client and server views using the negation operator.
\end{itemize}

\subsubsection{Generic servers}
\label{sec:generic-servers}

Parametric protocols enable the construction of generic servers that abstract over the
subsidiary protocols. Our first example created a stream of integers,
but what if we wanted a server for \lstinline+Stream Arith+ that  runs the \lstinline+Arith+ protocol ad infinitum?
To this end we write a generic server, \lstinline+stream+, that implements
the \lstinline+Stream a+ protocol, for any \lstinline+a+, and invokes a subsidiary server for the
sub-protocol \lstinline+a+. We abstract the type of this
sub-server with a parameterized type alias:
\begin{lstlisting}
type Service a = forall(s:S). ?a.s -> s
\end{lstlisting}
Thus armed, we can write the generic server and use it to run
\lstinline+Stream Arith+.
\begin{lstlisting}
stream : forall(a:P). Service a -> ?Stream a.EndT -> Unit
stream [a] serve c = match c with {
    Next c -> serve [?Stream a.EndT] c |> stream [a] serve }

streamArith : ?Stream Arith.EndT -> Unit
streamArith = stream [Arith] serveArith
\end{lstlisting}
We call the server constructed in this way \emph{passive} as it waits to receive a \lstinline+Next+
tag before running the subsidiary protocol. Dually, we can define an
\emph{active} server. 

\subsubsection{Active servers --- an exercise in duality and negation}
\label{sec:active-servers-an}

An active server offers the
\lstinline+Next+ tag to its client, indicating its readiness to accept a
request. Such a server generalizes the \lstinline+ones+ server.
All we need to do is flip two bits in the type:
\begin{lstlisting}
streamAct : forall(a:P). Service a -> !Stream -a.EndT -> Unit
streamAct [a] svc c =
  select Next [-a, EndT] c |> svc [!Stream -a.EndT] |> streamAct [a] svc
\end{lstlisting}
Why do we have to write  \lstinline+!Stream -a+ in the type?
One might expect the active server to run on a channel
whose type is the dual of \lstinline+?Stream a.EndT+, which is
\lstinline+!Stream a.EndW+.
But this change of direction also affects the
subsidiary protocol in parameter \lstinline+a+ as illustrated by the
typing of \lstinline+select+:
% If \lstinline+c+ has type \lstinline+!Stream Arith.end+, then
\begin{lstlisting}
  c : !Stream Arith.EndT |-
  select Next [Arith, EndT] c : !Arith.!Stream a. EndT
\end{lstlisting}
But this type is not an instance of the type
\lstinline+Service Arith = forall(s:S). ?Arith.s+.
Thanks to polarities, we can adapt to \lstinline+c : !Stream -Arith+ which switches the
direction to an output:
\begin{lstlisting}
  select Next [-Arith, EndT] c : ?Arith.!Stream -Arith. EndT
\end{lstlisting}

This way, we can reuse \lstinline+serveArith+ as follows.
\begin{lstlisting}
streamActArith : !Stream -Arith.EndT -> Unit
streamActArith = streamAct [Arith] serveArith
\end{lstlisting}

The same active server can also serve the infinite stream of \lstinline+ones+:
\begin{lstlisting}
streamActOnes : !Stream Int.EndT -> Unit
streamActOnes = streamAct [-Int] (sendInt 1)
\end{lstlisting}
Interestingly, the typing of \lstinline+streamActOnes+ involves \emph{double negation of
polarities}. Observe that \lstinline+sendInt 1 : Service -Int+, so that the transmitted stream is
in fact \lstinline+Stream -(-Int)+ $\equiv$ \lstinline+Stream Int+.

\subsubsection{A toolbox for generic servers}
\label{sec:toolb-gener-serv}

Inspired by the protocols \lstinline+Stream a+ and \lstinline+ListP a+
(from \cref{sec:introduction}), we can build a library of parameterized 
protocols with generic servers for them. These protocols reify the
standard type constructors, reminiscent of regular language operators.
\begin{lstlisting}
protocol Seq a b    = Seq a b                  -- product
protocol Either a b = Left a | Right b         -- sum
protocol Repeat a   = More a (Repeat a) | Quit -- iteration
\end{lstlisting}
We leave the straightforward implementation of the corresponding passive
generic servers, \lstinline+seq+, \lstinline+either+, and
\lstinline+repeat+, as well as their active variants, to the reader.

These protocol types and their generic servers are sufficient to
construct skeleton servers for all protocols whose flow can be described by regular languages. For example, we
can build the server for \lstinline+Repeat Arith+ by combination of
generic servers with the servers for negation and addition.
\begin{lstlisting}
type Neg    = Seq Int -Int
type Add    = Seq Int (Seq Int -Int)
type Arith  = Either Neg Add

serveNeg : Service Neg
serveNeg [s] c = match c with {
  Seq c -> let (x, c) = receiveInt [!Int.s] c in
           sendInt [s] (0-x) c }

serveAdd : Service Add
serveAdd [s] c = match c with {
  Seq c -> let (x, c) = receiveInt [?Seq Int -Int.s] c in
           match c with {
    Seq c -> let (y, c) = receiveInt [!Int.s] c in
	     sendInt [s] (x+y) c }}

serveArith : Service Arith
serveArith = either [Neg] serveNeg [Add] serveAdd

serveAriths : Service (Repeat Arith)
serveAriths = repeat [Arith] serveArith
\end{lstlisting}

There is some overhead compared to the arithmetic server shown in
\cref{sec:simple-protocols-add}. Previously, a client only had to handle the tags
\lstinline+Neg+ and \lstinline+Add+ besides the \lstinline+Repeat+
protocol. Here, \lstinline+Neg+ and \lstinline+Add+ are replaced by
\lstinline+Left+ and \lstinline+Right+ and the client has to handle
the \lstinline+Seq+ tags on top. Thus the protocol constructed from
the generic parts has some overheads, as we discuss in
\cref{sec:are-generic-servers}. We provide further examples of 
generic servers in \cref{sec:toolb-gener-serv-1}.

% \subsubsection{Transmission of parametric data}
% \label{sec:transm-param-data}

% Using parameterized protocols, we can define polymorphic operations for sending and receiving
% parametric datatypes. Suppose we want to send lists of arbitrary type \lstinline+a+, as long as someone can
% supply a way of sending a value of type \lstinline+a+.
% \begin{lstlisting}
% data List (a:T) = Nil | Cons a (List a)
% type Send (a:T) = a -> forall(s:S). !a.s -> s

% sendList : forall(a:T). Send a -> Send (List a)
% sendList [a] sendA xs [s] c = case xs of {
%     Nil       -> select Nil [a,s] c,
%     Cons x xs -> select Cons [a,s] c
%               |> sendA [!List a.s] x
% 	      |> sendList [a] sendA [s] xs }
% \end{lstlisting}
% To this end, one would provide implementations of \lstinline+Send a+
% for suitable primitive types \lstinline+a+ and modularly build further
% implementations.
% % \begin{lstlisting}
% % sendInt : Send Int
% % sendInt x [s] = send [Int, s] x
% % \end{lstlisting}
% \begin{lstlisting}
% sendListInt : Send (List Int)
% sendListInt = sendList [Int] sendInt

% sendListListInt : Send (List (List Int))
% sendListListInt = sendList [List Int] sendListInt
% \end{lstlisting}
% %Receiving can be handled similarly.

%%% Local Variables:
%%% mode: latex
%%% TeX-master: "main"
%%% End:

%% file: types.tex
\section{Types}
\label{sec:types}

The type structure underlying \algst comprises three parts: the kind
structure, the type language and protocol definitions. The type
language and protocol definitions are mutually recursive.  To
simplify the exposition, the formalized language is based on linear
System~F. While we formalize protocol
definitions, we do not formalize the standard treatment of (linear)
algebraic datatypes and pattern matching, although we use it in examples.

The type language encompasses types with different features
(as illustrated in \cref{sec:algebr-sess-types}). 
For a correct type abstraction, we need to endow the system with a classification 
of types through kinds.
%
% \paragraph{Kinds}
% \label{sec:kinds}
% \label{sec:what-your-kind}
%
\algst distinguishes three kinds.
\begin{description}
\item[$\kinds$] classifies session types, which in turn classify
  endpoints of communication channels.
% \item Kind $\kindm$ classifies types whose serialization is hard coded. These are
%   primarily primitive types like $\TInt$ and $\TBool$. We also
%   include session types.
\item[$\kindt$] classifies all types that classify run-time
  values.
\item[$\kindp$] classifies protocol types. The kind $\kindp$ subsumes any other
  kind by an implicit lifting that associates a type with a protocol
  suitable to transmit it. A protocol type per se classifies pure
  behavior, but no run-time values.
\end{description}

Kinds are linearly ordered in a subkinding relation, which is the
reflexive transitive closure of the strict ordering $\kinds
< \kindt < \kindp$. Type formation includes a standard kind
subsumption rule.

% The subsumption $\kindt < \kindp$ allows us to write types like
% $\TIn{(\TFun[] TU)}S$ or $\TOut{(\TForall\alpha\kind T)}S$. Such
% protocols are currently empty because they cannot be implemented in
% \algst. Alternatively, we could define another kind, $\kindl$, say,
% such that $\kindl < \kindp$, $\kindl < \kindt$, and $\kindt \not<
% \kindp$, but we chose to keep kinding simple.

% \paragraph{Functional types, session types, and protocol types}
% \label{sec:session-types}

% TYPES

% \Cref{fig:grammar-of-types} contains the grammar of types.

Here is the grammar of types:
%
% \begin{figure}[tp]
\begin{align*}
  \typec S, \typec T, \typec U \grmeq&
  \TUnit \grmor \TFun[] T U \grmor \TPair{T}{U} \grmor
  \TForall \alpha \kind T \grmor \typec{\alpha} \grmor
  && \text{functional types}
  \\
  & \TIn TS \grmor \TOut TS \grmor \TEndW \grmor \TEndT \grmor \TDual S \grmor
  && \text{session types}
  % \\
  % & \typec{\alpha} \grmor
  % & \text{functional and session types}
  \\
  & \TPApp \proto {\overline{T}} \grmor \pneg T 
  && \text{protocol types}
\end{align*}
%   \caption{Grammar of types}
%   \label{fig:grammar-of-types}
% \end{figure}
%
We mostly use the metavariable $\typec S$ to range over session
types (kind $\kinds$). Metavariables $\typec T, \typec U$ range over functional types
(kind $\kindt$) and protocol types (kind $\kindp$).
Functional types comprise the usual types of linear System F: unit,
function, product, universal quantification, and type variables.
% The handling of the unit type serves as blueprint for other primitive types.
The session type operators are receiving, sending,
passive and active termination, and the dual operator. The dual
operator swaps the direction of all communications in the spine of a
session type. Thus, a session type is a
sequence of communications either terminated by $\TEndW$ or by
$\TEndT$, or extensible with a type variable $\typec\alpha$ of kind $\kinds$.

The protocol type
$\TPApp \proto {\overline{T}}$ refers to the protocol declared as
$\proto$ with parameters $\typec{\overline{T}}$. The
operator ``\typec{$-$}'' reverses the direction of communication of its
protocol parameter. While duality transforms session types
outside-in, the reverse operator transforms protocols inside-out.

%
% Unlike in previous work on context-free session types
% \cite{DBLP:conf/icfp/ThiemannV16,DBLP:journals/toplas/Padovani19},
% there is no sequencing operator $\TSeq\_\_$, but sending and receiving
% types have their traditional prefix form,

A kind context $\Delta$ maps type variables to kinds and protocol
names to first-order arrow kinds that determine the
number of parameters of the  protocol.
Formally, we should write $\kindas\varepsilon\kindp$ for the kind of a
protocol constructor without parameters, but we often abbreviate it to $\kindp$.
% In informal settings, we just write $\kind$.
\begin{align*}
  \Delta & \grmeq \Empty
           \grmor \Delta, \tbind\alpha\kind
           \grmor \Delta, \tbind\proto{\kinda\kindp\kindp}
  % \\
  % \Pi & ::= M
  %       \grmor \protocolsfull; \Pi
\end{align*}

% TYPE FORMATION

\input{fig-alg-type-formation}
%\input{fig-type-formation}

\Cref{fig:type-formation} contains the type formation rules. We
present them in algorithmic form from the beginning. The judgment
$\typeSynth T \kind$ assigns  to type $\typec T$ its minimal kind $\kind$ under kind
context $\Delta$.  The corresponding checking judgment $\typeAgainst T \kind$ takes
inputs $\typec T$ and $\kind$ and checks if the synthesized kind is a
subkind of the expected one.

The types for unit (\rulenametypeUnit),  functions
(\rulenametypeArrow), pairs (\rulenametypePair), and universal
polymorphism (\rulenametypePoly)
have kind $\kindt$ like their constituents. Variables have any kind
that has been assigned to them (\rulenametypeVar).
\algst's notion of session type distinguishes passive $\TEndW$ and active
$\TEndT$ termination of a session (\rulenametypeEndW, \rulenametypeEndT), and the
dual operator (\rulenametypeDual). Conventional session types require types for
the values exchanged in messages; here we ask for protocols instead
(\rulenametypeDotIn and \rulenametypeDotOut).
Rule {\rulenametypeProtocol} builds a protocol type from a protocol name
$\proto$. We assume that the kind context contains a binding for $\proto$ that
determines its arity. All parameters to a protocol must be protocols
themselves. The latter constitutes no restriction because any type can be
lifted to a protocol. % We omit empty parameter lists.
A protocol type can change direction using the reverse operator (rule {\rulenametypeMsgNeg}).

% \begin{lemma}[Subsumption] If $\typeAgainst T\kind$ and $\isSubKind \kind
%   \kind'$, then $\typeAgainst T{\kind'}$.
% \end{lemma}
% \begin{proof}
%   Immediate from \rulenametypeSub and transitivity of subkinding.
% \end{proof}

% \paragraph{Protocol definitions}
% \label{sec:protocol-definitions}
% \label{sec:decl-spec-prot}

Algebraic protocol types $\TPApp \proto {\overline{\alpha}}$ are defined analogously to algebraic
datatypes.
\begin{gather*}
  \protocolsfull
  % \qquad
  % \protocolsdata
\end{gather*}
This declaration defines the protocol type constructor $\proto$ that takes as
many parameters as indicated by $\typec{\overline\alpha}$. The protocol has choices
indexed by $i\in I$. Each choice is tagged by a \emph{selector tag} $\termc{\constr_i}$
that guards a sequence of subprotocols $\typec{\overline{T_i}}$, to be processed in that
order. We assume that selector tags are globally unique. Program may use any
full application $\TPApp \proto {\overline U}$ as a protocol type.
%
% For a data protocol, we can indicate our intent by using the $\datak$
% keyword instead of $\protocolk$. Moreover, all parameters must be
% proper types, not protocols.
% \begin{gather*}
%   \protocolsdata
% \end{gather*}
%
%
% We start with a declarative rule set to determine well-formedness of a
% protocol definition.
%
We check a single protocol declaration using the formation rules where $\Delta$ may
refer to previously defined protocols:\footnote{For mutually
  recursive protocols $\typec{\overline\proto}$, we add the
  assumptions for all $\typec{\overline\proto}$ when checking each individual $\typec{\proto_i}$.}
\begin{mathpar}
  \inferrule{
    \protocolsfull \\
    \typeAgainst[\Delta, \tbind\proto{\kinda\kindp\kindp},
    \tbind{\overline{\alpha}}{\overline{\kindp}}]
    {\overline{T_i}}{\kindp}\\
    (\forall i\in I)
  }{
    \typeSynth{\proto}{\kinda\kindp\kindp}
  }
  % \\
  % \inferrule{
  %   \protocolsdata \\
  %   \typeAgainst[\Delta, \tbind\data{\kinda\kindt\kindt},
  %   \tbind{\overline{\alpha}}{\overline{\kindt}}]
  %   {\overline{T_i}}{\kindt}\\
  %   (\forall i\in I)
  % }{
  %   \typeSynth{\data}{\kinda\kindt\kindt}
  % }
\end{mathpar}

% \paragraph{Type conversion}
% \label{sec:type-conversion}

% TYPE CONVERSION

\input{fig-type-conversion}

Types obey a kind-indexed conversion
relation. \Cref{fig:type-conversion} presents the declarative
rules. Conversion is reflexive, symmetric, and transitive
(\rulenametconvReflex, \rulenametconvSymm, \rulenametconvTrans).
We define duality as a conversion on session types. The rules entail
that duality changes the direction of transmissions from the outside
(\rulenametconvDualEndW, \rulenametconvDualEndT,
\rulenametconvDualInM, \rulenametconvDualOutM)
and that duality is involutory (\rulenametconvDoubleS). Kind subsumption is lifted to conversion (\rulenametconvSub). The last line
governs the interaction between the reverse operator and the session
type constructors: reverse flips the direction from the inside
(\rulenametconvNegIn, \rulenametconvNegOut). Reverse is
also involutory (\rulenametconvNegInv). We omit the standard congruence rules.

%%% TYPE OPS _ Moved to expression typing

% The following two type operators are used to \emph{polarize} the protocol's invocation.
% \pt{extended the metafunctions ``!'' and ``?'' to $O$ by making them
%   act correctly on guards.}

% \input{fig-polarized-constructors}

% Rather obvious
% \begin{lemma}[Duality is for sessions]
%   \label{lem:duailty-st}
%   If $\typeSynth{\TDual{T}}{\kind}$, then $\typeAgainst{\TDual{T}}{\kinds}$ and
%   $\typeAgainst{T}{\kinds}$.
% \end{lemma}

While type formation is algorithmic, the declarative definition of
type conversion is more perspicuous. For the
upcoming algorithmic expression typing relation we define a normalization
algorithm for types, which is sound and complete with respect to the
declarative rules.

% We start by specifying the syntax of normal forms. Normal forms for well-formed types
% are $\typec{Q}$ as defined by the following grammar:
% %
% \begin{align*}
%   \typec{Q} &\grmeq \typec{R} \grmor \pneg{R}
%   \\
%   \typec{R} &\grmeq
%   % functional types
%   \TUnit \grmor \TFun[]{R}{R} \grmor \TPair{R}{R} \grmor
%   \TForall \alpha \kind {R} \grmor \typec\alpha
%   % session types
%   \grmor \TIn{R}{R} \grmor \TOut{R}{R}
%   \grmor \TEndW \grmor \TEndT  \grmor
%   \TDual{\alpha}
%   % protocol and data types
%   \grmor \TPApp \proto {\overline{Q}}
% \end{align*}
In a type in normal form,  the reverse operator can only occur at most once at the top level (of
a type of kind $\kindp$) and at the top level of the parameters of a protocol type. The
$\dualk$ operator only appears on type variables, at the end of
the spine of a session type. Thus an algorithm to compute normal forms
must push all $\dualk$ operators down the spine of a session type
by applying the conversion rules \rulenametconvDualInM, \rulenametconvDualOutM,
\rulenametconvDualEndW, and \rulenametconvDualEndT from left to
right. If it encounters another $\dualk$ operator along the
spine, it applies rule \rulenametconvDoubleS. We cannot further resolve
$\TDual\alpha$, hence it is a normal form.

The reverse operator only applies to protocol types. So it can only
occur in a message type like $\TIn{(-T)}S$ and in the parameter of a
protocol type. In a message type, normalization must apply rules
\rulenametconvNegIn and \rulenametconvNegOut exhaustively from left to
right, which removes the reverse operator from message types. In the
parameter of a protocol type, normalization applies rule
\rulenametconvNegInv exhaustively. Starting with an odd number of
reverse operators, a single operator remains; starting with an even
number, all reverse operators are removed.

\input{fig-normalisation}

Following this intuition normalization is defined in \cref{fig:normalisation} by two mutually
recursive functions $\Normal[+]\_$ and $\Normal[-]\_$ along with three
auxiliary functions. % also  in \cref{fig:type-polarized-operator}.
Positive normalization $\Normal[+]\_$ traverses and reconstructs all
non-session type constructs. It also gets invoked on the type of the
transmitted value for session types, i.e., on $\typec T$ in type
$\TOut TS$.
Normalizing the $\dualk$ operator switches forth and back  between
positive and negative normalization, $\Normal[+]\_$ and $\Normal[-]\_$.

Negative normalization with superscript
$^-$ indicates a pending $\dualk$ type constructor, so that
the function $\Normal[-]\_$ is only used for session
types. Normalization pushes the pending dual constructor along the spine of a session type to
implement rules \rulenametconvDualInM, \rulenametconvDualOutM,
\rulenametconvDualEndW, and \rulenametconvDualEndT. The pending dualization gets reified on
a type variable.

Normalizing a reversed protocol $\TMinus T$ invokes the negative directional
function $\polarizedparam - {T'}$ on $\typec
T$'s normal form to enforce the
conversion rule \rulenametconvDoubleS.

% As an example for the
% normalization of session types, we consider $\Normal{\TIn TS}$. In each case, we invoke $\Normal S$ on
% the continuation session and invoke $\Normal[+]T$ on the transmitted type
% $\typec T$. For sending (receiving), we invoke the positive (negative) directional function
% on the result. Materialization $\tosession TS$ expects a normal form for
% $\typec T$ and ``materializes'' the direction of the session type. It returns an input $\TIn{T'}S$ if the normal form is
% negative and an output $\TOut{T'}S$ if it is positive.
%
Here is an example of a step-by-step computation of the positive normal form of
 $\TDual{(\TIn{(-\TInt)}{\alpha})}$.
%\Cref{fig:example-computation-normal-form} goes through a concrete normalization step by step.
%
% \begin{figure}
  \begin{align*}
    & \Normal[+]{\TDual{(\TIn{(-\TInt)}{\alpha})}} \\
    ={}& \Normal[-]{{\TIn{(-\TInt)}{\alpha}}}  && \text{negative
                                                  normalization
                                                  remembers pending $\dualk$}\\
    ={}& \TFOut{\Normal[+]{-\TInt}}{\Normal[-]\alpha}  &&
                                                          \text{$\Normal[+]\_$
                                                          on the
                                                          payload,
                                                          $\Normal[-]\_$
                                                          on the spine}\\
    ={}& \TFOut{\TFMinus{\Normal[+]{\TInt}}}{\TDual\alpha} &&
                                                              \text{negative
                                                              normalization
                                                              reifies
                                                              $\dualk$
                                                              on type variable}\\
    % ={}& \TFOut{\TFMinus{\TInt}}{\TDual\alpha} \\
    % ={}& \TFOut{-{\TInt}}{\TDual\alpha} \\
    ={}& \tosession{\TMinus\TInt}{\TDual\alpha} &&
                                                   \text{materialization
                                                   applied to normal
                                                   form \dots}\\
    ={}& \TIn{{\TInt}}{\TDual\alpha} && \text{\dots{} fixes the direction}
  \end{align*}
%   \caption{Example: Step-by-step computation of the positive normal form of
%   $\TDual{(\TIn{(-\TInt)}{\alpha})}$}
%   \label{fig:example-computation-normal-form}
% \end{figure}

To test whether a type equivalence goal $\isConv TU\kind$ holds, we
compute
\emph{algorithmic type equivalence}, notation $\isEquiv TU$, which is
defined by comparision of normal forms: $\Normal[+]T = \Normal[+]U$  up to  $\alpha$-equivalence.
% We write $\typec T =_\alpha \typec U$ if types $\typec T$ and $\typec U$ are
% $\alpha$-equivalent.
% 
The following propositions establish soundness and
completeness of the normalization algorithm, as well as its worst case running
time.
% , which is linear in the size of the input type.

% Every type is convertible to its normal form.
%
\begin{theorem}[Soundness]
  \label{thm:soundness-type-equiv}
  Let $\typeSynth T\kind$ and $\typeSynth U\kind$.
  \begin{enumerate}
  \item\label{it:soundness-type-equiv-t} If $\Normal[+] T =_\alpha \Normal[+] U$,
    then $\isConv TU\kind$.
  \item\label{it:soundness-type-equiv-s} If $\kind = \kinds$ and  $\Normal[-] T =_\alpha \Normal[-]
    U$, then $\isConv TU\kinds$.
  \end{enumerate}
\end{theorem}

% If two types are convertible, then their normal forms are equal.
\begin{theorem}[Completeness]
  \label{thm:completeness-type-equiv}
  Let $\isConv TU\kind$.
  \begin{enumerate}
  \item  $\Normal[+]T =_\alpha \Normal[+]U$.
  \item    If $\kind=\kinds$, then $\Normal[-]T =_\alpha \Normal[-]U$.
  \end{enumerate}
\end{theorem}

\begin{theorem}[Time complexity of type equivalence]
  \label{thm:complexity}
  The test for\/ $\isEquiv TU$ runs in $O(|\typec T|+|\typec U|)$.
\end{theorem}

%%% Local Variables:
%%% mode: latex
%%% TeX-master: "main"
%%% End:

%% file: fig-alg-type-formation.tex
\begin{figure}[t!]
  \declrel{Type formation (synthesis, check-against)}{$\typeSynth T
    \kind$,\quad $\typeAgainst T \kind$}
  \begin{mathpar}
    % FUNCTIONAL TYPES
    \infrule{\rulenametypeUnit}{}{\typeSynth\TUnit\kindt}

    \infrule{\rulenametypeArrow}
    {\typeAgainst T\kindt \\ \typeAgainst U\kindt}
    {\typeSynth{\TFun[] TU}\kindt}
    
    \infrule{\rulenametypePair}
    {\typeAgainst T\kindt \\ \typeAgainst U\kindt}
    {\typeSynth{\TPair TU}\kindt}

    \infrule{\rulenametypeVar}{\tbind\alpha\kind \in \Delta}{\typeSynth\alpha\kind}
    
    \infrule{{\rulenametypePoly}}
    {\typeAgainst[\Delta,\tbind\alpha\kind] T {\kindc{\kindt}}}
    {\typeSynth {\TForall \alpha\kind T}{\kindt}}

    % SESSION TYPES
    % \infrule{\rulenametypeEndW}{}{\typeSynth\TEndW\kinds}
%
    % \infrule{\rulenametypeEndT}{}{\typeSynth\TEndT\kinds}
%
    \infrule{\rulenametypeDotIn}
    {\typeAgainst T \kindp \\ \typeAgainst S \kinds}
    {\typeSynth{?T.S} \kinds}

    \infrule{\rulenametypeDotOut}
    {\typeAgainst T \kindp \\ \typeAgainst S \kinds}
    {\typeSynth{!T.S} \kinds}

    \begin{minipage}{2cm} 
      \begin{mathpar}  
        \infrule{\rulenametypeEndW}{}{\typeSynth\TEndW\kinds}

        \infrule{\rulenametypeEndT}{}{\typeSynth\TEndT\kinds}
      \end{mathpar} 
    \end{minipage}

    \infrule{\rulenametypeDualS}
    {\typeAgainst S \kinds}
    {\typeSynth{\TDual S} \kinds}
    
    % GUARDED
    \infrule{\rulenametypeProtocol}
    { \tbind \proto {\kinda\kindp\kindp} \in \Delta
      \\ \typeAgainst{\overline T}{\overline{\kindp}}}
    {\typeSynth{\TPApp \proto {\overline{T}}}\kindp}

    \infrule{{\rulenametypeMsgNeg}}
    {\typeAgainst{T}{\kindp}}
    {\typeSynth{{-T}}\kindp}

    \infrule{\rulenametypeSub}
    {\typeSynth T\kind
      \\ \isSubKind \kind \kind'}
    {\typeAgainst T\kind'}
  \end{mathpar}
  % \declrel{Type formation (check-against)}{$\typeAgainst T \kind$}
  % \begin{mathpar}
  %   \infrule{\rulenametypeSub}
  %   {\typeSynth T\kind
  %     \\ \isSubKind \kind \kind'}
  %   {\typeAgainst T\kind'}
  % \end{mathpar}
  \caption{Algorithmic type formation rules}
  \label{fig:type-formation}
\end{figure}

% \infrule{\rulenametypeDot}
% {\typeAgainst G \kindg \\ \typeAgainst S \kinds}
% {\typeSynth{G.S} \kinds}
% 
% \infrule{\rulenametypeUp}{\typeAgainst T \kindt}{
% \typeSynth{\TUp T}\kindp}
% 
% \infrule{{\rulenametypeMsgPos}}
% {\typeAgainst{T}{\kindp}
% }
%   {\typeSynth{{+T}}\kindp}
%   
%   \infrule{\rulenametypeInt}{}{\typeSynth\TInt\kindt}
%   
%   \infrule{\rulenametypeBool}{}{\typeSynth\TBool\kindt}
%   

%%%   Local Variables:
%%%   mode: latex
%%%   TeX-master: "main"
%%%   End:

%% file: fig-type-conversion.tex
\begin{figure}[t!]
  \declrel{Type conversion}{$\isConv T T \kind$}
  \begin{mathpar}
    % EQUIVALENCE
    \infrule{\rulenametconvReflex}{\typeAgainst T \kind}{\isConv T T \kind}

    \infrule{\rulenametconvSymm}{\isConv {T_1} {T_2} \kind}{\isConv {T_2} {T_1} \kind}

    \infrule{\rulenametconvTrans}{\isConv {T_1} {T_2} \kind \\ \isConv {T_2} {T_3} \kind}{\isConv {T_1} {T_3} \kind}

    \infrule{\rulenametconvDualEndW}{}{\isConv{\TDual\TEndW}\TEndT\kinds}
    
    \infrule{\rulenametconvDualEndT}{}{\isConv{\TDual\TEndT}\TEndW\kinds}
    
    \infrule{\rulenametconvDualInM}{\typeAgainst T \kindp \\ \typeAgainst S {\kinds}}{
      \isConv{\TDual{(\TIn TS)}} {\TOut T{(\TDual S)}} \kinds}

    \infrule{\rulenametconvDualOutM}{\typeAgainst T \kindp \\ \typeAgainst S {\kinds}}{
      \isConv{\TDual{(\TOut TS)}} {\TIn T{(\TDual S)}} \kinds}

    \infrule{\rulenametconvDoubleS}{\typeAgainst S {\kinds}}{\isConv{\TDual{(\TDual S)}} S {\kinds}}

    \infrule{\rulenametconvSub}{
      \isConv {T_1}{T_2} {\kind} \\ \isSubKind\kind{\kind'}
    }{
      \isConv {T_1}{T_2} {\kind'}
    }
    
    \infrule{\rulenametconvNegIn}{\typeAgainst T \kindp \\ \typeAgainst S \kinds}{\isConv
      {\TIn{(-T)}S}{\TOut TS} \kinds}

    \infrule{\rulenametconvNegOut}{\typeAgainst T \kindp \\ \typeAgainst S \kinds}{\isConv
      {\TOut{(-T)}S}{\TIn TS} \kinds}
    
    \infrule{\rulenametconvNegInv}{\typeAgainst T \kindp}{\isConv {-(-T)}T \kindp}
  \end{mathpar}
  \caption{Type conversion (congruence rules omitted)}
  \label{fig:type-conversion}
\end{figure}

%% file: fig-normalisation.tex
\begin{figure}[t!]
    \declrel{Type normalization}{$\Normal[+]T=\typec T \quad \Normal[-]T=\typec T$}

  \begin{minipage}{0.49\linewidth}
    \begin{align*}
    % Positive normalisation
    % Functional
    \Normal[+]\TUnit &= \TUnit \\
    \Normal[+]{\TFun[] TU} &= \TFun[]{\Normal[+] T}{\Normal[+] U} \\
    \Normal[+]{\TPair TU} &= \TPair{\Normal[+] T}{\Normal[+] U} \\
    \Normal[+]{\TForall \alpha \kind T} &= \TForall\alpha \kind {\Normal[+] T} \\
    % Functional and session
    \Normal[+]\alpha &= \typec\alpha \\
    % Session
    \Normal[+]{\TIn TS} &= \TFIn{\Normal[+]T}{\Normal[+] S} \\
    \Normal[+]{\TOut TS} &= \TFOut{\Normal[+]T}{\Normal[+]S} \\
    \Normal[+]\TEndW &= \TEndW \\
    \Normal[+]\TEndT &= \TEndT\\
    \end{align*}
  \end{minipage}
  \begin{minipage}{0.49\linewidth}
    \begin{align*}
    \Normal[+]{\TDual T} &= \Normal[-]{T} \\
    % Protocol
    \Normal[+]{\TPApp \proto{\overline T}} &= \typec{\TPApp \proto {\overline{\Normal[+] {T}}}} \\
    \Normal[+]{\pneg T} &= \polarizedparam - {\Normal[+] T}\\
    \Normal[-]{\TDual T} &= \Normal[+]{T} \\
    % Negative noralisation
    \Normal[-]\alpha &= \TDual\alpha \\
    \Normal[-]{\typec{\TIn TS}} &= \TFOut{\Normal[+]T}{\Normal[-]S} \\
    \Normal[-]{\typec{\TOut TS}} &= \TFIn{\Normal[+]T}{\Normal[-]S} \\
    \Normal[-]\TEndW &= \TEndT \\
    \Normal[-]\TEndT &= \TEndW \\
    \end{align*}
  \end{minipage}

  \declrel{Materialization and the directional operators}
  {$\tosession TT = \typec{T} \quad 
    \polarizedparam - T = \typec{T} \quad
    \polarizedparam + T = \typec{T}$}
  \begin{align*}
    \tosession {\pneg T} U &= \TIn TU & 
    \polarizedparam - {\pneg T} &= \polarizedparam + T &
    \polarizedparam + {\pneg T} &= \polarizedparam - T
    \\
    \tosession TU &= \typec{\TOut TU} & 
    \polarizedparam - T &= \pneg T &
    \polarizedparam + T &= \typec{T}
  \end{align*}
  In the second line $\typec T \neq \pneg U$ in all cases.
  \caption{Type normalization and auxiliary metafunctions on session types and protocol types}
  \label{fig:normalisation}
  \label{fig:type-polarized-operator}
\end{figure}

%%% Local Variables:
%%% mode: latex
%%% TeX-master: "main"
%%% End:

%% file: processes.tex
\section{Expressions and Processes}
\label{sec:processes}

This section introduces the notions of expressions and processes, algorithmic
typing for expressions, and operational semantics for expressions and processes
based on labelled transition relations.
%
% \paragraph*{Syntax}
%
The syntax of constants $\termc c$, values $\termc v$, expressions
$\termc e$ and processes $\termc p$ is defined by the
grammar below.
\begin{align*}
  \termc c \grmeq & \unitk
       % \grmor \reck
       \grmor \forkk 
       \grmor \newk
       \grmor \receivek 
       \grmor \sendk 
       \grmor \selecte C
       \grmor \waitk
       \grmor \terminatek
  \\
  \termc v \grmeq & \termc{c}
      \grmor \termc{x}
      \grmor \abse{x}{T}{e}
      \grmor \rece xTv
      \grmor \tabse\alpha\kind v
      \grmor \paire{v}{v}
      \grmor \rece xTv
      \grmor \tappe{\receivek}{T}
  \\& \tappe{\tappe{\receivek}{T}}{T}
      \grmor \tappe{\sendk}{T}
      \grmor \tappe{\tappe{\sendk}{T}}{T}
      \grmor \appe{\tappe{\tappe{\sendk}{T}}{T}}{v}
      \grmor \tappe{\selecte C}{\overline T}
  \\
  \termc e \grmeq & \termc{v}
       \grmor \appe{e}{e}
       \grmor \tappe{e}{T}
       \grmor \letu ee
       \grmor \paire ee
       \grmor \lete xxee
      \grmor
  \\& \matchwithe eCxeiI
  \\
  \termc p \grmeq & \expp e 
  \grmor \parp{p}{p}
  \grmor \newp{x}{y}{p} 
\end{align*}

All constructors are standard either in linear functional languages or in
session type languages. From functional languages we find function introduction
and elimination ($\abse xTe$ and $\appe{e_1}{e_2}$), linear pair introduction
and elimination ($\paire{e_1}{e_2}$ and $\lete xy{e_1}{e_2}$), linear unit
introduction and elimination ($\unitk$ and $\letu{e_1}{e_2}$). From session
languages we find most of the constructors in the form of constants, including
$\forkk$ to spawn a new thread, $\newk$ to create a new channel, $\receivek$ to
receive a value on a given channel, $\sendk$ to send a value on a given channel,
$\waitk$ to wait for a channel to be closed, $\terminatek$ to force channel
closing, and $\appe \selectk {C}$ to exercise a choice on a given
channel indicated by selector name $\constr$. The
operator to offer a choice on some channel cannot be captured by a type in our
type language, hence we have made it an expression ($\matchwithe eCxeiI$).
At the level of processes, we have expressions as threads $\expp e$,
parallel composition $\termc{p_1 \PAR p_2}$, and 
introduction of a channel by declaring its
two ends $\termc x, \termc y$: $\newp xyp$.

% \paragraph{Typing}

%\input{fig-ctx}
\input{fig-type-constants}
\input{fig-alg-typing}

Expression typing is again mostly standard with respect to (linear) functional
or session type languages. The definitions for typing assume a set of
protocol declarations as an implicit argument. The typings of
$\appe\selectk{C}$ (\cref{fig:type-constants}) and $\matchk$
(\cref{fig:alg-typing}, \rulenameexpMatch) access these protocol declarations. 
The $\typeofop$ operator in~\cref{fig:type-constants} returns the
types for constants in normal form. The typing rules for expressions are defined
in~\cref{fig:alg-typing}. 
The type system is algorithmic, implemented with two mutually recursive
judgments: one to synthetize a type of an expression, the other to check an
expression against a type. The typing rules rely on type variable contexts
$\Delta$ introduced in \cref{sec:types}, on term variable contexts $\Gamma$
assigning types $\typec T$ to term variables $\termc x$. There are two sorts of
entries in term variable contexts: linear entries, noted $\ebind xT$, and
unrestricted entries, noted $\eubind xT$. The latter are used for recursive
functions only (\cf rules \rulenameexpUVar and \rulenameexpRec) and are inspired
by work on quantitative type theory
\cite{DBLP:conf/lics/Atkey18,DBLP:conf/birthday/McBride16,DBLP:journals/pacmpl/BernardyBNJS18}.

The type
system maintains the invariant that entries in term variable contexts are always normalized:
the rules that write into contexts, namely
\rulenameexpAbs, \rulenameexpRec and \rulenameexpLet, normalize the
new entry. 
Judgments are of the form $\expSynth{\Gamma_1} eT{\Gamma_2}$ or
$\expAgainst{\Gamma_1} eT{\Gamma_2}$, where $\Gamma_2$ represents the
part of context $\Gamma_1$ \cite{walker:substructural-type-systems}
that is not consumed by $\termc e$ and $\typec T$ is in normal form.
% or the unused part of a session type for a given endpoint
% \cite{DBLP:journals/iandc/Vasconcelos12,DBLP:conf/forte/ZalakainD21}.
The axioms
(\rulenameexpConst and \rulenameexpVar) mostly copy the input context to the
output. Rule \rulenameexpApp is a good example of context threading: 
to type expression $\appe {e_1}{e_2}$, expression $\termc{e_1}$ is given the
input context $\Gamma_1$, produces $\Gamma_2$ which is given to $\termc{e_2}$,
which turn produces $\Gamma_3$ that constitutes the outgoing context of
application $\appe {e_1}{e_2}$.
The rules that add linear entries to the context (\rulenameexpAbs, \rulenameexpLet
and \rulenameexpMatch) all check that the new variables are not present
in the outgoing context, thus ensuring that the corresponding values are fully
consumed. 
Rule \rulenameexpTApp explicitly normalizes the output type for normalization is
not preserved by substitution; all rules that introduce new entries in the
context (\rulenameexpAbs, \rulenameexpRec, \rulenameexpLet, and
\rulenameexpMatch) make sure that types are normalized, thus ensuring that contexts contain only
normalized types.
% other rules
% construct types in normal form, either by assumption or by induction on a
% premise (see \cref{sec:appendix-exps-procs}).

% ======= I don't think we need this
% Rule \rulenameexpTApp normalizes the output type, all other rules
% construct types in normal form, either by assumption or by induction
% on a premise (see \cref{sec:appendix-exps-procs}, \cref{lem:synth-against-nf}).
% >>>>>>> e9014ece406b68d500fbdef7db6faf10a460f9ae

Rules \rulenameexpRec and \rulenameexpUVar govern the typing of
recursive functions. Rule \rulenameexpRec binds the recursive function $x$
using an unrestricted binding $\eubind x\_$ so that any number of
recursive calls is possible in the body via \rulenameexpUVar. Unrestricted
bindings provide for termination of recursive functions. Such functions cannot use linear variables from the
environment, which we enforce by requiring the incoming context
$\Gamma$ to be equal to the outgoing context.

The truly new rule is \rulenameexpMatch, which indicates that all
rules are parametric with respect to an implicit set of protocol
declarations: there is a declaration for each protocol $\proto$ which
has a kinding in $\Delta$. The rule starts by synthesizing the 
input type of the expression to be matched;
the type is 
$\TIn{(\TPApp \proto {\overline U})}{S}$ for some
protocol constructor $\proto$. At this point, the declaration of protocol
$\proto$ is looked up in the implicit set of declarations. The
$\matchk$ branch for each choice of selector $\termc{C_i}$ is then typed in turn. The
parameters $\typec{\overline \alpha}$ of the protocol are instantiated with
the arguments $\typec{\overline U}$ to produce a sequence of types
$\typec{\overline{T_i}\tsubs{\overline U}{\overline \alpha}}$. This sequence of
types is converted to a normalized session type as described in \cref{sec:types},
which is assigned to variable $\termc{x_i}$. Now, each branch synthesizes a type
$\typec{V_i}$ and context $\Gamma_i$, for $i\in I$. In all branches,
the output types and the contexts must coincide (up to $\alpha$
equivalence) with the output type 
and context of the $\matchk$, i.e., $\typec V$ and $\Gamma_3$. In practice,
the implementation selects one branch, say $k\in I$, and compares all
other output types and contexts with $\typec{V_k}$ and $\Gamma_k$.
In this rule, and also in the type for $\selectk$ in \cref{fig:type-constants}, we extend the
directional operators and  materialization (from \cref{fig:type-polarized-operator}) 
to sequences
of parameters by mapping:
$\polarizedparam{\pm}{\overline{T}}{} =
\overline{\polarizedparam{\pm}{T}{}}$ as well as
$\tosession {\varepsilon}S = \typec S$ and $\tosession {T\overline T}S =
\tosession{T}{\tosession {\overline T}S}$.

Rule \rulenameexpCheck relies on the previous judgement to check expression $\termc e$ against 
type $\typec T$, up to $\alpha$-equivalence. This simple comparison is
sufficient because synthesis outputs $\typec U$ in normal form and
checking must be invoked with $\typec T$ in normal form.
%
% \begin{lemma}\
%      \label{lem:synth-against-nf}
%      \begin{enumerate}
%           \item\label{it:synth-nf} If $\expSynth{\Gamma_1} eT{\Gamma_2}$ and $\isCtx[\Delta]{\Gamma_1}{\kindt}$, 
%           then $\isCtx[\Delta]{\Gamma_2}{\kindt}$ and $\typec T$ is in normal form.
%           \item\label{it:against-nf}  If $\expAgainst{\Gamma_1} eT{\Gamma_2}$ and $\isCtx[\Delta]{\Gamma_1}{\kindt}$
%           and $\typec T$ is in normal form, then $\isCtx[\Delta]{\Gamma_2}{\kindt}$.
%      \end{enumerate}
%      \begin{proof}\
%      \begin{enumerate}
%           \item By rule induction on the first hypothesis.
%           \item Using \cref{it:synth-nf} and observing that type $\typec U$ such that 
%           $\typec U =_\alpha \typec T$ is also in normal form.
%           \qedhere
%      \end{enumerate}         
%      \end{proof}
% \end{lemma}
%
%
% \paragraph*{Operational semantics}

The operational semantics for expressions and processes is defined via a 
transition relation labelled by the below actions.
\begin{align*}
  \sigma &\grmeq \receivel xv
	   \grmor \sendl xv
       \grmor \receivel xC
	   \grmor \sendl xC
       \grmor \closel x
	   \grmor \openl x
       && \text{labels for session operations}
  \\
  \lambda &\grmeq \sigma
       \grmor \beta
       \grmor \forkl v
	   \grmor \newl xxT
       && \text{labels for expressions}
  \\
  \pi &\grmeq \sigma 
  \grmor \tau
  \grmor \scopel xxxx
  \grmor (\parl \pi \pi)
       && \text{labels for processes}
\end{align*}
The $\sigma$ labels capture the six operations on sessions: $\receivek$,
$\sendk$, $\matchk$, $\selectk$, $\waitk$ and $\terminatek$.
Labels $\lambda$ are for expressions: $\beta$ is for internal actions and the
remaining two for $\forkk$ and $\newk$.
Labels $\pi$ are for processes: $\tau$ is for silent actions (the counterpart of
$\beta$ for expressions), $\scopel abxa$ or $\scopel abxb$ is for opening the
scope of a $\termc{(\nu ab)}$ binding, and $\pi\PAR\pi$ for parallel
composition.

\input{fig-label-exps-abbrev}

The axioms of the labelled transition system for
expressions is defined by the rules in \cref{fig:lts-exps-abbrev}
(\cref{sec:appendix-exps-procs}, \cref{fig:lts-exps}, contains the whole set).
We group under label $\beta$ the common reduction rules of the polymorphic
lambda calculus. Rules Act-Fork and Act-New record in the label the value forked
and the names and types of the channel ends created. These two transitions will
then be handled at the process level. Then we have the six transitions for the
common session operations.
% These are now in appendix
% Lastly we have the call-by-value structural rules.
% For example, rule Act-AppL reduces the left hand side of an application, whereas
% Act-AppR reduces the right hand side when the left is a value.

\input{fig-label-procs}

The labelled transition system for processes is defined by the rules in
\cref{fig:lts-procs}. The rules are adapted from
\citet{DBLP:conf/concur/0001KDLM21,DBLP:journals/corr/abs-2105-08996} which in
turn are inspired by \citet{DBLP:journals/corr/abs-2106-11818}. The first four
rules convert expression transitions into process transitions: the session
transitions $\sigma$ are passed directly from expressions to processes (rule
Act-Session). The next three rules all yield silent ($\tau$) transitions:
Act-Fork launchs a new thread and Act-New creates a new channel. Rules Act-JoinL
and Act-JoinR account for the commutative nature of parallel composition. Rules
Act-ParL, Act-ParR and Act-Res allow transitions underneath parallel composition
and scope restriction. The rules for opening and closing a scope (left and right
versions) are adapted to the double binder scheme from the labelled transitions
systems for the $\pi$-calculus~\cite{DBLP:journals/iandc/MilnerPW92b}. These
rules are required for we work with free output, unlike the above cited
works~\cite{DBLP:conf/concur/0001KDLM21,DBLP:journals/corr/abs-2105-08996,DBLP:journals/corr/abs-2106-11818}
that work with bound output.

Reduction underneath a prefix (rule Act-Res) can only occur if the channels ends
$\termc x$ and $\termc y$ (two bound variables) do not capture the free
variables in the label. For this purpose, we define the free variables of a
process label $\pi$ as follows.
\begin{align*}
  % sigma labels
  \FVL {\receivel xv} = \FVL {\sendl xv} = \{x\} \cup \FVE v
  \qquad
  \FVL {\receivel xC} = \FVL {\sendl xC} = \FVL {\closel x} = \FVL {\openl x} = \{x\}
  \\
  \FVL \tau = \emptyset
  \qquad
  \FVL {\scopel wxyz} = \{y,z\} \setminus \{w,x\}
  \qquad
  \FVL {\parl \pi {\pi'}} = \FVL \pi \cup \FVL {\pi'}
\end{align*}

%%% Local Variables:
%%% mode: latex
%%% TeX-master: "main"
%%% End:

%% file: fig-type-constants.tex
\begin{figure}[t!]
  \declrel{Types for constants}{$\typeof c = \typec T$}
  % 
  % \begin{align*}
  %   \typeof \unitk & = \TUnit
  %   \\
  %   % \typeof \reck &= \TForall{\alpha}{\kindt}{\TForall{\beta}{\kindt}{((\alpha\rightarrow\beta)
  %   %                 \rightarrow (\alpha\rightarrow\beta)) \rightarrow (\alpha\rightarrow\beta)}}
  %   % \\
  %   \typeof \forkk & = \TFun[] {(\TFun[] \TUnit \TUnit)} \TUnit
  %   \\
  %   \typeof \newk & = \TForall\alpha\kinds{\TPair\alpha{\TDual \alpha}}
  %   \\
  %   \typeof \receivek &= \TForall\alpha\kindt{\TForall\beta\kinds
  %                       {?\alpha. \beta \to \TPair\alpha\beta}}
  %   \\
  %   \typeof \sendk &= \TForall\alpha\kindt{\TForall\beta\kinds { \alpha \to {! \alpha}. \beta \to \beta}}
  %   \\
  %   \typeof {\appe \selectk {C_k}} & = \TForall{\overline{\alpha}}{\overline\kindp}
  %                                    {
  %                                    {\TForall \beta \kinds {\TFun[]
  %                                    {\TOut{(\TPApp \proto {\overline\alpha})}\beta}
  %                                    {\TFOut{\overline{T_k}}\beta}
  %                                    }}}\quad
  %                                    \text{if }\protocolsfull \text{ and } k\in I
  %   \\
  %   \typeof \waitk &= \TFun[]{\TEndW}{\TUnit}
  %   \\
  %   \typeof \terminatek &= \TFun[]{\TEndT}{\TUnit}
  % \end{align*}
  \begin{align*}
    % \typeof \reck &= \TForall{\alpha}{\kindt}{\TForall{\beta}{\kindt}{((\alpha\rightarrow\beta)
    %                 \rightarrow (\alpha\rightarrow\beta)) \rightarrow (\alpha\rightarrow\beta)}}
    % \\
    \typeof \unitk & = \TUnit
    &
    \typeof \forkk & = \TFun[] {(\TFun[] \TUnit \TUnit)} \TUnit
    \\
    \typeof \newk & = \TForall\alpha\kinds{\TPair\alpha{\TDual \alpha}}
    &
    \typeof \terminatek &= \TFun[]{\TEndT}{\TUnit}
    \\
    \typeof \receivek &= \TForall\alpha\kindt{\TForall\beta\kinds
                        {?\alpha. \beta \to \TPair\alpha\beta}}
    & 
    \typeof \waitk &= \TFun[]{\TEndW}{\TUnit}
    \\
    \typeof \sendk &= \TForall\alpha\kindt{\TForall\beta\kinds { \alpha \to {! \alpha}. \beta \to \beta}}
  \end{align*}
  $
    \typeof {\appe \selectk {C_k}} = \TForall{\overline{\alpha}}{\overline\kindp}
                                     {
                                     {\TForall \beta \kinds {\TFun[]
                                     {\TOut{(\TPApp \proto {\overline\alpha})}\beta}
                                     {\TFOut{\overline{T_k}}\beta}
                                     }}}
                                    \text{ if }\protocolsfull \text{ and } k\in I
  $
  \caption{Types for constants}
  \label{fig:type-constants}
\end{figure}

% \\
% \ESelect{C_i} &= \TBForall {(\overline{\alpha\colon\kind})(\tbind\beta\kinds}
% {\oplus(\TPApp P{\overline{\alpha}}). \beta \to  {\overline{!T_i^\pm}. \beta}}
% &
% \offerk &= \TBForall {(\overline{\alpha}\colon K)(\tbind\beta\kinds} {\&
% (\TPApp P {\overline{\alpha}}).\beta \to 
% \TVariant C {{\overline{?T^\pm}}.\beta}}
% \\ 
% & \text{if } \protocols

%%% Local Variables:
%%% mode: latex
%%% TeX-master: "main"
%%% End:

%% file: fig-alg-typing.tex
\begin{figure}[t!]
\declrel{Type synthesis}{$\expSynth\Gamma eT\Gamma$}
\begin{mathpar}
  \axiom{\rulenameexpConst}{\expSynth{\Gamma}{c}{\typeof{\termc c}}{\Gamma}}
  
  \infrule{\rulenameexpVar}{\typeAgainst{T}{\kindt}}
  {\expSynth{\Gamma,\ebind xT}{x}{\typec T}{\Gamma}}

  \infrule{\rulenameexpUVar}{\typeAgainst{T}{\kindt}}
  {\expSynth{\Gamma,\eubind xT}{x}{\typec T}{\Gamma,\eubind xT}}

  \infrule{\rulenameexpLetUnit}
  {\expAgainst{\Gamma_1}{e_1}{\TUnit}{\Gamma_2} \\\\
    \expSynth{\Gamma_2}{e_2}{\typec T}{\Gamma_3}
  }
  {\expSynth{\Gamma_1}{\letu {e_1}{e_2}}{\typec T}{\Gamma_3}}
  
  \infrule{\rulenameexpAbs}
  {\typeAgainst{T}{\kindt} \\ \Normal[+]T = \typec V \\\\
    \expSynth{\Gamma_1, \ebind x{V}}{e}{U}{\Gamma_2} \\ \termc{x}\not\in\Gamma_2}
  {\expSynth{\Gamma_1}{\abse xTe}{\TFun[]{V}{U}}{\Gamma_2}}

  \infrule{\rulenameexpApp}
  {\expSynth{\Gamma_1}{e_1}{\TFun[]{T}{U}}{\Gamma_2}\\
    \expAgainst{\Gamma_2}{e_2}{T}{\Gamma_3}} 
  {\expSynth{\Gamma_1}{\appe {e_1}{e_2}}{\typec U}{\Gamma_3}}
  \quad
  \infrule{\rulenameexpPair}{\expSynth{\Gamma_1}{e_1}{T}{\Gamma_2} \\
  \expSynth{\Gamma_2}{e_2}{U}{\Gamma_3}} 
  {\expSynth{\Gamma_1}{\paire {e_1}{e_2}}{\TPair {\typec T} {U}}{\Gamma_3}}

  \infrule{\rulenameexpRec}
  {\typeAgainst{T}{\kindt} \\
    \Normal[+]{\TFun[]TU}=\typec V \\
    \expAgainst{\Gamma, \eubind x{V}}{v}{{V}}{\Gamma, \eubind x{V}}}
  {\expSynth{\Gamma}{\rece x{{\TFun[]TU}}v}{V}{\Gamma}}

  \infrule{\rulenameexpTAbs}
  {\expSynth[\Delta, \typec{\alpha\colon \kind}]{\Gamma_1}{v}{\typec T}{\Gamma_2}} 
  {\expSynth{\Gamma_1}{\tabse \alpha\kind v}{\TForall \alpha\kind T}{\Gamma_2}}

  \infrule{\rulenameexpTApp}
  {\expSynth{\Gamma_1}{e}{\TForall \alpha\kind U}{\Gamma_2} \\ \typeAgainst{T}{\kind}} 
  {\expSynth{\Gamma_1}{\tappe e T}{\Normal[+]{\typec U\tsubs{T}{\alpha}}}{\Gamma_2}}

  % \infrule{\rulenameexpPair}{\expSynth{\Gamma_1}{e_1}{T}{\Gamma_2} \\
  %   \expSynth{\Gamma_2}{e_2}{U}{\Gamma_3}} 
  % {\expSynth{\Gamma_1}{\paire {e_1}{e_2}}{\TPair {\typec T} {U}}{\Gamma_3}}

  \infrule{\rulenameexpLet}
  {\expSynth{\Gamma_1}{e_1}{\TPair TU}{\Gamma_2}\\
    \expSynth{\Gamma_2, \ebind x{\Normal[+]T}, \ebind y{\Normal[+]U}}{e_2}{V}{\Gamma_3}\\
    \termc{x},\termc{y}\not\in\Gamma_3} 
  {\expSynth{\Gamma_1}{\lete xy{e_1}{e_2}}{\typec V}{\Gamma_3}}

  \infrule{\rulenameexpMatch}{
    %\Normal[+]{S}={\TIn{(\TPApp \proto {\overline U})}{S'}} \\
    \protocolsfull \\
    \expSynth {\Gamma_1} e {\TIn{(\TPApp \proto {\overline U})}{S}} {\Gamma_2} \\
    \forall i\in I: \\ \expSynth {\Gamma_2,\ebind{x_i}{\matchin}} {e_i} {V_i} {\Gamma_i} \\
    \termc{x_i}\not\in\Gamma_i \\
    \typec{V} =_\alpha \typec{V_i} \\
    % \isEquiv{\Gamma_3}{\Gamma_i}
    {\Gamma_3}=_\alpha{\Gamma_i}
  }{
    \expSynth {\Gamma_1} {\matchwithe eCxeiI} {V} {\Gamma_3}
  }
\end{mathpar}
\declrel{Type against}{$\expAgainst\Gamma eT\Gamma$}
\begin{mathpar}
  \infrule{\rulenameexpCheck}
  {\expSynth{\Gamma_1} e U {\Gamma_2} \\ \typec T =_\alpha\typec U}
  {\expAgainst{\Gamma_1} e T {\Gamma_2}}  
\end{mathpar}
\caption{Algorithmic typing rules for expressions}
\label{fig:alg-typing}
\end{figure}

%% file: fig-label-exps-abbrev.tex
\begin{figure}[t!]
  \declrel{Labelled transition system for expressions}{$\trans e \lambda e$}
  \begin{mathpar}
    \inferrule[Act-App]{}{\trans{\appe{(\abse{x}{T}{e})}{v}}{\beta}{\termc{e}\esubs{v}{x}}}
    
	\inferrule[Act-TApp]{}{\trans{\tappe{(\tabse{\alpha}{\kind}{v})}{T}}{\beta}{\termc{v}\tsubs{T}{\alpha}}}

	\inferrule[Act-Let]{}{\trans{\lete{x}{y}{{\paire uv}}{e}}{\beta}{\termc{e}\esubs{u}{x}\esubs{v}{y}}}

	\inferrule[Act-Let*]{}{\trans{\letu{\unitk}{e}}{\beta}{\termc{e}}}

    \inferrule[Act-Rec]{}{\trans{\appe{({\rece xTv})}
        u}{\beta}{\appe{(\termc{v}\esubs{\rece xTv}{x})}u}}

    \inferrule[Act-Fork]{}{\trans{\appe{\forkk}{v}}{\forkl{v}} \unitk}
	
	\inferrule[Act-New]{}{\trans{\tappe \newk T}{\newl xyT}{\paire xy}}

	\inferrule[Act-Rcv]{}{\trans{\appe{\tappe{\tappe \receivek T}{U}}{x}}{\receivel xv}{\paire{v}{x}}}
	
	\inferrule[Act-Send]{}{\trans{\appe{\appe{\tappe{\tappe \sendk T}{U}}{v}}{x}}{\sendl xv}{x}}
	
	\inferrule[Act-Match]{}{\trans{\matchwithe xCyeiI}{\receivel
        x{C_k}}{e_k\esubs x{y_k}}}

	\inferrule[Act-Sel]{}{\trans{\appe{\tappe{\selecte{C}}{\overline T}}{x}}{\sendl x{C}}{x}}
	
	\inferrule[Act-Wait]{}{\trans{\waite{x}}{\closel x}\unitk}

	\inferrule[Act-Term]{}{\trans{\closee{x}}{\openl x}{\unitk}}
  \end{mathpar}
  \caption{Labelled transition system for expressions (selected rules). The full version, including
  call-by-value structural rules, is presented in \cref{sec:appendix-exps-procs}, \cref{fig:lts-exps}}
  \label{fig:lts-exps-abbrev}
\end{figure}

%%% Local Variables:
%%% mode: latex
%%% TeX-master: "main"
%%% End:

%% file: fig-label-procs.tex
\begin{figure}[t!]
    \declrel{Labelled transition system for processes}{$\trans p \pi p$}
\begin{mathpar}
	% structural rules
	% \emph{Structural rules}\hfill\\
	\inferrule[Act-Session]{
      \trans{\termc{e_1}}{\sigma}{\termc{e_2}}}
      % \and \lambda\neq \forkl{v}, \newl xyT}
	{\trans{\expp{e_1}}{\sigma}{\expp{e_2}}}	

	\inferrule[Act-Beta]{
      \trans{\termc{e_1}}{\beta}{\termc{e_2}}}
      % \and \lambda\neq \forkl{v}, \newl xyT}
	{\trans{\expp{e_1}}{\tau}{\expp{e_2}}}	

    \inferrule[Act-Fork]{\termc{e_1}\lts{\forkl{v}}{\termc{e_2}}}
	{\expp{e_1}\lts{\tau}\parp{\expp{e_2}}{\expp{\appe v\unitk}}}
    
    \inferrule[Act-New]{\termc{e_1}\lts{\newl xyT}{\termc{e_2}}}
	{\expp{e_1}\lts{\tau}{\newp{x}{y}{\expp{e_2}}}}
	
	\inferrule[Act-JoinL]{\trans{p_1}{\pi_1}{q_1} \and
      \trans{p_2}{\pi_2}{q_2}}
	{\trans{\parp{p_1}{p_2}}{\parl{\pi_1}{\pi_2}}{\parp{q_1}{q_2}}} 
	
	\inferrule[Act-JoinR]{\trans{p_1}{\pi_1}{q_1} \and
      \trans{p_2}{\pi_2}{q_2}}
	{\trans{\parp{p_1}{p_2}}{\parl{\pi_2}{\pi_1}}{\parp{q_1}{q_2}}} 
	
	\inferrule[Act-Msg]{\trans{p_1}{\parl{\receivel xv}{\sendl yv}}{p_2}}
	{\trans{\newp xy{p_1}}{\tau}{\newp xy{p_2}}}

	\inferrule[Act-Bra]{\trans{p_1}{\parl{\receivel xC}{\sendl yC}}{p_2}}
	{\trans{\newp xy{p_1}}{\tau}{\newp xy{p_2}}}

	\inferrule[Act-Wait]{\trans{p_1}{\parl{\closel x}{\openl y}}{p_2}}
	{\trans{\newp xy{p_1}}{\tau}{p_2}}

	\inferrule[Act-ParL]{\trans{p_1}{\pi}{p_2}}
	{\trans{\parp{p_1}{q_1}}{\pi}{\parp{p_2}{q_1}}}
	
	\inferrule[Act-ParR]{\trans{q_1}{\pi}{q_2}}
	{\trans{\parp{p_1}{q_1}}{\pi}{\parp{p_1}{q_2}}}
	
	\inferrule[Act-Res]
    {\trans{p_1}{\pi}{p_2} \and \termc x,\termc y\not\in\FVL\pi}
	{\trans{\newp{x}{y}{p_1}}{\pi}{\newp{x}{y}{p_2}}}

	\inferrule[Act-OpenL]{\trans{p_1}{\sendl xa}{p_2} \and \termc x\neq \termc
      a,\termc b}
	{\trans{\newp{a}{b}{p_1}}{\scopel abxa}{p_2}} % FIX LTS MACRO

	\inferrule[Act-OpenR]{\trans{p_1}{\sendl xb}{p_2} \and \termc x\neq \termc
      a,\termc b}
	{\trans{\newp{a}{b}{p_1}}{\scopel abxb}{p_2}}

	\inferrule[Act-CloseL]{\trans{p_1}{\parl{\receivel xa}{\scopel abya}}{p_2}}
	{\trans{\newp{x}{y}{p_1}}{\tau}{\newp{x}{y}{\newp{a}{b}{p_2}}}}

	\inferrule[Act-CloseR]{\trans{p_1}{\parl{\receivel xb}{\scopel abyb}}{p_2}}
	{\trans{\newp{x}{y}{p_1}}{\tau}{\newp{x}{y}{\newp{a}{b}{p_2}}}}
\end{mathpar}
\caption{Labelled transition system for processes}
\label{fig:lts-procs}
\end{figure}

%%% Local Variables:
%%% mode: latex
%%% TeX-master: "main"
%%% End:

%% file: metatheory.tex
% \section{Metatheory}
% \label{sec:metatheory}

% This section introduces the main results of our language: preservation and
% progress possibly ending in a deadlock.

\input{fig-proc-formation}

We complete section with the main results of our language.
Towards this end we need to talk about context and process typing.
Type formation (\cref{fig:type-formation}) is lifted pointwise to contexts in 
judgement $\isCtx{\Gamma}{\kind}$.
%
% , defined by the rules below. Notice that the
% last premise in rule Ctx-Empty checks the invariant whereby contexts contain
% only normalized terms.
% %
% \begin{mathpar}
%      \infrule{Ctx-Empty}
%      {}
%      {\isCtx{\Empty}{\kind}}
%  
%      \infrule{Ctx-Ext}
%      {\isCtx{\Gamma}{\kind}\\ \typeAgainst{T}{\kind} \\ \typec T = \Normal[+]{T}}
%      {\isCtx{\Gamma, \ebind{x}{T}}{\kind}}
% \end{mathpar}
%
The rules for typing processes are in
\cref{fig:proc-formation}. The rule for threads requires the expression
$\termc e$ to consume all resources given in context $\Gamma$. Rule P-Par splits
the context in two parts and gives each to one of the parallel processes. The
rule is not algorithmic but we plan to type check expressions only; processes
are a runtime artefact. Rule P-New guesses a type $\typec T$ for channel end
$\termc x$, then channel end $\termc y$ gets the type $\TDual T$, both in normal
form. Again, this rule is not algorithmic and need not be.
The labelled transitions convey behaviour in their labels. Because we use an
algorithmic typing system with leftovers, we require a labelled transition
system for typing contexts to account for the impact of the behaviour of labels
in the contexts. The statement of preservation relies on the labelled transition
system for contexts, defined in \cref{sec:proofs} (\cref{fig:lts-ctx}).

\begin{theorem}[Preservation]\
  \label{thm:preservation}
  \begin{enumerate}
  \item\label{it:pres-synth} If $\trans{e_1}{\lambda}{e_2}$ and
    $\expSynth{\Gamma_1,\Gamma_2,\Gamma_3}{e_1}{T}{\Gamma_3}$ and
    $\transCtx{\Gamma_2}{\lambda}{\Gamma_2'}$ and 
    $\isCtx{\Gamma_1,\Gamma_2,\Gamma_3}{\kindt}$, then
    $\expSynth{\Gamma_1,\Gamma_2',\Gamma_3}{e_2}{T'}{\Gamma_3}$ and
    $\isEquiv{T'}{T}$.
  \item\label{it:pres-against} If $\trans{e_1}{\lambda}{e_2}$ and
    $\expAgainst{\Gamma_1,\Gamma_2,\Gamma_3}{e_1}{T}{\Gamma_3}$ and
    $\transCtx{\Gamma_2}{\lambda}{\Gamma_2'}$ and 
    $\isCtx{\Gamma_1,\Gamma_2,\Gamma_3}{\kindt}$, then
    $\expAgainst{\Gamma_1,\Gamma_2',\Gamma_3}{e_2}{T}{\Gamma_3}$.
  \item\label{it:pres-proc} If $\trans{p_1}{\pi}{p_2}$ and
    $\isProc[\Gamma_1,\Gamma_2]{p_1}$ and $\transCtx{\Gamma_2}{\pi}{\Gamma_2'}$
    and $\isCtx[\Empty]{\Gamma_1,\Gamma_2}{\kindt}$,
    then $\isProc[\Gamma_1,\Gamma'_2]{p_2}$.
  \end{enumerate}
\end{theorem}

%\paragraph{Progress}

To the below set of labels we call the \emph{progress set} for $\termc x,\termc y$:
\begin{equation*}
\{
% processes
\tau,\,
(\parl{\receivel xv}{\sendl yv}),\,
(\parl{\receivel xC}{\sendl yC}),\,
(\parl{\closel x}{\openl y}),\,
(\parl{\receivel xa}{\scopel abya}),\,
(\parl{\receivel xb}{\scopel abyb})
\}  
\end{equation*}

\emph{Completed processes} are of the form
$\termc{\expp\unitk \PAR \dots \PAR \expp\unitk}$.
A process is a \emph{deadlock} if, for all its subterms of the form $\newp xyp$
and all transitions $\trans{p}{\pi}{p'}$, we have $\pi$ not in the progress set
for channel $\termc x,\termc y$.
To identify a deadlocked process we find all its $\nu$ subterms, and for each
subterm we find all its transitions. Then, if all transitions fall outside the
progress set, the process is deadlocked.
When the kind context $\Delta$ is apparent from the context, we write
$\Gamma^{\kinds}$ to denote a typing context composed solely of session types,
that is, such that $\ebind xT\in\Gamma$ implies $\typeAgainst T \kinds$.

\begin{theorem}[Progress possibly leading to deadlock]\
  \label{thm:progress}
  \begin{enumerate}
  \item\label{item:prog-exps} If
    $\expSynth{\Gamma_1^{\kinds}}{e_1}{T}{\Gamma_2}$, then either $\termc{e_1}$
    is a value or $\trans{e_1}{\lambda}{e_2}$.
  \item  If $\expAgainst{\Gamma_1^{\kinds}}{e_1}{T}{\Gamma_2}$, then either $\termc{e_1}$
    is a value or $\trans{e_1}{\lambda}{e_2}$.
  \item\label{item:prog-procs} If $\isProc[\Gamma^\kinds]{p_1}$, then
    $\termc{p_1}$ is completed, $\trans{p_1}{\pi}{p_2}$, or $\termc{p_1}$ is
    a deadlock.
  \end{enumerate}
\end{theorem}
%

%%% Local Variables:
%%% mode: latex
%%% TeX-master: "main"
%%% End:

%% file: fig-proc-formation.tex
\begin{figure}[t!]
  \declrel{Typing rules for processes}{$\isProc p$}
  \begin{mathpar}
    \infrule{P-Exp}
    {\expAgainst[\Empty]{\Gamma} e {\TUnit} {\Empty}}
    {\isProc{\expp{e}}}

    \infrule{P-Par}
    {\isProc[\Gamma_1]{p_1} \and \isProc[\Gamma_2]{p_2} }
    {\isProc[\Gamma_1,\Gamma_2]{\parp{p_1}{p_2}} }

    \infrule{P-New}
    {\typeAgainst[\Empty]{T}{\kinds} \\ \isProc[\Gamma, \ebind x{\Normal[+]{T}} , \ebind y{\Normal[-]{T}}]{p}}
    {\isProc{\newp{x}{y}{p}} }
  \end{mathpar}
  \caption{Typing rules for processes}
  \label{fig:proc-formation}
\end{figure}

%%% Local Variables:
%%% mode: latex
%%% TeX-master: "main"
%%% End:

%% file: implementation.tex
\section{Implementation}
\label{sec:implementation}

The implementation of \algst consists of a type checker and an interpreter, both
written in Haskell. Aiming at supporting more realistic programs, the
implementation features pragmatic extensions, including extensions to the
formal system and a simple module system for name space management.

\paragraph{Type checking}
\label{sec:type-checking}

The syntax of the implementation closely matches
the examples in this paper. Some places need additional kind annotations
because 
the type system of the implementation is not restricted to linear types. Its kind structure distinguishes
between linear and unrestricted variants of the kinds $\kinds$ and $\kindt$. Subkinding
allows us to subsume unrestricted versions of types to linear ones, that is, $\kindtu < \kindtl$. This subkinding generates a 
simple subtyping relation that essentially allows us to provide an unrestricted value where a linear
one is expected. The protocol kind $\kindp$ does not split in two versions, but still any type can be lifted to a protocol
type. \citet{DBLP:journals/corr/abs-2106-06658} investigate the metatheory of a system with a
similar kind structure (without the kind $\kindp$).

The type checker implementation closely follows the bidirectional
rules given in the paper. It necessarily contains an implementation of the type
normalization algorithm shown in \cref{fig:normalisation}. This
algorithm is extended to cater for the type isomorphism $\forall
\alpha. (T \to U) \cong T \to \forall \alpha.U$ (if $\alpha$ not free in $T$).  In this way, type applications can
be placed more liberally as long as their sequence is preserved.

The normalization algorithm is invoked exactly as indicated by the
algorithmic typing rules in \cref{fig:alg-typing}. By design, the
rules keep the typing assumptions and the outcome of inference in
normal form to minimize the number of times normalization is invoked.

The type checker includes further checking rules to enable writing some lambda abstractions without
type annotations. These rules are adaptions of well-known rules
\cite{DBLP:journals/csur/DunfieldK21}. For example:
\begin{mathpar}
  \infrule{\rulenameexpAbsn}
  {\expAgainst{\Gamma_1, \ebind xT}{e}{U}{\Gamma_2} \\ \termc{x}\not\in\Gamma_2}
  {\expAgainst{\Gamma_1}{\absne xTe}{\TFun[]{T}{U}}{\Gamma_2}}

  \infrule{\rulenameexpAppn}
  { \expSynth{\Gamma_1}{e_2}{T}{\Gamma_2}\\
   \expAgainst{\Gamma_2}{e_1}{\TFun[]{T}{U}}{\Gamma_3}} 
  {\expAgainst{\Gamma_1}{\appe {e_1}{e_2}}{\typec U}{\Gamma_3}}
\end{mathpar}

The implementation fully supports the overloaded use of data
constructors as protocol operators. We can perform pattern matching against the constructors of
a \lstinline+List+ datatype and use the same constructors in selecting alternatives of the
\lstinline+!List+ protocol. Moreover, the \lstinline+case+ operator is also overloaded to serve as
the \lstinline+match+ operator.

% The operations \lstinline+wait+ and \lstinline+terminate+ are not supported in
% the implementation. Instead, an exhausted channel is an unrestricted value of
% type \lstinline+End+ which can be discarded freely as it is not necessary to
% explicitly close a connection.

\paragraph{Interpretation}
\label{sec:interpretation}

The interpreter employs
standard techniques for the functional part and for
representing values at run time. 
It uses one universal type 
with constructors for each type in the language.
Processes in \algst are mapped to
Haskell threads. A fork in the language is implemented
using \texttt{forkIO} in Haskell. The interpreter is invoked recursively
for the new thread.

Communication channels are implemented using shared-memory abstractions from
Haskell's \texttt{Control.Concurrent} library. Synchronous
communication (as in this paper) is implemented
by pairs of \texttt{MVar}s, a semaphore-like concurrent datastructure
\cite{DBLP:conf/popl/JonesGF96} that behaves like a buffer of size
one. The implementation has an option to switch to
asynchronous channels using bounded
queues (\texttt{TBQueue}) from the \texttt{stm} library.

As channels are implemented in shared memory, the 
\lstinline|send| primitive is naturally polymorphic. Internally, the interpreter mediates
values of the universal type mentioned above between \lstinline+send+ and
\lstinline+receive+ operations.

\label{sec:benchmarking}
\begin{figure}[tp]
\begin{lstlisting}[morekeywords={rec},moreemph={End}]
--- protocol and type in AlgST syntax ---
protocol Repeat x = More x (Repeat x) | Quit
?Repeat Int . !(Char, EndT) . EndT
--- corresponding type in FreeST syntax ---
(rec repeat0 : 1S . &{More : ?Int ; repeat0 ; Skip,
                      Quit : Skip}) ; (!(Char, End) ; End)
--- example of an equivalent AlgST type ---
Dual (!Repeat Int . ?(Char, EndT) . Dual EndT)
--- example of a non-equivalent AlgST type ---
?Repeat String . !(Char, EndT) . EndT
\end{lstlisting}
  \caption{An \algst type instance, its FreeST counterpart and examples of equivalent and non-equivalent \algst types, following the rules of our test suite generator.}
  \label{fig:generated-example}
\end{figure}

\paragraph{Benchmarking}

We substantiate our claim that replacing a worst-case superlinear algorithm for
type equivalence by a linear one yields actual run-time improvements with an experiment.
We compare \algst's type equivalence algorithm with
FreeST, a freely available~\cite{FreeST-Language} programming language
containing an implementation of type equivalence for polymorphic
context-free session types
\cite{DBLP:conf/tacas/AlmeidaMV20,DBLP:journals/corr/abs-2106-06658}.
% \footnote{Available from
%   \url{http://rss.di.fc.ul.pt/tools/freest/}.}
As the type checker regularly requires type equivalence checks, the performance
of the equivalence algorithm should be understood as a lower
bound of the type checker performance.

To create a collection of equivalent and non-equivalent test cases, we 
implemented a generator of instances of \algst types. 
An instance comprises a set of mutually recursive algebraic protocols and a session type 
referring to them. We carefully restrict protocols and types
so that a translation from \algst instances to FreeST types is
possible: the generator avoids polymorphic and nested recursion and restricts the
occurrences of the negation operator to the top level of protocol
constructor arguments. 
For each instance $\typec{T}$, we randomly apply the properties of normalization 
to generate 
an equivalent \algst type $\typec{T'}$ (see~\cref{fig:normalisation}). The pairs 
$(\typec{T}, \typec{T'})$ constitute our test suite of equivalent types. 
Non-equivalent tests are obtained from each $\typec{T}$ 
by either introducing an additional quantifier, or changing a sub-part of the
type to something of the same kind. Possible
replacements are the quantified type variables and the types built into AlgST,
such as \lstinline+Int+ or \lstinline+String+. The \algst type is
translated to a session type in FreeST. Protocols are translated inline at
every point of use as recursive branch or choice types, depending on wether it
appears in a sending or receiving context.
For single constructor types, the translation omits the constructor
tag. The arguments of the constructors are translated into nested
sequences of single interactions. See \cref{fig:generated-example} for an example.

Benchmarks were run using the gauge package v0.2.5 \cite{gauge} with a
timeout of 2~minutes. AlgST was compiled using GHC 9.2.7, FreeST used GHC 8.10.7.
Each test suite comprises 324 tests. FreeST incurred 69 timeouts for the positive tests and 77 timeouts for the negative tests.
\Cref{fig:timings} compares the execution time of the \algst and FreeST type equivalence algorithms 
for the equivalent and non-equivalent test cases. 
%(\cref{fig:time-positiv} and \ref{fig:time-negative}, resp.). 
% The plots are presented in log scale and vary with the number of nodes of the \algst type in the abstract syntax tree.
The plots substantiate the differences in the execution time of FreeST and \algst, distinguishing the 
linear evolution of \algst and 
the exponential behavior of FreeST as a function of the number of nodes in the abstract
syntax tree of the AlgST type. 
%The plots are represented in log scale.

\begin{figure}[tp]
  \begin{subfigure}{.5\textwidth}
    \centering
    \includegraphics[width=\linewidth]{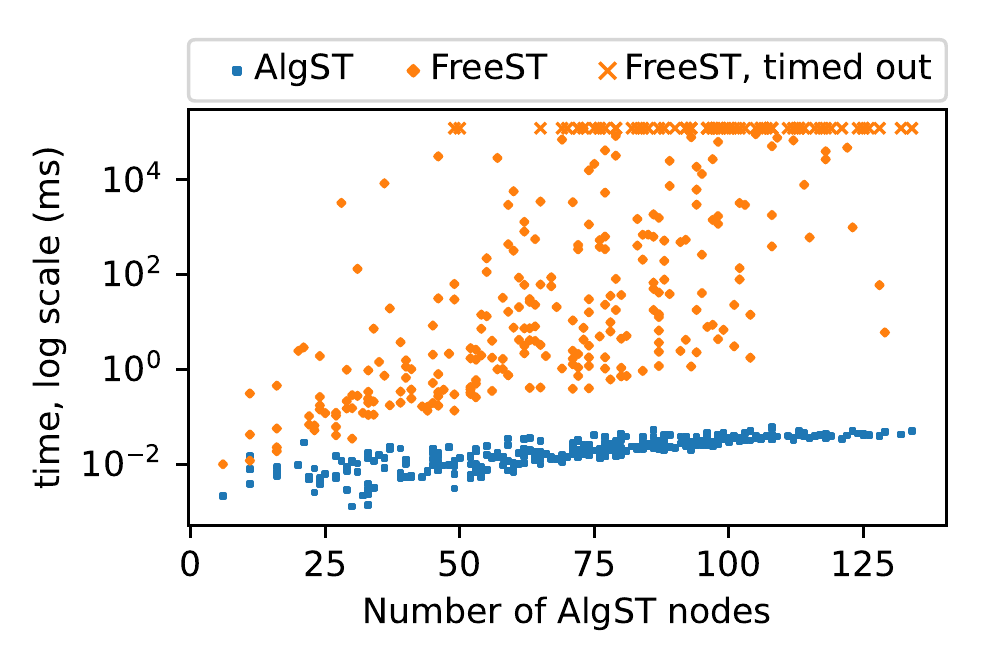}
    \caption{equivalent test cases}
    \label{fig:time-positiv}
  \end{subfigure}%
  \begin{subfigure}{.5\textwidth}
    \includegraphics[width=\linewidth]{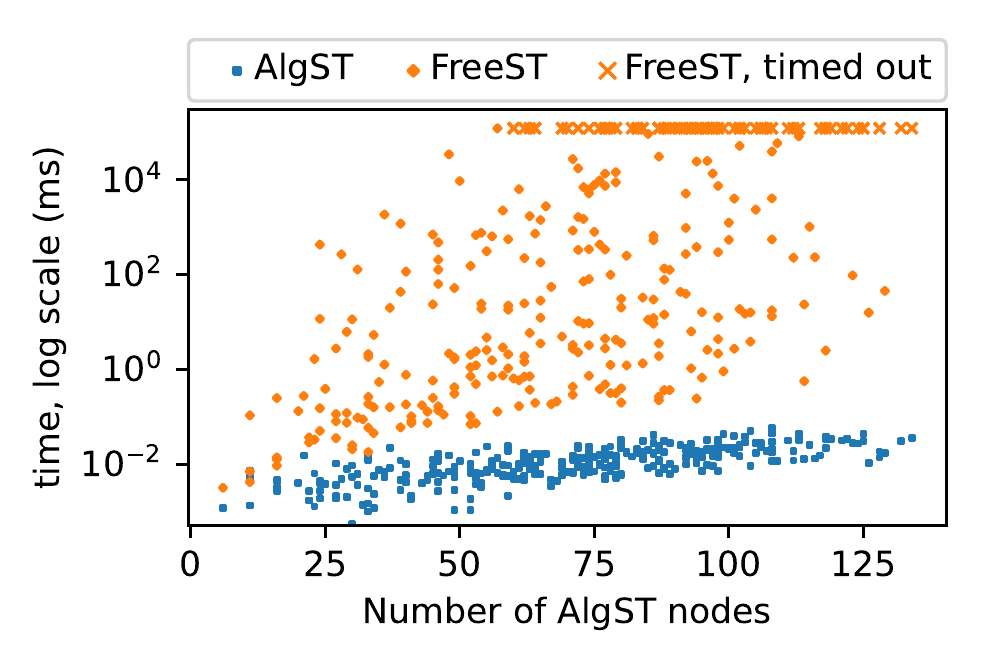}
    \caption{non-equivalent test cases}
    \label{fig:time-negative}
  \end{subfigure}
  \caption{Execution time of the \algst and FreeST type equivalence algorithms on the equivalent and non-equivalent test cases, represented in log scale as a function of the number of \algst nodes in the AST.}
  \label{fig:timings}
\end{figure}

%%% Local Variables:
%%% mode: latex
%%% TeX-master: "main"
%%% End:

%% file: related.tex
\section{Related Work}
\label{sec:related}

This survey of related work concentrates on binary session types and their
treatment of recursion. Many works in the area concentrate on different aspects
and treat recursion informally as an add-on. In particular, the early work of
\citet{DBLP:conf/parle/TakeuchiHK94,DBLP:conf/concur/Honda93} did not consider
recursive types.
The school of session types that takes its foundations from the
sequent calculus for dual intuitionistic logic
\cite{DBLP:conf/concur/CairesP10,DBLP:journals/mscs/CairesPT16}
initially focuses on finite behaviors. While
recursion would break the logical consistency of their system, they
consider dependently typed variants that include replication
\cite{DBLP:conf/cpp/PfenningCT11}. 

\smallskip
\emph{Session type systems with recursion.}
Generally, recursive types come in two flavors, equirecursive and
isorecursive \cite{DBLP:conf/lics/AbadiF96}, and the impact of
these different approaches on type checking is well-studied
\cite{DBLP:books/daglib/0005958}. They have been found to be equally 
expressive \cite{DBLP:journals/pacmpl/PatrignaniMD21}. Most
session type systems in the literature support equirecursive session
types restricted to tail recursion and treat them informally.

Recursive types have first been introduced to session types by
\citet{DBLP:conf/esop/HondaVK98} in the context of $\pi$-calculus. Their
system relies on equirecursion and is syntactically restricted to
tail recursion.
\citet{DBLP:conf/esop/GayH99,DBLP:journals/acta/GayH05} study
subtyping in the same setting using a coinductive definition. 
% They
% define an algorithm for subtyping and develop a sound type checker.
% Most subsequent work refers to this treatment and the algorithm
% extends to their setting.
%
\citet{DBLP:conf/ppdp/CastagnaDGP09} consider recursive session types
as infinite regular trees. Their types are tagless in the sense that
choices are made according to the set of accepted values and the viability
of the continuation.
The functional session type calculus of
\citet{DBLP:journals/jfp/GayV10} includes subtyping and equirecursive
types, but the latter are restricted to tail-recursive session types.
Fundamentals of session types \cite{DBLP:journals/iandc/Vasconcelos12}
proposes a range of $\pi$-calculus-based systems with increasing
complexity. One particular focus is the treatment of replication and unrestricted
channel ends, which requires recursive types. In this work, recursive
types are modeled by infinite regular trees, represented using the
familiar $\mu$ notation with contractive, tail recursive bodies.
% no subtyping

\citet{DBLP:conf/esop/ToninhoCP13} extend the linear-logic based line
of work and focus on the combination of functions and processes
including recursive processes and consequently recursive
types for processes and functions. The examples indicate the equirecursive variant.
Apparently, there is no restriction in the use of recursion and
sequencing is implemented using the tensor $\otimes$.
% Saving some space; also I do not quite understand the restriction
% However, linearity stops them to write a type
% like $\alpha\otimes\alpha$ when $\alpha$ is a session type:
% operationally, it would mean that a process spawns a copy of itself,
% which is not possible in their system (unless via replication, which
% is modeled with a different logical connective).
%
\citet{DBLP:conf/tgc/ToninhoCP14} cover a different aspect of
equirecursion. They consider corecursive protocol definitions and
analyze their productivity. The main question they pursue is when a
corecursive protocol produces the next value after finitely many
internal steps.
SILL \cite{DBLP:conf/fossacs/PfenningG15} is a process-based language design that
integrates linear and affine types as well as synchronous and
asynchronous communication. It builds directly on the linear logic interpretation of session
types. Its examples make use of polymorphism and equirecursive types in the form of tail
recursive type equations, neither of which are reflected in the type
theoretical presentation.
% Subsequent work on manifest sharing also
% takes polymorphism and equirecursion for granted \cite{DBLP:journals/pacmpl/BalzerP17}.
%
\citet{DBLP:conf/concur/BartolettiSZ14} offer a semantic analysis of
session types on the basis of labeled transition systems. Their
system features equirecursion in tail position.

% \begin{itemize}
% \item X equi, regular, pi-calculus \cite{DBLP:conf/esop/HondaVK98},
%   first with recursion?
% \item X equi, regular, pi-calculus, subtyping \cite{DBLP:conf/esop/GayH99}
% \item X equi, regular \cite{DBLP:journals/jfp/GayV10}
% \item X equi, regular, pi-calculus ``fundamentals''
%   \cite{DBLP:journals/iandc/Vasconcelos12}
% \item X equi, regular, process and functional
%   \cite{DBLP:conf/esop/ToninhoCP13}
% \item X equiv, regular, process \cite{DBLP:conf/tgc/ToninhoCP14}: this
%   is really corecursion, focuses on productivity where a process is
%   guaranteed to take finitely many internal steps and then engages in
%   the next communication
% \item equi, regular, pi-calculus \cite{DBLP:journals/corr/Dardha14}
%   (focus on encoding of recursive sessions in recursive linear pi
%   calculus, duality and recursion)
% \item X equi, regular, lts \cite{DBLP:conf/concur/BartolettiSZ14}
% \item X SILL: implicit use of regular equirecursion \cite{DBLP:conf/fossacs/PfenningG15}
% \item X iso(?), regular \cite{DBLP:conf/icfp/LindleyM16}
% \item X equi, regular \cite{DBLP:journals/pacmpl/BalzerP17}
% \item equi, regular \cite{lindley17:_light_funct_session_types}
% \item Padovani's FUSE with OCaml implementation. 
% \item equi, regular \cite{DBLP:conf/esop/VasconcelosCAM20}
% \item infinite regular trees \cite{DBLP:journals/pacmpl/CicconeP22}
% \item X defined by different kinds of systems of equations \cite{DBLP:conf/fossacs/GayPV22}
% \end{itemize}

\citet{DBLP:conf/icfp/LindleyM16} define two
functional session type systems with recursion, $\mu$GV and $\mu$CP, and show that
they are related by typing- and semantics-preserving translations. So
we only compare to $\mu$GV in the following. $\mu$GV is based on
a simply-typed linear lambda calculus extended with notions of recursion and
corecursion. These are defined as least and greatest (respectively) fixed points of
strictly positive functors --- analogous to a logically
meaningful version of algebraic (co-) datatypes --- with generic fold and
unfold operators. Finally, they lift this isorecursive treatment of
recursive data to the level of session types by encoding branching and
recursion, mapping the corresponding data functors over session
types and wrapping them in promises.
We conclude that $\mu$GV features session types with isorecursion and
that it could be used to encode the simply-typed fragment of \algst,
but without parameterized data types and without the over
direction. Where $\mu$GV uses promises (i.e., a separate channel) to encode the argument to the
fold of the recursive type (i.e., the constructor arguments of an
algebraic datatype), \algst transmits the body of the fold on the same
channel, thus avoiding the overhead of channel creation.
Moreover, $\mu$GV is mainly a theoretical contribution
whereas a full implementation of \algst is available.

\smallskip
\emph{Context-free and nested session type systems.}
\citet{DBLP:conf/icfp/ThiemannV16} proposed
context-free session types to enable the type-safe encoding of
non-regular protocols like the serialization of tree structures or the
stack protocol from the introduction without encurring the overhead of
session initiation. They changed the asymmetric structure of the
traditional session type syntax $\TIn TS$ and $\TOut TS$, where an
item of type $\typec T$ is communicated and the session continues with
$\typec S$,
to a symmetric structure where session types can be built from the
atomic transmissions, $\TGIn T$ and $\TGOut T$, choice and branching,
$\TChoice{S_1}{S_2}$ and $\TBranch{S_1}{S_2}$,  and a monoidal sequential
composition operator $\TSeq{S_1}{S_2}$ with unit $\TSkip$. 
% For example, here
% is the type for transmitting a binary tree
% (\cf \cref{sec:algebr-datatyp-as}) in their system:
% $ \typec{\mu X.} \TChoice{\TGOut\TInt}{(\TSeq XX)}$.
This design,
together with equirecursive sessions, leads to a number of technical
problems. First, types come with a non-trivial equational theory:
sequential composition forms a monoid with unit $\TSkip$ and it
distributes over choice and branching. Second, polymorphic recursion
is needed as a recursive call may have to be instantiated with
different continuation sessions. Third, the interaction between
equirecursion and the omission of the restriction to tail recursion
gives rise to a difficult, but decidable, type equivalence relation.
Subsequent work investigated tractable algorithms for type
equivalence
\cite{DBLP:journals/corr/abs-1904-01284,DBLP:conf/tacas/AlmeidaMV20}.
\citet{DBLP:journals/corr/abs-2106-06658} integrate context-free
session types with explicit polymorphism in the style of
System~F. Recent work extends this approach with higher-order types
\cite{DBLP:journals/corr/abs-2203-12877}.

\citet{DBLP:conf/esop/Padovani17,DBLP:journals/toplas/Padovani19}
proposes a different approach to context-free session types embodied
in a functional session calculus called \fuse. This calculus uses the
type language of context-free session types, but
sidesteps the problems with the sequencing operator and type
equivalence by introducing a resumption operator $\{e\}_c$ that
deconstructs a sequence of sessions on channel $c :
\TSeq{S_1}{S_2}$. It does so by retyping the channel to $c : \typec{S_1}$ for
evaluating $e$, which must return the same channel after finishing
$\typec{S_1}$.
More formally, the resumption obeys some kind of frame rule: if  $c :
\typec{S_1} \vdash e : \typec{T \times \TOne}$, then
$c : \TSeq{S_1}{S_2} \vdash \{e\}_c : \typec{T \times S_2}$. The
resumption temporarily removes the continuation
$\typec{S_2}$ from the session type and sticks it back on
afterwards. The type system has means to make sure that the channel returned by $e$ is
identical to $c$, which is required for soundness. The type $\TOne$
does not signify the end of the session, but rather that some known
part of the session has been processed and the session will be
continued in the context (of the resumption).
The management of the choice and branch operations is very close to
\algst. \fuse implements choice by transmitting the constructor tag in
a variant type that selects the different continuations. Its branch
operator matches on the received constructor and thus selects the
desired continuation. However, the handling of protocol constructors
in \algst includes the construction of the continuation type from the
types and directions given in the protocol declaration. \fuse has no
equivalent to a protocol declaration nor to the type operator
\lstinline+-+ to swap the direction of a transmission.
\citet{DBLP:journals/toplas/Padovani19} implements
\fuse in OCaml and demonstrates that context-free session types
are amenable to type inference, if the programmer helps with judicious
placement of resumptions. Support for equirecursion is already built into  OCaml's type inference
engine. It is also shown that subtyping is undecidable.

\citet{DBLP:conf/esop/DasDMP21} propose the metatheory of 
nested session types featuring parametrized definitions of session types
with nested polymorphic recursion. 
At first glance, the nested session type 
$\syntax{List}[\alpha] = \oplus\{\syntax{nil}\colon \syntax{end}, \syntax{cons}\colon \alpha \otimes \syntax{List}[\alpha]\}$
is very similar to the algebraic protocol
$\protocolk\ \TPApp{\syntax{List}}\alpha = \termc{\syntax{Nil}}\ |\
\termc{\syntax{Cons}}\ \alpha\ (\typec{\syntax{List}}\ \alpha)$, but
there are significant differences. On the
algebraic side, we can materialize the \typec{\syntax{List}} protocol in a type  
that sends or receives a list of integers, while the nested type $\syntax{List}$ is 
already committed to a concrete direction.
Although type nesting enables to capture the modularity characterizing 
algebraic sessions, there are several features that distinguish AlgST from 
nested session types:
Das et al. propose a structural and equirecursive, rather than a
nominal, view of type definitions, and interpret types coinductively, rather than inductively.
These features have significant impact on the complexity of the 
equivalence problem for nested session types and on its implementation:
Das et al. present a sound but incomplete algorithm for type equivalence.
By considering a nominal interpretation of algebraic session types, we reduced the 
complexity of the equivalence problem from 
doubly-exponential~\cite{DBLP:conf/esop/DasDMP21} to 
linear and designed a sound and complete algorithm to verify type
equivalence. The examples in section~2 of
\citet{DBLP:conf/esop/DasDMP21} translate into \algst and type check with our implementation.

\citet{DBLP:conf/fossacs/GayPV22} consider equirecursive protocols
defined by different kinds of systems of equations (with and without
parameters). Imposing different restrictions on the parameterization
yields different classes of protocols with decision problems for type
equivalence ranging from polytime to undecidable. Their largest
decidable class of pushdown session types is shown to be equivalent to
nested session types, which we can model with \algst.

Compared to all competiting systems \cite{DBLP:journals/pacmpl/ThiemannV20,DBLP:journals/corr/abs-2106-06658,DBLP:journals/toplas/Padovani19,DBLP:conf/esop/DasDMP21}, \algst's materialization and
directional operators are unique. These
operators enable the construction of protocols without committing to a
particular direction. 

%%% MPST:
%%% (equi, regular) Less is more: multiparty session types revisited

% \paragraph{Protocol composition}
% \label{sec:protocol-composition}

% \citet{DBLP:journals/corr/abs-2203-02461} define a notion of protocol
% composition which interleaves two or more protocols. To
% restrict this liberal approach to sensible interleavings, protocols
% are extended with an assertion language that specifies safe
% interleavings. Interleaving composition is complementary to our
% modular approach to protocol construction: we combine subprotocols
% rather than interleave them.

\input{related-lin-unrestricted}

%% file: related-lin-unrestricted.tex
\smallskip
\emph{Linear and unrestricted types.}
Our 
%formal 
calculus is mostly linear, except that we
%. The only exception we make is to
allow unrestricted bindings for recursive functions, which would be
rendered useless otherwise. 
%On the other hand, we want to speak about
%recursive protocols, so we felt like recursive functions are needed in
%the formal development. 
Our implementation goes well beyond that in
allowing fully fledged unrestricted computations besides ensuring that
session typed channels are treated linearly.
% That said, all related work is closer to our implementation than to
% the formal calculus.
%
There is some work on recursion and linearity in the context of lambda calculus
\cite{DBLP:journals/lisp/AlvesFFM10}
%. This work 
adding iterators on
linear natural numbers to the calculus, thereby avoiding the need for
a rule like {\rulenameexpUVar}.
%, which we need to bind the recursive function to a name.

System F\degree~\cite{DBLP:conf/tldi/MazurakZZ10}
extends System F with kinds to integrate polymorphism with linear and
unrestricted types. This work serves as a blueprint for much
subsequent work in this area including LFST and FreeST (see below), as well as our
implemented language.
They categorize all objects, including functions, into one-use or
many-use resources. This style can be traced back to 
\citet{walker:substructural-type-systems}. 
% One example in the paper applies the language to protocols built with
% linear capabilities, reminiscent of session types.
%
LFST \cite{lindley17:_light_funct_session_types} formalizes the
session type system of Links, which features a mix of linear and
unrestricted types, embedding a standard functional programming
language. The substructural behavior of types is specified by a
sophisticated kind structure.  The kinding used in our implementation
is inspired by LFST, though LFST has additional kinds to describe row
types that we do not support.
\citet{DBLP:journals/corr/abs-2106-06658}
consider an integration of context-free session types with
System~F. Their type system includes a kinding discipline that manages
linear and unrestricted versions of ordinary types and
message types (possible payloads for the send and receive
operations). Our implemented system extends the kind structure in a
similar way.

Linear Haskell (LH) \cite{DBLP:journals/pacmpl/BernardyBNJS18}
is a proposal to integrate linear types with stock functional
programming. While the work discussed so far characterizes resources,
LH supports the lollipop type $S \multimap T$ which classifies
functions that use their argument exactly once. Originally, LH forced
the programmer to allocate resources using continuations. This restriction
has been fixed in subsequent work \cite{DBLP:journals/pacmpl/SpiwackKBWE22}.
There is further work on integrating linear and dependent types (e.g.,
\cite{DBLP:conf/popl/KrishnaswamiPB15}) which relies on
bifurcating the language in two calculi connected by an adjunction.

%% file: conclusion.tex
\section{Conclusion}
\label{sec:conclusion}

Algebraic protocols combine naming and recursion with choice %/branch
and sequence type constructors known from session types. Like algebraic datatypes,
they offer parameterization and mutual recursion. Unlike previous work on
session types, algebraic
protocols do not commit to a particular direction of communication,
even if direction reversal can be specified at the type level. Algebraic protocols
and session types subsume the expressive power of context-free and nested
session type systems while reducing the complexity of type checking
from doubly-exponential to linear time.

%%% Local Variables:
%%% mode: latex
%%% TeX-master: "main"
%%% End:

%% file: artifact.tex
\section*{Software Availability}
\label{sec:artifact}

The proposed system and implementation are publicly available as an artifact \cite{artifact}. The interested reader can follow the steps in the artifact documentation to set up a pre-built docker image or a self-built docker image. Besides the implementation of the type checker and the interpreter, we provide programs for all the examples in the paper and examples that establish a comparison with related work \cite{DBLP:conf/esop/DasDMP21,DBLP:conf/icfp/ThiemannV16}. In the artifact, the reader can also find instructions on how to reproduce our benchmarking results for type equivalence.

%%% Local Variables:
%%% mode: latex
%%% TeX-master: "main"
%%% End:

%% file: appendix-faq.tex
\section{Frequently Asked Questions}
\label{sec:freq-asked-quest}

\subsection{Session types vs.\ protocol types}
\label{sec:what-diff-betw-1}
 A session type can be assigned to a channel end to describe
the entire communication behavior of this channel up to its closing. Thus, a
session type classifies a run-time value, namely a channel end.
A protocol type does not correspond to any run-time value.  It
describes composable behavior for bidirectional communication. A protocol
type becomes part of a session type by binding it to a particular
direction with the $!$ and $?$ operators.

In a sense, a protocol type describes a composable part of a
communication protocol or the difference between two session
types. This view is emphasized by our pervasive use of channel passing
style. 

\subsection{Negated recursion}
\label{sec:what-about-protocol}
It is possible to define an algebraic protocol that negates the polarity of its recursive use.
Such a protocol is fine, though unusual as it changes
direction with every unfolding of the recursion. Here is an example:
\begin{lstlisting}
protocol Flipper = Flipper -Int -Flipper
\end{lstlisting}
A server for \lstinline+Flipper+ can be implemented with a pair of
mutually recursive functions or even with just one function:
\begin{lstlisting}
flipper : !Flipper.EndT -> Unit
flipper c = let c = select Flipper [EndT] c in
            let (x, c) = receive [Int, ?Flipper.EndT] c in
	    match c with {
	      Flipper c -> send [Int, !Flipper.EndT] x c |> flipper }
\end{lstlisting}

\subsection{Mutual recursion}
\label{sec:do-you-support}
Mutually recursive protocols pose no problem. In general, the server
would be implemented by mutually recursive functions.
% In our recoding of the example in
% \cref{sec:what-about-protocol}, a single recursive function is sufficient:
% \begin{lstlisting}
% protocol Flip = Flip -Int Flop
% protocol Flop = Flop Int Flip

% flipflop : !Flip.End -> End
% flipflop c = let c = select Flip [End] c in
%             let (x, c) = receive [Int, !Flop.End] c in
% 	    select Flop [End] c
% 	    |> send [Int, !Flip.End] x []
%             |> flipflop
% \end{lstlisting}
\begin{lstlisting}
protocol Flip = Flip -Int Flop
protocol Flop = Flop Int Flip

flip : !Flip.EndT -> Unit
flip c = select Flip [EndT] c |> receive [Int, !Flop.EndT] |> flop

flop : (Int, !Flop.EndT) -> Unit
flop p = let (x, c) = p in
         select Flop [EndT] c |> send [Int, !Flip.EndT] x |> flip
\end{lstlisting}
\subsection{Duality vs.\ polarities}
\label{sec:what-diff-betw}
The \lstinline+Dual+ operator works on \emph{session types} and it inverts
the direction of communication from the outside. Intuitively, the
\lstinline+Dual+ operator traverses the spine of a session type and
flips all $!$ to $?$ and vice versa. It does \textbf{not} enter the
type of the message ``payload''. For example (with $\typec T$ the
payload type and $\typec S$ the continuation session):
\begin{align*}
  \TDual{(\TIn TS)} &\equiv \TOut T{\TDual S} &   \TDual{(\TOut TS)} &\equiv \TIn T{\TDual S}
\end{align*}
In contrast, the polarity operator \lstinline+-+ works on
\emph{protocol types} and inverts the direction of communication from
the inside. This inversion happens at the boundary where a protocol
type $\proto$ is turned into a payload type for a session:
\begin{align*}
  \TIn {(\TMinus P)}S &\equiv \TOut P S &\TOut {(\TMinus P)}S &\equiv \TIn P S
\end{align*}
The soundness of the type system guarantees that the direction of communication is
finally settled by the time we get to a primitive sending or receiving operation.

Like \lstinline+Dual+, the \lstinline+-+ operator is involutory:
\begin{align*}
  \TDual{(\TDual{S})} &\equiv \typec{S} &   \TMinus{(\TMinus T)} &\equiv \typec{T}
\end{align*}

\subsection{Recursion and duality}
\label{sec:what-about-inter}
\citet{BernardiH14,BernardiH16} discovered a problem in the interaction
between duality and recursive session types, which is exposed most
concisely by the recursive type $\typec{T =\mu X. !X.X}$.
The dual of this type is $\TDual{(\mu X. !X.X)} = \typec{\mu
  X.\TIn{(\mu X. !X.X)}X}$, which can be seen by dualizing the type's
expansion
\begin{equation*}
  \TDual{(\mu X. !X.X)} = \TDual{(!(\mu X.!X.X). (\mu X.!X.X))} \\
                        = \TIn{(\mu X.!X.X)}{\TDual{(\mu X.!X.X)}}
\end{equation*}
Using equations we could define the type $\typec T$ by $\typec{T = \TOut{T}T}$.
Translated to algebraic session types, the protocol underlying $\typec T$ is as follows:
\begin{lstlisting}
protocol X = Mu T X
type T = !X.EndT

selectMu : T -> !T.T
selectMu c = select Mu [EndT] c

dualT : Dual T -> ?X.EndT
dualT c = c

matchMu : Dual T -> ?T.(Dual T)
matchMu c' = match c' with { Mu c' -> c' }
\end{lstlisting}
This behavior matches exactly the outcome for $\typec T$ according to
\citet{BernardiH16}. The function \lstinline+selectMu+ performs the
unfolding of $\typec T$ whereas \lstinline+matchMu+ performs the
unfolding of $\TDual T$. The type of the identity function
\lstinline+dualT+ shows that \lstinline+Dual T = ?X.EndT+. 

\subsection{Efficiency of generic servers}
\label{sec:are-generic-servers}
In \cref{sec:toolb-gener-serv}, we explore an example with generic servers 
to showcase the expressivity of the proposed parametric protocols.
The example has a significant overhead as there are many tagging
operations compared to a manually crafted version.  

The present work is a first step and does not offer a solution
supported by formal theorems. But, let us call a protocol $\proto$
\emph{tagless} if it is non-recursive and declares a  single
constructor tag $\constr\ T_1\dots T_n$. In analogy to Haskell's
\texttt{newtype}, an implementation can reduce the
run-time cost of a tagless protocol to zero by omitting the select and
match operations. 

\subsection{Equirecursive vs. isorecursive session types}
\label{sec:this-session-types}
Most work on session types is using equirecursive types for two reasons. One is historic:
\citet{DBLP:conf/esop/GayH99,DBLP:journals/acta/GayH05} were the first
to study coinductively defined relations on (tail-) recursive session types. They defined the (by now)
standard notion of subtyping and showed that it is
decidable. Subsequent work has built on top of their results. The
other reason is that ``The equirecursive interpretation of a session type guarantees that no
message is required for the unfolding of a recursive type.'' (\citet{DBLP:journals/pacmpl/BalzerP17})

Algebraic protocols conflate the unfolding of the recursion with
the branching, just like algebraic datatypes do.
For protocols that have finite runs, algebraic protocol do not cause extra messages.
As such protocols use recursion guarded by branching, a
message to decide the branching is needed anyway.

The protocol \lstinline+Stream a+ in \cref{sec:param-prot} does cause
extra messages as it (roughly) corresponds to an isorecursive reading of the type $\mu X. !a.X$.

\subsection{Subtyping, type classes, type inference}
\label{sec:can-you-do}

Subtyping is undecidable for context-free session types
\cite{DBLP:journals/toplas/Padovani19}. However, as \algst's session
typing relies mostly on nominal types, there is scope for decidable
subtyping. We have a concrete proposal, but it is not implemented
and its metatheory has not been checked, yet.

The types \lstinline+Send a+ and \lstinline+Recv a+ for sending and
receiving data of type \lstinline+a+ in \cref{sec:transm-param-data}
look like types for a dictionary in a language with type classes or
traits
\cite{DBLP:journals/toplas/HallHJW96,steve18:_rust_progr_languag}. 
In future work we plan to explore the addition of implicit arguments as
in Scala \cite{DBLP:journals/pacmpl/OderskyBLBMS18}. 

Type inference is future work.

\subsection{Settling of direction}
\label{sec:settling-direction}
It might be confusing to see a type like $\TIn TS$ and $\TOut TS$, but
with the possibility that the direction changes later on. However,
this behavior is completely predictable. As long as $\typec T$ is a
protocol type variable  $\alpha : \kindp$, the direction can change
because $\alpha$ could be substituted by some protocol $\TMinus {T'}$. 
As soon as $\typec T$ is an ordinary type of kind $\kindt$ (even if
$\typec T$ is a variable), the direction is fixed. For example, in the
type of primitive send
\lstinline+forall(a:M). forall(s:S).  a -> !a.s -> s+, the payload type
\lstinline+a+ has kind $\kindm$, so the direction is fixed. \\

%%% Local Variables:
%%% mode: latex
%%% TeX-master: "main"
%%% End:

%% file: appendix-motivation.tex
\section{A Toolbox for Generic Servers}
\label{sec:toolb-gener-serv-1}

Analogously to the \lstinline+Stream+ protocol, we can define a protocol
and a generic server for finite repetion of a subsidiary protocol:
\begin{lstlisting}
protocol Repeat x = More x (Repeat x) | Quit

repeat : forall(p:P). Service p -> Service (Repeat p)
repeat [p] serveP [s] c = match c of {
  Quit c -> c,
  More c -> serveP [?Repeat p.s] c |> repeat [p] serveP [s]
}
\end{lstlisting}
This protocol may then be used similarly to  \lstinline+Stream+:
\begin{lstlisting}
repeatArith = repeat [Arith] serveArith
\end{lstlisting}
We can define an active version that offers the subprotocol $n$ times
and use that with \lstinline+Arith+, too:
\begin{lstlisting}
repeatActive : Int -> forall(p:P). Service p -> Service -(Repeat -p)
repeatActiveArith : Int -> Service -(Repeat -Arith)
\end{lstlisting}
While the implementation is straightforward, the astute reader might
ask why we provide the negative parameter to \lstinline+Service+ in
\lstinline+Service -(Repeat -p)+. In particular, could we not use the
\lstinline+Dual+ operator to switch the direction of the service?

To see the problem with this proposal, let's expand the two candidates:
\begin{lstlisting}
Dual (Service (Repeat -p)) = Dual (forall(s:S). ?Repeat -p.s -> s)
\end{lstlisting}
This proposal is doomed because the use of \lstinline+Dual+ on the
right side is not well-kinded, as neither the
universal type nor the function types are session types in our
system.
% Even applying \lstinline+Dual (?Repeat -p. s)+ results in
% \lstinline+!Repeat -p. Dual s+ with a dualized continuation.

\begin{lstlisting}
Service -(Repeat -p) = forall(s:S). ?- (Repeat -p). s -> s
                     = forall(s:S). !Repeat -p. s -> s
\end{lstlisting}
The negative parameter swaps the receiving operator in
\lstinline+Service+ to sending with surgical precision.

%%% Local Variables:
%%% mode: latex
%%% TeX-master: "main"
%%% End:

%% file: appendix-types.tex
\section{Proofs for Types}

\begin{example}
  We generalize the example from the introduction to a
  parametric stack protocol.
\begin{lstlisting}
protocol Stack a = Pop -a | Push a (Stack a) (Stack a)
\end{lstlisting}
  We set  $ \Delta = \types{Neg} : \kindp \to \kindp, \types{a} :
  \kindp$ and obtain
  the straightforward derivation
  \begin{mathpar}
    \inferrule*{
      \inferrule*
      {\inferrule*
      {\typeSynth{ \types{a} }{\kindp}}
      {\typeAgainst{ \types{a} }{\kindp}}}
      {\typeAgainst{ \types{-a} }{\kindp}} \\
      \inferrule*
      {\typeSynth{ \types{a} }{\kindp}}
      {\typeAgainst{ \types{a} }{\kindp}} \\
      \inferrule*
      {\inferrule*
        {\types{Stack} : {\kindp} \to \kindp \in \Delta \\
          \inferrule*
          {\typeSynth{\types{a}}{\kindp}}
          {\typeAgainst{\types{a}}{\kindp}}}
        {\typeSynth{ \types{Stack~a}}{\kindp}}}
      {\typeAgainst{ \types{Stack~a}}{\kindp}}
    }
    {
      \typeSynth[]{ \types{Stack} }{\kindp \to \kindp}
    }
  \end{mathpar}
  The three subderivations correspond to the arguments of the protocol
  constructors. We conclude that the definition of the protocol 
  \lstinline+Stack a+ is well-formed.
\end{example}

\begin{lemma}[Properties of the materialization and directional operators]\
  \label{lem:ops-dual}
  \begin{enumerate}
    \item\label{it:dual-norm} $\isConv{\TDual{(\tosession TS)}}{\tosession{\polarizedparam{-}{T}}{\TDual{S}}}\kinds$
    \item\label{it:plus-minus} $\polarizedparam{-}{\polarizedparam{+}{T}} = \polarizedparam{-}{T}$
    \item\label{it:minus-minus} $\polarizedparam{-}{\polarizedparam{-}{T}} = \polarizedparam{+}{T}$
  \end{enumerate}
\end{lemma}
\begin{proof}\
  \begin{enumerate}
  \item By case analysis on the definition of the
    $\tosessionop$ operator (\cref{fig:type-polarized-operator}).

    Case $\typec T = \pneg U $.
    %$\isConv{\TDual{(\tosession {-T}S)}}{\tosession{\polarizedparam{-}{-T}}{\TDual{S}}}\kinds$.
    Since $\tosession {\pneg U}S = \TIn{U}{S}$, rule \rulenametconvDualInM gives
    $\isConv{\TDual{(\tosession {\pneg U}S)}}{\TOut{U}{\TDual S}}\kinds$. Observing that
    $\polarizedparam{-}{-U}=\polarizedparam{+}{U}=\typec U$ we conclude that
    $\tosession{\polarizedparam{-}{\pneg U}}{\TDual{S}}=\TOut{U}{\TDual S}$.

    Case $\typec T \neq \pneg U $
    %$\isConv{\TDual{(\tosession {T}S)}}{\tosession{\polarizedparam{-}{T}}{\TDual{S}}}\kinds$.
    Since $\tosession {T}S = \TOut{T}{S}$, rule \rulenametconvDualOutM gives
    $\isConv{\TDual{(\tosession {T}S)}}{\TIn{T}{\TDual S}}\kinds$. Conclude observing
    that $\polarizedparam{-}{T}=\typec{-T}$.
    \item The proof follows by case analysis on the definition of the directional operators.

    Case $\typec T = \pneg U $. We have:
    $\polarizedparam{-}{\polarizedparam{+}{\pneg U}} = \polarizedparam{-}{\polarizedparam{-}{U}}
    = \polarizedparam{-}{\pneg{U}}$.

    Case $\typec T \neq \pneg U $. We have:
    $\polarizedparam{-}{\polarizedparam{+}{T}} = \polarizedparam{-}{T}$.

    \item The proof follows by case analysis on the definition of the directional operators.

    Case $\typec T = \pneg U $. We have:
    $\polarizedparam{-}{\polarizedparam{-}{\pneg U}} = \polarizedparam{-}{\polarizedparam{+}{U}}
    = \polarizedparam{-}{U} = \pneg U = \polarizedparam{+}{\pneg U}$.

    Case $\typec T \neq \pneg U $. We have:
    $\polarizedparam{-}{\polarizedparam{-}{T}} = \polarizedparam{-}{\pneg T}= \polarizedparam{+}{T}$.
  \end{enumerate}
\end{proof}

\begin{lemma}[Agreement for type conversion]
  \label{it:agreement-type-conv}
  If $\isConv T U \kind$, then $\typeAgainst T \kind$ and
  $\typeAgainst U \kind$.
\end{lemma}
\begin{proof}
  By rule induction on the hypothesis.
\end{proof}

We start by specifying the syntax of types in normal form. Normal forms for well-formed types
are $\typec{Q}$ as defined by the following grammar:
\begin{align*}
  \typec{Q} &\grmeq \typec{R} \grmor \pneg{R}
  \\
  \typec{R} &\grmeq
  % functional types
  \TUnit \grmor \TFun[]{R}{R} \grmor \TPair{R}{R} \grmor
  \TForall \alpha \kind {R} \grmor \typec\alpha
  % session types
  \grmor \TIn{R}{R} \grmor \TOut{R}{R}
  \grmor \TEndW \grmor \TEndT  \grmor
  \TDual{\alpha}
  % protocol and data types
  \grmor \TPApp \proto {\overline{Q}}
\end{align*}

\begin{proposition}[Kind preservation]
  ~\label{lemma:kind-preservation}
  \begin{enumerate}
  \item For all $\typeSynth T\kind$,  $\isConv{\Normal[+] T}T\kind$.
  \item For all $\typeSynth T\kinds$,  $\isConv{\Normal[-] T}{\TDual T}\kinds$.
  \end{enumerate}
\end{proposition}
\begin{proof}
  See Agda development \texttt{nf-sound+} and \texttt{nf-sound-}.
\end{proof}

\begin{proof}[Proof of \cref{thm:soundness-type-equiv} (Soundness)]
  \Cref{it:soundness-type-equiv-t}. Suppose that $\Normal[+] T =_\alpha \Normal[+] U$. From
  \cref{lemma:kind-preservation}, we have  $\isConv{\Normal[+]
    T}T\kind$ and  $\isConv{\Normal[+] U}U\kind$. Conclude by symmetry
  and transitivity of conversion as type equality is in $\equiv$.
  The proof for \cref{it:soundness-type-equiv-s} is analogous.
\end{proof}

\begin{proof}[Proof of \cref{thm:completeness-type-equiv} (Completeness)]
  See Agda development \texttt{nf-complete} and \texttt{nf-complete-}.
\end{proof}

To obtain a complexity bound for the computation of a normal form,
we exploit its syntactic form.
\begin{lemma}
  \label{lemma:syntax-of-type-normal-forms}
  If $\Normal[\pm]T = \typec{T'}$, then $\typec{T'}$ is described by $\typec{Q}$.
\end{lemma}
\begin{proof}
  See Agda development \texttt{nf-normal}.
\end{proof}
\begin{proposition}
  For each $\typeSynth T\kind$,
  the worst case running time of $\Normal[\pm]T$ is in $O(|\typec{T}|)$.
\end{proposition}
\begin{proof}
  A run of $\Normal T$ visits each of the $|\typec{T}|$ nodes in the input type
  once. At each node, it either reconstructs the type from the normal forms of
  the sub-terms in $O(1)$ time, or (in the transmission rules) it applies
  polarity correction to normal forms. By
  Lemma~\ref{lemma:syntax-of-type-normal-forms}, a normal form starts with at
  most one negative polarity, so correction takes at most two steps, at cost
  $O(1)$.
\end{proof}

\begin{proof}[Proof of \cref{thm:complexity} (Time complexity of type equivalence)]
  A corollary of \cref{lemma:syntax-of-type-normal-forms}.
\end{proof}

%%% Local Variables:
%%% mode: latex
%%% TeX-master: "main"
%%% End:

%% file: appendix-procs.tex
\section{Proofs for expressions and processes}
\label{sec:appendix-exps-procs}
\label{sec:proofs}

% In this section we present examples and the proofs of the main results 
% of our work, along with some auxiliary lemmas. We preserve the structure 
% of the paper, split into subsections, to facilitate the mapping between 
% the supplementary material and the main document

\input{fig-label-exps}

\begin{lemma}\
  \label{lem:synth-against-nf}
  \begin{enumerate}
       \item\label{it:synth-nf} If $\expSynth{\Gamma_1} eT{\Gamma_2}$ and $\isCtx[\Delta]{\Gamma_1}{\kindt}$, 
       then $\isCtx[\Delta]{\Gamma_2}{\kindt}$ and $\typec T$ is in normal form.
       \item\label{it:against-nf}  If $\expAgainst{\Gamma_1} eT{\Gamma_2}$ and $\isCtx[\Delta]{\Gamma_1}{\kindt}$
       and $\typec T$ is in normal form, then $\isCtx[\Delta]{\Gamma_2}{\kindt}$.
  \end{enumerate}
  \begin{proof}\
  \begin{enumerate}
       \item By rule induction on the first hypothesis.
       \item Using \cref{it:synth-nf} and observing that a type $\typec U$ such that 
       $\typec U =_\alpha \typec T$ is also in normal form.
  \end{enumerate}         
  \end{proof}
\end{lemma}

%% section 5
% \subsection{Metatheory}

\input{fig-label-ctx}

% \paragraph{Preliminaries} We start with some basic results.

\begin{lemma}[Properties of normalization]\
  \label{lem:norm-st}
  % Let $\typeSynth U \kind$.
  \begin{enumerate}
  \item\label{it:anti-plus} If $\typeSynth{\Normal[+] T} \kind$, then
    $\typeAgainst T \kind$.
  \item\label{it:anti-minus} If $\typeSynth{\Normal[-] T} \kind$, then
    % $\kind = \kinds$ and 
    $\typeAgainst T \kinds$.
  \item\label{it:anti-tosession} If $\typeSynth{\tosession TS} \kind$, then
    % $\kind = \kinds$ and 
    $\typeAgainst T \kindp$ and $\typeAgainst S \kinds$.
  \item\label{it:anti-pm} If $\typeSynth{\polarizedparam \pm T} \kind$,
    then $\typeAgainst T \kind$.
  \end{enumerate}
\end{lemma}
\begin{proof}
  The proof is by mutual induction on the various hypotheses. Most cases follow
  by a straightforward induction. We detail the case for
  $\Normal[-]{\typec{\TIn TS}} = \TFOut{\Normal[+]T}{\Normal[-]S}$.
  \Cref{it:anti-tosession} gives
  $\typeAgainst {\polarizedparam +{\Normal[+]T}} \kindp$ and
  $\typeAgainst {\Normal[-]S} \kinds$. Rule \rulenametypeSub gives
  $\typeSynth {\polarizedparam +{\Normal[+]T}} \kindp$ and
  $\typeSynth {\Normal[-]S} \kinds$. \Cref{it:anti-pm} gives
  $\typeSynth {\Normal[+]T} \kindp$. Again rule \rulenametypeSub gives
  $\typeAgainst {\Normal[+]T} \kindp$. \Cref{it:anti-plus} gives
  $\typeSynth T \kindp$. Once again, rule \rulenametypeSub gives
  $\typeAgainst T \kindp$. From $\typeAgainst {\Normal[-]S} \kinds$, induction
  gives $\typeSynth S \kinds$ and \rulenametypeSub gives
  $\typeAgainst S \kinds$. Finally, rule \rulenametypeDotIn gives
  $\typeSynth{\TIn TS}\kinds$ and rule \rulenametypeSub yields
  $\typeAgainst{\TIn TS}\kinds$.
\end{proof}

\begin{lemma}[Substitution]
  \label{lem:substitution}\
    \begin{enumerate}
    \item\label{it:substitution-tt} If $\typeSynth[\Delta,\tbind\alpha\kind]{T}{\kind'}$ and $\typeAgainst{U}{\kind}$, then
      $\typeSynth{T\tsubs U\alpha}{\kind'}$.
    \item\label{it:substitution-et} If $\expSynth[\Delta,\tbind\alpha\kind]{\Gamma_1}{e}{T}{\Gamma_2}$ and $\typeAgainst{U}{\kind}$,
      then $\expSynth{\Gamma_1}{e\tsubs{U}{\alpha}}{\Normal[+]{T\tsubs{U}{\alpha}}}{\Gamma_2}$.
    \item\label{it:substitution-ee} If $\expSynth{\Gamma_1,\ebind xU}{e}{T}{\Gamma_2}$ and
      $\expAgainst{\Gamma_3}{v}{U}{\Empty}$, then
      $\expSynth{\Gamma_1,\Gamma_3}{e\esubs vx}{T}{\Gamma_2}$.
    \item\label{it:substitution-ee2} If $\expSynth{\Gamma_1,\eubind x{U}}{e}{T}{\Gamma_2,\eubind x{U}}$ and
      $\expAgainst{\Gamma_3}{v}{U}{\Empty}$, then
      $\expSynth{\Gamma_1,\Gamma_3}{e\esubs vx}{T}{\Gamma_2}$.
    \end{enumerate}
  \end{lemma}
  \begin{proof}
    The proof follows by mutual induction on the first hypothesis.
  \end{proof}

\begin{lemma}[Agreement]
  \label{lem:agreement}
  Let $\isCtx{\Gamma_1}{\kindt}$.
  \begin{enumerate}
  \item\label{it:agreement-expL} If $\expSynth{\Gamma_1,\ebind xT}{e}{U}{\Gamma_2}$, then $\typeAgainst{T}{\kind}$.
  \item \label{it:agreement-expR} If $\expSynth{\Gamma_1}{e}{T}{\Gamma_2}$, then $\typeAgainst{T}{\kind}$.
  \item\label{it:agreement-proc} If $\isProc[\Gamma_1,\ebind xT] p$, then $\typeAgainst[\Empty]{T}{\kind}$.
\end{enumerate}
\end{lemma}
\begin{proof}\
  \begin{enumerate}
    \item By rule induction on the first hypothesis.
    \item By rule induction on the first hypothesis, using~\cref{lem:substitution} and \cref{it:agreement-type-conv}.
    \item By rule induction on the first hypothesis, using \rulenameexpCheck and~\cref{it:agreement-expL}.
  \end{enumerate}
\end{proof}

Monotonicity tells us that resources are not created during the process of type
checking.

\begin{lemma}[Monotonicity]
  \label{lem:monotonicity}
  If $\expSynth{\Gamma_1} vT {\Gamma_2}$, then $\Gamma_2\subseteq\Gamma_1$.
\end{lemma}
\begin{proof}
  A straightforward rule induction on the hypothesis. The base cases are E-Const
  and E-Var. The cases for E-Abs, E-TAbs, E-Pair and E-TApp follow by induction
  and transitivity of the subset relation.
\end{proof}

\begin{lemma}[Strengthening]\
  \label{lem:strengthening}
  \begin{enumerate}
  \item If $\typeSynth[\Delta, \tbind{\alpha}{\kind}]{T}{\kind'}$ and $\typec \alpha\not\in \FVT{\typec T}$,
    then $\typeSynth[\Delta]{T}{\kind'}$.
  \item
    If $\expSynth{\Gamma_1,\ebind xU}{e}{T}{\Gamma_2,\ebind xU}$, then
    $\expSynth{\Gamma_1}{e}{T}{\Gamma_2}$.
  % \item If $\expAgainst{\Gamma_1,\ebind xU}{e}{T}{\Gamma_2,\ebind xU}$, then
  %   $\expAgainst{\Gamma_1}{e}{T}{\Gamma_2}$.
  \end{enumerate}
\end{lemma}
\begin{proof}\
\begin{enumerate}
  \item The proof is by induction hypothesis on the first hypothesis, using \rulenametypeSub.
  \item The proof is by induction hypothesis on the first hypothesis.
\end{enumerate}  
\end{proof}

\begin{lemma}[Weakening]\
  \label{lem:weakening}
  % \begin{enumerate}
  % \item
    If $\expSynth{\Gamma_1}{e}{T}{\Gamma_2}$, then
    $\expSynth{\Gamma_1,\ebind xU}{e}{T}{\Gamma_2,\ebind xU}$.
  % \end{enumerate}
\end{lemma}
\begin{proof}
  The proof follows by induction on the hypothesis.
\end{proof}

% \begin{lemma}[Deprecated; a corollary of \cref{lem:strengthening}; we may not
%   need it, state it direclty in proofs]\
%   \begin{enumerate}
%   \item If $\expSynth{\Gamma_1}{v}{T}{\Gamma_2}$, then
%     $\Gamma_1 = \Gamma_2,\Gamma_3$ and $\expSynth{\Gamma_3}{v}{T}{\Empty}$.
%   \item If $\expAgainst{\Gamma_1}{v}{T}{\Gamma_2}$, then
%     $\Gamma_1 = \Gamma_2,\Gamma_3$ and $\expAgainst{\Gamma_3}{v}{T}{\Empty}$.
%   \end{enumerate}
% \end{lemma}

\paragraph{Preservation and progress} We now prove the main results.

\begin{proof}[Proof of \cref{thm:preservation} (Preservation)]\

  The proof follows by mutual rule induction on the first hypothesis. 
  \begin{enumerate}
    \item 
    Rule Act-App. In this case, $\trans{\appe{(\abse{x}{U}{e})}{v}}{\beta}{\termc{e}\esubs{v}{x}}$
    and $\Gamma_2 = \Gamma_2' = \Empty$. Inversion of \rulenameexpApp on 
    $\expSynth{\Gamma_1,\Gamma_3}{\appe{(\abse{x}{U}{e})}{v}}{\typec T}{\Gamma_3}$
    yields (a) $\expSynth{\Gamma_1,\Gamma_3}{\abse{x}{U}{e}}{\TFun[]{V}{T}}{\Gamma'}$
    and
    (b) $\expAgainst{\Gamma'}vV{\Gamma_3}$. By monotonicity (\cref*{lem:monotonicity}), we 
    know that $\Gamma_3\subseteq \Gamma'$; say that $\Gamma' = \Gamma_1',\Gamma_3$.
    Applying inversion of \rulenameexpAbs on (a) and
    strengthening (\cref{lem:strengthening}) on (b), we get 
    $\expSynth{\Gamma_1,\Gamma_3,\ebind x{\Normal[+]U}}{e}{T}{\Gamma_1',\Gamma_3}$ and
    $\expAgainst{\Gamma_1'}{v}{V}{\Empty}$ (and $\typec V = \Normal[+] U$). By \rulenametconvReflex and
    agreement (\cref{lem:agreement}) on the second hypothesis, we know that $\isEquiv TT$.
    Conclude applying the substituion lemma 
    (\cref{lem:substitution},~\cref{it:substitution-ee}) and weakening (\cref{lem:weakening}).
    
    Rule Act-TApp. Again, $\Gamma_2 = \Gamma_2' = \Empty$. Apply inversion  of \rulenameexpTApp on the 
    second hypothesis, followed by inversion of \rulenameexpTAbs and conclude 
    with~\cref{lem:substitution},~\cref{it:substitution-et}, using $\typec {T'} = \typec T$.
    
    Rule Act-Let. Apply inversion of \rulenameexpLet to the second hypothesis, followed by inversion of \rulenameexpPair,
    then apply substitution (\cref{lem:substitution},~\cref{it:substitution-ee}) and conclude using $\typec {T'} = \typec T$.
    
    Rule Act-Let*. In this case, $\trans{\letu{\unitk}{e}}{\beta}{\termc{e}}$ and $\Gamma_2 = \Gamma_2' = \Empty$.
    Using inversion of \rulenameexpConst and \rulenameexpCheck for $\unitk$, we know that premises 
    to \rulenameexpLetUnit include 
    $\expSynth{\Gamma_1,\Gamma_3}{e}{\typec T}{\Gamma_3}$, which is what we want to prove, since $\isEquiv TT$.
    
    Rule Act-Rec. Again, $\Gamma_2 = \Gamma_2' = \Empty$. Premises to rule \rulenameexpApp 
    are $\expSynth{\Gamma_1,\Gamma_3}{{\rece xVv}}{\TFun[]{U}{T}}{\Gamma_1,\Gamma_3}$
    and $\expAgainst{\Gamma_1,\Gamma_3}uU{\Gamma_3}$. Applying inversion of \rulenameexpRec on the former, conclude 
    that $\Normal[+]V= \TFun[]{U}{T}$ and then apply inversion of \rulenameexpCheck and the substitution lemma (\cref{lem:substitution}) to 
    get $\expSynth{\Gamma_1,\Gamma_3}{v\esubs{\rece xVv}x}{\TFun[]{U}{T}}{\Gamma_1,\Gamma_3}$. Conclude
    applying \rulenameexpApp and using the fact that $\isEquiv TT$.

    Rule Act-Fork. %By L-Fork, we know that $\Gamma_2$ is such that 
    %$\expAgainst[\Empty]{\Gamma_2}{v}{\TFun[]\TUnit\TUnit}{\Empty}$. 
    Inversion of \rulenameexpApp
    yields $\expSynth{\Gamma_1,\Gamma_2,\Gamma_3}{\forkk}{\TFun[]UT}{\Gamma'}$ and
    $\expAgainst{\Gamma'}vU{\Gamma_3}$. Applying \rulenameexpConst on the former we know that
    $\Gamma' = \Gamma_1,\Gamma_2,\Gamma_3$, $\typec U= \TFun[]{\TUnit}{\TUnit}$ and $\typec T=\TUnit$.
    By L-Fork we know that $\expAgainst{\Gamma_2}v{\TFun[]{\TUnit}{\TUnit}}{\Empty}$
    and $\Gamma_1 = \Empty$ in this case. By applying \rulenameexpConst for $\unitk$
    we conclude that $\expSynth{\Gamma_3}\unitk\TUnit{\Gamma_3}$, as we wanted.

    Rule Act-New. In this case, $\trans{\tappe \newk U}{\newl xyU}{\paire xy}$.
    By L-New, we know that $\Gamma_2 = \Empty$ and $\Gamma_2' = \ebind{x}{\Normal[+]{U}}, \ebind{y}{\Normal[-]{U}}$.
    Inversion of \rulenameexpTApp and \rulenameexpConst on the second hypothesis yields
    $\typec T=\TPair{\Normal[+]U}{\Normal[-]{U}}$ and $\Gamma_1 = \Empty$. Using 
    \rulenameexpVar and \rulenameexpPair we conclude that 
    $\expSynth{\Gamma_3,\ebind{x}{\Normal[+]{U}},\ebind{y}{\Normal[-]{U}}}{\paire{x}{y}}{\TPair{\Normal[+]{U}}{\Normal[-]{U}}}{\Gamma_3}$.
    
    Rule Act-Rcv. In this case, $\trans{\appe{\tappe{\tappe \receivek U}{V}}{x}}{\receivel xv}{\paire{v}{x}}$.
    By L-Rcv, $\Gamma_2 = \ebind{x}{\TIn{U}{S}}$ and $\Gamma_2' = \ebind{x}{S}, \Gamma$, 
    where $\Gamma$ is such that $\expAgainst[\Empty]{\Gamma}{v}{U}{\Empty}$. 
    Inversion of \rulenameexpApp on the second hypothesis yields 
    $\expSynth{\Gamma_1,\ebind{x}{\TIn{U}{S}}, \Gamma_3}{\tappe{\tappe \receivek U}{V}}{\TFun[]{W}{T}}{\Gamma_1'}$ and 
    $\expAgainst{\Gamma_1'}{x}{W}{\Gamma_3}$. Inversion of \rulenameexpVar on the latter
    yields $\Gamma_1' = \ebind{x}{\TIn{U}{S}}, \Gamma_3$ and $\typec{W}=\TIn{U}{S}$.
    Inversion of \rulenameexpTApp applied twice, followed by \rulenameexpConst enables us to 
    conclude that $\Gamma_1 = \Empty$ and $\typec{T}=\TPair{U}{S}$. Using L-Rcv, \rulenameexpCheck and
    weakening we know that $\expSynth{\ebind{x}{S},\Gamma,\Gamma_3}vU{\ebind{x}{S},\Gamma_3}$, using \rulenameexpVar
    and weakening we get $\expSynth{\ebind{x}{S},\Gamma_3}xS{\Gamma_3}$. Conclude applying \rulenameexpPair and
    recalling that $\isEquiv{\TPair{U}{S}\,\blk{=}\,T}{T}$.

    Rule Act-Send. In this case, $\trans{\appe{\appe{\tappe{\tappe \sendk U}{S}}{v}}{x}}{\sendl xv}{x}$.
    By L-Send, $\Gamma_2 = \Gamma, \ebind{x}{\TOut{U}{S}}$ where $\Gamma$ is 
    such that $\expAgainst[\Empty] \Gamma v U \Empty$. Proceed similarly to the previous case, 
    applying inversion of \rulenameexpApp twice, followed by \rulenameexpTApp and then 
    inversion of \rulenameexpConst over $\sendk$ to conclude that $\typec S=\typec T$. On the other hand, invert 
    \rulenameexpVar and L-Send with weakening applied to the other premises to conclude that 
    $\Gamma_1 = \Empty$. 
    The result follows applying \rulenameexpVar and weakening with $\Gamma_2' = \ebind{x}{T}$,
    and knowing that $\isEquiv TT$.

    Rule Act-Match. In this case, $\trans{\matchwithe xCyeiI}{\receivel
    x{C_k}}{e_k\esubs x{y_k}}$. By L-Match, we have $\Gamma_2 = \ebind{x}{\TIn{(\TPApp P{\overline{U}})}{S'}}$
    for $\protocolsfull$. Premises to \rulenameexpMatch include
    (a) $\expSynth {\Gamma_1,\ebind{x}{\TIn{(\TPApp P{\overline{U}})}{S'}}, \Gamma_3} x {\TIn{(\TPApp P{\overline{U}})}{S'}} {\Gamma_1,\Gamma_3}$,
    (b) $\expSynth {\Gamma_1,\Gamma_3,\ebind{y_i}{\matchin}} {e_i} {V_i} {\Gamma_3}$ and
    (c) ${T}=_\alpha{V_i}$.
    By L-Match we know that $\Gamma_2' = \ebind{x}{\matchin[k]}$.
    Applying \rulenameexpVar and \rulenameexpCheck we get (d) $\expAgainst{\ebind{x}{\matchin[k]}}x{\matchin[k]}{\Empty}$.
    Conclude applying the substitution lemma (\cref*{lem:substitution}) to (b) and (d). 
  %   {
  %   \expSynth {\Gamma_1} e {S} {\Gamma_2} \\
  %   \Normal[+]{S}={\TIn{(\TPApp P {\overline U})}{S'}} \\
  %   \protocolsfull \\
  %   \expSynth {\Gamma_2,\ebind{x_i}{\matchin}} {e_i} {V_i} {\Gamma_i} \\
  %   \termc{x_i}\not\in\Gamma_i\\
  %   \isEquiv{V_k}{V_i} \\
  %   \isEquivCtx{\Gamma_k}{\Gamma_i}%  \\
  %   % k \in I
  % }

    Rule Act-Sel. We have $\trans{\appe{\tappe{\selecte{C}}{\overline U}}{x}}{\sendl x{C}}{x}$.
    Using L-Sel, we know that
    $\Gamma_2 = \ebind{x}{\TOut{(\TPApp P{\overline{U}})}{S'}}$ for $\protocolsfull$.
    Use inversion of \rulenameexpApp and \rulenameexpTApp on the second hypothesis to conclude that 
    $\typec T = \TFOut{\overline{T_k}\tsubs{\overline U}{\overline\alpha}}{S'}$. Using inversion of 
    \rulenameexpVar conclude that $\Gamma_1 = \Empty$. Finally, use \rulenameexpVar
    and weakening (\cref{lem:weakening}) to get 
    $\expSynth{\ebind{x}{\matchout[k]},\Gamma_3}x{\TFOut{\overline{T_k}\tsubs{\overline U}{\overline\alpha}}{S'}}{\Gamma_3}$.
    Conclude observing that 
    $\isEquiv{\TFOut{\overline{T_k}\tsubs{\overline U}{\overline\alpha}}{S'} \,\blk{=}\, \typec T}{T}$.

    Rule Act-Wait. In this case, $\Gamma_2 = \ebind{x}{\TEndW}$ and $\Gamma_2' = \Empty$.
    Inversion of \rulenameexpApp on the second hypothesis yields 
    $\expSynth{\Gamma_1,\ebind{x}{\TEndW},\Gamma_3}{\waitk}{\TFun[] UT}{\Gamma'}$
    and $\expAgainst{\Gamma'}xU{\Gamma_3}$. Apply \rulenameexpConst to the former 
    and \rulenameexpVar to the latter to conclude that $\Gamma_1=\Empty$, $\typec{U} = \TEndW$
    and $\typec{T}=\TUnit$. The result follows using \rulenameexpConst on $\unitk$ and weakening (\cref*{lem:weakening}),
    using $\typec{T'} = \typec   T$.

    Rule Act-Term is similar.

    Rule Act-AppR. In this case, $\trans{\appe{v}{e_1}}{\lambda}{\appe{v}{e_2}}$.
    Using the second hypothesis
    %, $\expSynth{\Gamma_1,\Gamma_2,\Gamma_3}{\appe{v}{e_1}}{T}{\Gamma_3}$, 
    and rule \rulenameexpApp
    we know that 
    (a) $\expSynth{\Gamma_1,\Gamma_2,\Gamma_3}{v}{\TFun[]{U}{T}}{\Gamma_1',\Gamma_2,\Gamma_3}$
    and (b) $\expAgainst{\Gamma_1',\Gamma_2,\Gamma_3}{e_1}{U}{\Gamma_3}$.
    By induction hypothesis on~\cref*{it:pres-against}, using (b), we get 
    $\expAgainst{\Gamma_1',\Gamma_2',\Gamma_3}{e_2}{U}{\Gamma_3}$. Applying weakening and strengthening 
    to (a) we get $\expSynth{\Gamma_1,\Gamma_2',\Gamma_3}{v}{\TFun[]{U}{T}}{\Gamma_1',\Gamma_2',\Gamma_3}$
    Apply E-App to conclude that $\expSynth{\Gamma_1,\Gamma_2',\Gamma_3}{\appe{v}{e_2}}{T}{\Gamma_3}$.

    Rules Act-AppL, Act-TAppE, Act-LetE, Act-LetE*, Act-PairL, Act-PairR, Act-MatchE are similar.
    \item  
    Inversion of \rulenameexpCheck on the second hypothesis yields 
    $\expSynth{\Gamma_1,\Gamma_2,\Gamma_3}{e_1}{U}{\Gamma_3}$ and ${U}=_\alpha{T}$.
    By induction hypothesis on~\cref*{it:pres-synth} we know that 
    $\expSynth{\Gamma_1,\Gamma_2',\Gamma_3}{e_2}{U'}{\Gamma_3'}$ with 
    $\isEquiv{U'}{U}$. By \rulenametconvSymm
    and \rulenametconvTrans we know that $\isEquiv{T}{U'}$. Using \rulenameexpCheck
    we conclude that $\expAgainst{\Gamma_1,\Gamma_2',\Gamma_3}{e_2}{T}{\Gamma_3'}$.
    % We detail the case for rule Act-Rcv. In this case, we have 
    % $\trans{\appe{\tappe{\tappe \receivek T}{U}}{x}}{\receivel xv}{\paire{v}{x}}$.
    % Applying rule \rulenameexpCheck to the first hypothesis, we get 
    % $\expSynth{\Gamma_1,\Gamma_2,\Gamma_3}{\appe{\tappe{\tappe \receivek T}{U}}{x}}{T}{\Gamma_3}$
    % and $\isEquiv TT$. By induction hypothesis on~\cref*{it:pres-synth}, we 
    % have $\expSynth{\Gamma_1,\Gamma_2',\Gamma_3}{\paire{v}{x}}{T}{\Gamma_3}$.
    % Conclude applying \rulenameexpCheck.
    \item  Rule Act-Session. Inversion of Act-Session yields $\trans{\termc{e_1}}{\sigma}{\termc{e_2}}$
    and inversion of P-Exp on the second hypothesis implies $\expAgainst[\Empty]{\Gamma_1,\Gamma_2} {\termc{e_1}} {\TUnit} {\Empty}$.
    Applying induction hypothesis on~\cref{it:pres-against} we get 
    $\expAgainst[\Empty]{\Gamma_1,\Gamma_2'} {\termc{e_2}} {\TUnit} {\Empty}$. Conclude applying P-Exp.
    
    Rule Act-Beta is similar, considering $\Gamma_2 = \Gamma_2' = \Empty$, as prescribed by L-Tau and L-Beta.

    Rule Act-Fork. Inversion of Act-Fork on the first hypothesis and
    of P-Exp on the second hypothesis 
    yields (a) $\termc{e_1}\lts{\forkl{v}}{\termc{e_2}}$ and
    (b) $\expAgainst[\Empty]{\Gamma_1} {\termc{e_1}} {\TUnit} {\Empty}$.
    By L-Fork we know that (c) $\transCtx{\Gamma_1}{\forkl{v}}{\Empty}$
    and (d) $\expAgainst[\Empty]{\Gamma_1}{v}{\TFun[]\TUnit\TUnit}{\Empty}$.
    Applying induction hypothesis on~\cref{it:pres-against} to (a), (b) and (c) we get 
    $\expAgainst[\Empty]{\Empty} {\termc{e_2}} {\TUnit} {\Empty}$.
    Inversion of \rulenameexpCheck on (d) and \rulenameexpConst for $\unitk$
    enables the application of \rulenameexpApp to deduce
    $\expSynth[\Empty]{\Gamma_1} {\appe{v}{\unitk}} {\TUnit} {\Empty}$.
    Conclude using \rulenameexpCheck, P-Exp and P-Par.

    Rule Act-New. By L-Tau, $\Gamma_2 = \Empty$. Inversion of P-Exp on the second hypothesis 
    yields $\expAgainst[\Empty]{\Gamma_1} {\termc{e_1}} {\TUnit} {\Empty}$. Applying inversion of 
    Act-New and induction hypothesis on~\cref{it:pres-against} we get 
    $\expAgainst[\Empty]{\Gamma_1, \ebind{x}{\Normal[+]T}, \ebind{y}{\Normal[-]T}} {\termc{e_2}} {\TUnit} {\Empty}$.
    Using P-Exp we know that \linebreak $\isProc[\Gamma_1, \ebind{x}{\Normal[+]T}, \ebind{y}{\Normal[-]T}]{\expp{e_2}}.$
    Conclude using P-New.

    Rule Act-JoinL. Inversion of P-Par dictates a context split: 
    $\isProc[\Gamma_1, \Gamma_2^1]{p_1}$ and $\isProc[\Gamma_1, \Gamma_2^2]{p_2}$.
    Using inversion of L-Par and applying induction hypothesis on \cref{it:pres-proc} 
    we get $\isProc[\Gamma_1, \Gamma_3]{q_1}$ and $\isProc[\Gamma_1, \Gamma_4]{q_2}$,
    where $\Gamma_2' = \Gamma_3,\Gamma_4$.
    Conclude with P-Par and L-Par.

    Rule Act-JoinR is similar.

    Rule Act-Msg. In this case, $\trans{\newp xy{p_1}}{\tau}{\newp xy{p_2}}$.
    By inversion of Act-Msg, we know that $\trans{p_1}{\parl{\receivel xv}{\sendl yv}}{p_2}$. 
    For a $\tau$-transition, the second hypothesis reads as $\isProc[\Gamma_1]{\newp xy{p_1}}$
    (recall rule L-Tau). By P-New, we have 
    $\isProc[\Gamma, \ebind x{\Normal[+]{T}} , \ebind y{\Normal[-]{T}}]{p_1}$.
%    $\isProc[\Gamma_1, \ebind xT , \ebind y{\TDual{T}}]{p_1}$.
    By L-Par, L-Rcv and L-Send, we know that 
    $\Normal[+]{T}= \TIn{U}{S}$
    and 
    $\transCtx{\Gamma_1, \ebind{x}{\TIn{U}{S}}, \ebind{y}{\TOut{U}{\Normal[-]{S}}}}{\parl{\receivel{x}{v}}{\sendl{y}{v}}}{\Gamma_1, \ebind{x}{S},\ebind{y}{\Normal[-]{S}}}$
    and $\expAgainst[\Empty]{\Gamma_1}{v}{U}{\Empty}$.
    Applying induction hypothesis on~\cref*{it:pres-proc} we get 
    $\isProc[\Gamma_1, \ebind xS , \ebind y{\Normal[-]{S}}]{p_2}$.
    Conclude applying P-New.

    Rule Act-Bra. Premises to P-New include 
    (a) $\isProc[\Gamma_1, \ebind x{\Normal[+]{T}} , \ebind y{\Normal[-]{T}}]{p_1}$.
    By L-Match, L-Sel and L-Par we know that $\Normal[+]{T} = \TIn{(\TPApp{P}{\overline U})}{S'}$
    where $\protocolsfull$ and $\termc{C} = \termc{C_k}$ for some $k\in I$. Thus, we can rewrite (a)
    as  $\isProc[\Gamma_1, \ebind x{\TIn{(\TPApp{P}{\overline U})}{S'}} , \ebind y{\TOut{(\TPApp{P}{\overline U})}{\Normal[-]{S'}}}]{p_1}$.
    Inversion of Act-Bra yields $\trans{p_1}{\parl{\receivel xC}{\sendl yC}}{p_2}$.
    Applying induction hypothesis on~\cref*{it:pres-proc} we get
    $\isProc[\Gamma_1, \ebind x{\matchin[k]} , \ebind y{\TFOut {\overline{T_{k}}\tsubs{\overline U}{\overline\alpha}}{\Normal[-]{S'}}}]{p_2}$. 
    Apply \cref{lem:ops-dual}, \cref{it:dual-norm} and P-New to conclude.

    Rule Act-Wait. Premises to P-New include 
    $\isProc[\Gamma_1, \ebind x{\Normal[+]{T}} , \ebind y{\Normal[-]{T}}]{p_1}$.
    By L-Par, L-Wait and L-Term, we know that $\Normal[+]{T} = \TEndW$.
    Applying induction hypothesis on~\cref*{it:pres-proc} we get  
    $\isProc[\Gamma_1]{p_2}$.

    Rule Act-ParL. Inversion of Act-ParL yields $\trans{p_1}{\pi}{p_2}$. Premises to P-Par
    are (a) $\isProc[\Gamma_1, \Gamma_2^1]{p_1}$ and (b) $\isProc[\Gamma_1, \Gamma_2^2]{q_1}$, 
    where $\Gamma_2= \Gamma_2^1, \Gamma_2^2$.
    Applying induction hypothesis on~\cref*{it:pres-proc} to (a) we get $\isProc[\Gamma_1,\Gamma_3]{p_2}$,
    where $\transCtx{\Gamma_2^1}{\pi}{\Gamma_3}$. Conclude using (b) and applying P-Par.

    Rule Act-ParR is similar.

    Rule Act-Res. Premises to Act-Res are $\trans{p_1}{\pi}{p_2}$ and $\termc x,\termc y\not\in\FVL\pi$.
    Since $\termc x,\termc y\not\in\FVL\pi$, we know that $\transCtx{\Empty}{\pi}{\Empty}$.
    Premises to P-New include 
    $\isProc[\Gamma_1, \ebind x{\Normal[+]{T}} , \ebind y{\Normal[-]{T}}]{p_1}$.
    Induction hypothesis on~\cref*{it:pres-proc} yields   
    $\isProc[\Gamma_1, \ebind x{\Normal[+]{T}} , \ebind y{\Normal[-]{T}}]{p_2}$.
    Conclude using P-New.
    
    Rule Act-OpenL. By L-OpenL we know that $\Gamma_2 = \ebind{x}{\TOut{T}{S}}$.
    Inversion of P-New yields \linebreak
    $\isProc[\Gamma_1, \ebind x{\TOut{T}{S}} , \ebind{a}{\Normal[+]{U}}, \ebind b{\Normal[-]{U}}]{p_1}$.
    Inversion of L-Send, \rulenameexpCheck and \rulenameexpVar yields ${\Normal[+]{U}} = \typec T$. By induction hypothesis
    on~\cref*{it:pres-proc}, conclude that 
    $\isProc[\Gamma_1, \ebind x{S} , \ebind b{\Normal[-]{U}}]{p_2}$.

    Rule Act-OpenR is similar.

    Rule Act-CloseL. By L-Tau, $\transCtx{\Empty}{\tau}{\Empty}$.
    Inversion of P-New yields 
    $\isProc[\Gamma_1, \ebind{x}{\Normal[+]{T}}, \ebind y{\Normal[-]{T}}]{p_1}$.
    By L-Rcv, we know that $\Normal[+]{T} = \TIn{U}{S}$. By L-Rcv, L-OpenL, L-Par
    and induction hypothesis on~\cref*{it:pres-proc} we get 
    $\isProc[\Gamma_1, \ebind x{S} , \ebind{y}{\Normal[-]{S}},\ebind b{\Normal[-]{U}},\Gamma]{p_2}$
    where $\expAgainst[\Empty]{\Gamma}{a}{U}{\Empty}$. Applying \rulenameexpCheck and 
    \rulenameexpVar we know that $\Gamma = \ebind{a}{U}$.
    Conclude applying P-New twice.

    Rule Act-CloseR is similar.
  \end{enumerate}
\end{proof}

Canonical forms tell us the form of values from the form of types.

\begin{lemma}[Canonical forms]
  \label{lem:cforms}
  Let $\expSynth{\Gamma_1} vT {\Gamma_2}$.
  \begin{enumerate}
  \item If $\typec T = \TUnit$, then $\termc v = \unitk$.
  \item If $\typec T = \TFun[] UV$, then
    $\termc v = \abse xUe$ or
    $\termc v = \rece xTe$ or
    $\termc v = \forkk$ or
    $\termc v = \waitk$ or 
    $\termc v = \terminatek$ or
    $\termc v = \tappe{\tappe{\receivek}{W}}{X}$ or
    $\termc v = \tappe{\tappe{\sendk}{W}}{X}$ or
    $\termc v = \appe{\tappe{\tappe{\sendk}{W}}{X}}{u}$ or
    $\termc v = \tappe{\selecte C}{\overline W}$.
  \item If $\typec T = \TPair UV$, then $\termc v = \paire uw$.
  \item If $\typec T = \TForall \alpha \kind U$, then
    $\termc v = \tabse \alpha \kind u$ or
    $\termc v = \receivek$ or
    $\termc v = \tappe \receivek V$ or
    $\termc v = \sendk$ or
    $\termc v = \tappe \sendk V$ or
    % $\termc v = \reck$ or
    % $\termc v = \tappe \reck V$ or
    % $\termc v = \newk$ or
    $\termc v = \tappe{\selecte C}{\overline V}$ or
    $\termc v = \selecte C$ or 
    $\termc v = \newk$.
  \item If $\typeAgainst T \kinds$, then $\termc v = \termc x$.
  % \item If $\typec T = \TEndW$ or $\typec T = \TEndT$, then $\termc v = \termc x$.
  \end{enumerate}
\end{lemma}
\begin{proof}%[Proof of \cref{lem:cforms} (Canonical forms)]
  A straightforward analysis of the rules that apply to judgement
  $\expSynth{\Gamma_1} vT {\Gamma_2}$. The cases for $\typec T = \TFun[] UV$ and
  $\typec T = \TForall \alpha \kind U$ come from rules \rulenameexpAbs,
  \rulenameexpTAbs as well as from the constants that bear these two types. The
  only values for session types (types of kind $\kinds$) are variables
  (representing channel ends).
\end{proof}

\begin{lemma}[Completed processes are closed]
  \label{lem:completed-closed}
  If $\termc{p}$ is completed, then $\isProc[\Empty]{p}$.
\end{lemma}

\begin{proof}[Proof of \cref{thm:progress} (Progress)]
  By mutual rule induction in the first hypothesis, in each case.

  Rules \rulenameexpConst, \rulenameexpVar, \rulenameexpAbs and
  \rulenameexpTAbs. Constants, variables, lambda abstractions and type
  abstractions are values.

  Rule \rulenameexpLetUnit. Since $\termc{e} = \letu {e_1}{e_2}$ is not a value
  it must reduce. Distinguish two cases. When $\termc{e_1}$ is a value,
  canonical forms (\cref{lem:cforms}) give $\termc{e_1} = \unitk$ and
  $\termc{e}$ has a transition by rule Act-Let*. Otherwise, the premises to rule
  \rulenameexpLetUnit include $\expAgainst{\Gamma_1^{\kinds}}{e_1}{\TUnit}{\Gamma_2}$ and
  induction gives $\trans {e_1} \lambda {e'_1}$ (because $\termc{e_1}$ is not a
  value). Hence $\termc{e}$ reduces by rule Act-LetE*.

  Rule \rulenameexpRec. Expression $\rece xTv$ is a value

  Rule \rulenameexpApp. Since $\termc{e} = \appe {e_1}{e_2}$ is not a value
  it must reduce. Distinguish two cases.
  When $\termc{e_1}$ is not a value, by induction $\termc{e_1}$ reduces and so
  does $\termc{e}$ (by rule ActAppL).
  When $\termc{e_1}$ is a value, canonical forms (\cref{lem:cforms}) give a
  series of cases that we analyse in turn. When $\termc{e_1}$ is an abstraction
  we further distinguish two cases. If $\termc{e_2}$ is a value, then
  $\termc{e}$ reduces by rule Act-App. Otherwise monotonicity
  (\cref{lem:monotonicity}) gives $\Gamma_2 \subseteq \Gamma_1^\kinds$ hence
  $\Gamma_2^\kinds$. The induction hypothesis allow concluding that $\termc{e}$
  reduces by rule Act-AppR.
  The case for $\reck$ is similar to that of $\lambda$.
  When $\termc{e_1} = \forkk$, we further distinguish two cases. When
  $\termc{e_2}$ is a value, $\termc{e}$ reduces by rule Act-Fork. Otherwise,
  using induction expression $\termc{e}$ reduces by rule Act-AppR.
  When $\termc{e_1} = \waitk$, we further distinguish two cases. When
  $\termc{e_2}$ is a value, canonical forms give $\termc{e_2} = \termc{x}$,
  hence $\termc{e}$ reduces by rule Act-Wait. Otherwise,  monotonicity
  (\cref{lem:monotonicity}) places us in the induction hypothesis, hence
  $\termc{e_2}$ reduces (because it is not a value), and $\termc{e}$ reduces by
  rule Act-AppR.
  The subcase for $\terminatek$ is similar.
  When $\termc{e_1} = \tappe{\tappe{\receivek}{W}}{X}$, we further distinguish
  two cases. When $\termc{e_2}$ is a value, canonical forms give
  $\termc{e_2} = \termc{x}$, hence $\termc{e}$ reduces by rule Act-Rcv.
  Otherwise, monotonicity (\cref{lem:monotonicity}) places us in the induction
  hypothesis, hence $\termc{e_2}$ reduces (because it is not a value), and
  $\termc{e}$ reduces by rule Act-AppR.
  The subcases for $\tappe{\tappe{\sendk}{W}}{X}$ and
  $\appe{\tappe{\tappe{\sendk}{W}}{X}}{u}$ and $\tappe{\selecte C}{\overline W}$
  are similar.

  Rule \rulenameexpTApp. These cases are analogous to those of rule
  \rulenameexpApp.

  Rule \rulenameexpPair. Distinguish three cases for
  $\termc{e} = \paire{e_1}{e_2}$. When both $\termc{e_1}$ and $\termc{e_2}$ are
  values, $\termc{e}$ is a value. When $\termc{e_1}$ is not a value, induction
  ensures that $\termc{e_1}$ reduces and thus $\termc{e}$ has a transition by
  rule Act-PairL. Finally, when $\termc{e_1}$ alone is a value, the premises to
  rule \rulenameexpPair include
  $\expSynth{\Gamma_1^{\kinds}}{e_1}{T}{\Gamma_2}$.
  % and $\expSynth{\Gamma_2}{e_2}{U}{\Gamma_3}$.
  Monotonicity (\cref{lem:monotonicity}) gives
  $\Gamma_2 \subseteq \Gamma_1^{\kinds}$, hence $\Gamma_2^{\kinds}$ and we are
  in the induction hypothesis. Conclude with rule Act-PairR.

  Rule \rulenameexpLet. Similar to \rulenameexpLetUnit.

  Rule \rulenameexpMatch. In this case $\termc{e} = \matchwithe{e'}CxeiI$. The
  premises to the rule include $\expSynth{\Gamma_1^{\kinds}}{e'}{S}{\Gamma_2}$ and
  $\Normal[+]{S}={\TIn{(\TPApp P {\overline U})}{S'}}$. Agreement
  (\cref{lem:agreement}) gives $\typeAgainst{S}{\kind}$. Distinguish two cases.
  When $\termc{e'}$ is value, normalization of session types
  (\cref{lem:norm-st}) and canonical forms (\cref{lem:cforms}) give
  $\termc{e'} = \termc{x}$ and $\termc{e}$ reduces by rule Act-Match. Otherwise,
  the result follows by induction followed by rule Act-MatchE.

  Rule P-Exp. In this case $\termc{p} = \expp e$. Distinguish two cases. When
  $\termc{e}$ is a value, canonical forms give $\termc{e} = \unitk$ and
  $\termc{p}$ is completed. Otherwise, by induction $\trans e \lambda {e'}$ and
  we further distinguish three cases according to label $\lambda$:
  
  \begin{center}
  \begin{tabular}{cc}
    $\lambda$ & LTS rule\\\hline
    $\forkl v$ & Act-Fork\\
    $\newl xyT$ & Act-New\\
    $\sigma$ & Act-Session\\
    $\beta$ & Act-Beta\\\hline
  \end{tabular}
\end{center}

  Rule P-Par. In this case $\termc{p} = \parp{p_1}{p_2}$. By induction each
  $\termc{p_i}$ is either completed, reduces or is a deadlock. We must analyse
  eight cases. When both $\termc{p_i}$ are completed, $\termc{p}$ is completed.
  When one of the $\termc{p_i}$ is a deadlock, then so is $\termc{p}$. In case
  both $\termc{p_i}$ reduce, then $\termc{p}$ reduces by ACT-JoinR (or
  Act-JoinL). Otherwise, when $\termc{p_1}$ reduces (and $\termc{p_2}$ is
  completed), $\termc{p}$ reduces with rule Act-ParL. Finally, hen $\termc{p_2}$
  reduces (and $\termc{p_1}$ is completed), $\termc{p}$ reduces with rule
  Act-ParR.

  Rule P-New. In this case $\termc{p} = \newp xy{p_1}$. The premise to the rule
is $\isProc[\Gamma, \ebind x{\Normal[+]T} , \ebind y{\Normal[-]T}]{p_1}$. Agreement
(\cref{lem:agreement}) and the fact that duality is defined on session
types alone %(\cref{lem:duality-st}) 
establish that
$(\Gamma,\ebind x{\Normal[+]T}, \ebind y{\Normal[-]T})^\kinds$. By induction $\termc{p_1}$ is
completed, reduces or is a deadlock. Because completed processes are closed
(\cref{lem:completed-closed}), $\termc{p_1}$ cannot be completed. If
$\termc{p_1}$ is deadlocked, then so is $\termc{p}$, by definition. It remains
to check the case when $\termc{p}$ is not a deadlock and reduces. Suppose that
$\trans{p_1}{\pi}{p_2}$, with $\pi$ in the progress set for
$\termc{x},\termc{y}$. Then $\termc{p}$ reduces according to the table below.
\begin{center}
  \begin{tabular}{cc}
    $\pi$ & LTS rule\\\hline
    $\sigma$ & Act-Session\\
    $\forkl v$ & Act-Fork\\
    $\newl abT$ & Act-New\\
    $\parl{\receivel xv}{\sendl yv}$ & Act-Msg\\
    $\parl{\receivel xC}{\sendl yC}$ & Act-Bra\\
    $\parl{\closel x}{\openl y}$ & Act-Wait\\
    $\parl{\receivel xa}{\scopel abya}$ & Act-CloseL\\
    $\parl{\receivel xb}{\scopel abyb}$ & Act-CloseR\\\hline
  \end{tabular}
\end{center}
\end{proof}

%%% Local Variables:
%%% mode: latex
%%% TeX-master: "main"
%%% End:

%% file: fig-label-exps.tex
\begin{figure}[t!]
  \declrel{Labelled transition system for expressions}{$\trans e \lambda e$}
  \begin{mathpar}
    \inferrule[Act-App]{}{\trans{\appe{(\abse{x}{T}{e})}{v}}{\beta}{\termc{e}\esubs{v}{x}}}
    
	\inferrule[Act-TApp]{}{\trans{\tappe{(\tabse{\alpha}{\kind}{v})}{T}}{\beta}{\termc{v}\tsubs{T}{\alpha}}}

	\inferrule[Act-Let]{}{\trans{\lete{x}{y}{{\paire uv}}{e}}{\beta}{\termc{e}\esubs{u}{x}\esubs{v}{y}}}

	\inferrule[Act-Let*]{}{\trans{\letu{\unitk}{e}}{\beta}{\termc{e}}}

    \inferrule[Act-Rec]{}{\trans{\appe{({\rece xTv})}
        u}{\beta}{\appe{(\termc{v}\esubs{\rece xTv}{x})}u}}

    \inferrule[Act-Fork]{}{\trans{\appe{\forkk}{v}}{\forkl{v}} \unitk}
	
	\inferrule[Act-New]{}{\trans{\tappe \newk T}{\newl xyT}{\paire xy}}

	\inferrule[Act-Rcv]{}{\trans{\appe{\tappe{\tappe \receivek T}{U}}{x}}{\receivel xv}{\paire{v}{x}}}
	
	\inferrule[Act-Send]{}{\trans{\appe{\appe{\tappe{\tappe \sendk T}{U}}{v}}{x}}{\sendl xv}{x}}
	
	\inferrule[Act-Match]{}{\trans{\matchwithe xCyeiI}{\receivel
        x{C_k}}{e_k\esubs x{y_k}}}

	\inferrule[Act-Sel]{}{\trans{\appe{\tappe{\selecte{C}}{\overline T}}{x}}{\sendl x{C}}{x}}
	
	\inferrule[Act-Wait]{}{\trans{\waite{x}}{\closel x}\unitk}

	\inferrule[Act-Term]{}{\trans{\closee{x}}{\openl x}{\unitk}}

	\inferrule[Act-AppL]{\trans{e_1}{\lambda}{e_2}}
	{\trans{\appe{e_1}{e_3}}{\lambda}{\appe{e_2}{e_3}}}
	
	\inferrule[Act-AppR]{\trans{e_1}{\lambda}{e_2}}
	{\trans{\appe{v}{e_1}}{\lambda}{\appe{v}{e_2}}}
	
	\inferrule[Act-TAppE]{\trans{e_1}{\lambda}{e_2}}
	{\trans{\tappe{e_1}T}{\lambda}{\tappe{e_2}T}}
	
	% \inferrule[Act-LetE]{\trans{e_1}{\lambda}{e_2}}{\trans{\lete xy{e_1}{e_3}}{\lambda}{\lete xy{e_2}{e_3}}}
	%
	% \inferrule[Act-LetE*]{\trans{e_1}{\lambda}{e_2}}{\trans{\letu {e_1}{e_3}}{\lambda}{\letu{e_2}{e_3}}}

	\inferrule[Act-PairL]{\trans{e_1}{\lambda}{e_2}}
	{\trans{\paire{e_1}{e_3}}{\lambda}{\paire{e_2}{e_3}}}	
	
	\inferrule[Act-PairR]{\trans{e_1}{\lambda}{e_2}}
	{\trans{\paire{v}{e_1}}{\lambda}{\paire{v}{e_2}}}

	\inferrule[Act-MatchE]{\trans{e_1}{\lambda}{e_2}}
	{\trans{\matchwithe{e_1}CxeiI}{\lambda}{\matchwithe{e_2}CxeiI}}

	\inferrule[Act-LetE]{\trans{e_1}{\lambda}{e_2}}{\trans{\lete xy{e_1}{e_3}}{\lambda}{\lete xy{e_2}{e_3}}}

	\inferrule[Act-LetE*]{\trans{e_1}{\lambda}{e_2}}{\trans{\letu {e_1}{e_3}}{\lambda}{\letu{e_2}{e_3}}}
  \end{mathpar}
  \caption{Labelled transition system for expressions}
  \label{fig:lts-exps}
\end{figure}

%%% Local Variables:
%%% mode: latex
%%% TeX-master: "main"
%%% End:

%% file: fig-label-ctx.tex
\begin{figure}[t!]
    \declrel{Labelled transition system for contexts}{$\Gamma\lts{\lambda} \Gamma$\quad $\Gamma\lts{\pi} \Gamma$}
\begin{mathpar}
    \inferrule[L-Rcv]
    {\expAgainst[\Empty]{\Gamma}{v}{T}{\Empty}}
    {\transCtx{\ebind{x}{\TIn{T}{S}}}{\receivel{x}{v}}{\Gamma,\ebind{x}{S}}}
    
    \inferrule[L-Send]
    {\expAgainst[\Empty] \Gamma v T \Empty}
    {\transCtx{\Gamma, \ebind{x}{\TOut{T}{S}}}{\sendl{x}{v}}{\ebind{x}{S}}}
    
    \inferrule[L-Match]
    {\protocolsfull}
    {\transCtx{\ebind{x}{\TIn{(\TPApp \proto{\overline{U}})}{S}}}{\receivel{x}{C_k}}{\ebind{x}{\matchin[k]}}}
    
    \inferrule[L-Sel]
    {\protocolsfull}
    {\transCtx{\ebind{x}{\TOut{(\TPApp \proto{\overline{U}})}{S'}}}{\sendl{x}{C_k}}{\ebind{x}{\matchout[k]}}}
    
    \inferrule[L-Wait]{}
    {\transCtx{\ebind{x}{\TEndW}}{\closel{x}}{\Empty}}
    
    \inferrule[L-Term]{}
    {\transCtx{\ebind{x}{\TEndT}}{\openl{x}}{\Empty}}
    
    \inferrule[L-Beta]{}
    {\transCtx{\Empty}{\beta}{\Empty}}
    
    \inferrule[L-Fork]
    {\expAgainst[\Empty]{\Gamma}{v}{\TFun[]\TUnit\TUnit}{\Empty}}
    {\transCtx{\Gamma}{\forkl{v}}{\Empty}}
    
    \inferrule[L-New]{}
    {\transCtx{\Empty}{\newl xyT}{\ebind{x}{\Normal[+] T}, \ebind{y}{\Normal[-]{T}}}}
    
    \inferrule[L-Tau]{}
    {\transCtx{\Empty}{\tau}{\Empty}}
    
    \inferrule[L-OpenL]{}
    {\transCtx{\ebind{x}{\TOut{T}{S}}}{\scopel abxa}{\ebind{b}{\Normal[-]{T}},\ebind{x}{S}}}
\quad
    \inferrule[L-OpenR]{}
    {\transCtx{\ebind{x}{\TOut{T}{S}}}{\scopel abxb}{\ebind{a}{\Normal[-]{T}},\ebind{x}{S}}}
\quad
    \inferrule[L-Par]
    {\transCtx{\Gamma_1}{\pi_1}{\Gamma_3}
    \and \transCtx{\Gamma_2}{\pi_2}{\Gamma_4}}
    {\transCtx{\Gamma_1,\Gamma_2}{\parl{\pi_1}{\pi_2}}{\Gamma_3,\Gamma_4}}
\end{mathpar}
\caption{Labelled transition system for typing contexts}
\label{fig:lts-ctx}
\end{figure}
%%% Local Variables:
%%% mode: latex
%%% TeX-master: "main"
%%% End:

%% file: appendix-cfst.tex
\section{Embedding context-free session types}
\label{sec:appendix-cfst-algst}

This section presents an informal investigation of two questions that relate \algst and
context-free session types (CFST
\cite{DBLP:journals/corr/abs-2106-06658}): can we map any \algst
program into an equivalent CFST program? And: can we map any CFST program to an
equivalent \algst program?

This investigation is necessarily informal because the features of
\algst and CFST do not line up. \algst is a minimalist mostly linear
version of System~F with products, sessions, and algebraic
protocols. CFST is also based on System~F, but has a richer kind
structure that enables unrestricted as well as linear features,
variant and record types, and the full slate of context-free session
primitives.

To enable a meaningful comparison,  we focus on the linear fragment of
CFST~\cite{DBLP:journals/corr/abs-2106-06658}
using the notation of \citet{DBLP:journals/corr/abs-2203-12877},
restricted to pairs and ignoring De Bruijn indices.
In \algst, we only consider types that are not polymorphic in
protocols because CFST has no equivalent. We add linear algebraic
datatypes to cope with variant types in CFST.

% We decided to 
% present the theory of algebraic session types in a simplified form: all types are linear, 
% we do not have records and we omitted the definition of $\datak$type definitions in 
% the formalization of our system (despite the obvious similarities with $\protocolk$ definitions). 
% For that reason, this comparison should 
% be considered in a loose, rather than in a strict, sense. In the context-free side, we 
% consider pairs instead of records as well.

% and ignore De Bruijn 
%indices (as used by Costa et al.~\cite{DBLP:journals/corr/abs-2203-12877}) because they would 
%unnecessarily burden the analysis.

%\input{fig-algst-vs-cfst}

\paragraph{Non-parameterized algebraic sessions into (linear) context-free sessions} 
The translation from 
algebraic session types that do not parameterize over protocol types
(AlgST${}_0$) to context-free session types 
(CFST)  is  given by $\algsttocfst{\typec T}$, which is defined on
types in normal form, only:
\begin{mathpar}
    \algsttocfst{\TUnit} = \syntax{unit} \quad
    \algsttocfst{\TFun[]TU} = \algsttocfst {\typec{T}} \rightarrow \algsttocfst {\typec{U}} \quad
    \algsttocfst{\TPair TU} = \algsttocfst {\typec{T}} \times \algsttocfst {\typec{U}}\quad
    \algsttocfst{\TForall{\alpha}{\kindt} T} = \forall \alpha:\kindtl. \algsttocfst {\typec{T}}\\
	\algsttocfst{\typec{\alpha}} = \alpha\qquad
	%% sessions
	\algsttocfst{\TIn{T}{U}}= ?\algsttocfst{\typec{T}}; \algsttocfst{\typec U} \qquad
	\algsttocfst{\TOut{T}{U}}= !\algsttocfst{\typec T}; \algsttocfst{\typec U} \qquad 
	\algsttocfst{\TEndT} = \syntax{skip} \qquad 
	\algsttocfst{\TEndW} = \syntax{skip}\\
	\algsttocfst{\TIn{\TPApp{P}{\overline U}}{S}}= P_U; \algsttocfst{\typec S} 
	\text{ where } 
	P_U \doteq \&\{ C_i\colon
	\algsttocfst{\TFIn {\overline{T_i[\overline U/\overline\alpha]}}{\TEndT}} \}_{i\in I}
	\text{ for }
\protocolk\ \typec{\proto\ \overline{\alpha}}~\blk{=}~\typec{\{{\termc{\constr_i}\ \overline{T_i}}\}_{i\in I} } \\
	\algsttocfst{\TOut{\TPApp{P}{\overline U}}{S}}= P_U; \algsttocfst{\typec S} 
	\text{ where } 
	%% \algsttocfst{\polarizedparam{-}{T}} = \syntax{dual}\, \algsttocfst{\typec T} \\
	P_U \doteq \oplus\{ C_i\colon
        \algsttocfst{\TFOut {\overline{T_i[\overline U/\overline\alpha]}}{\TEndT}} \}_{i\in I}
        \text{ for }
	\protocolk\ \typec{\proto\ \overline{\alpha}}~\blk{=}~\typec{\{{\termc{\constr_i}\ \overline{T_i}}\}_{i\in I} }
    \end{mathpar}

Functional types and 
message exchanges are naturally mapped onto their counterparts. The end of a session is mapped 
to $\syntax{skip}$, whereas algebraic protocols are translated to
type names defined as internal choices or branches, depending on whether they appear 
in sending or receiving positions. The choice labels coincide with the protocol
constructors. 

We claim that if $\isEquiv{T}{U}$, then $\algsttocfst{\typec T}\simeq\algsttocfst{\typec U}$.

\input{fig-cfst-vs-algst}

\paragraph{Linear context-free sessions embedded into algebraic sessions} 
In this direction, we omit polymorphism to obtain a lighter, more
perspicuous notation. Our comparative approach scales to the
polymorphic setting, but at the price of high notational overhead. 

The translation of the \emph{linear} and \emph{monomorphic} fragment of CFST 
to AlgST is split between the functional and the session typed fragments,
which are mutually defined and presented in~\cref{fig:cfst-vs-algst-f}. The translation of a
context-free session type $T$ to the algebraic setting is given by $\cfsttoalgst T$. 
The functional operators have natural counterparts in the algebraic side.
When translated to the isorecursive setting, functional recursive types are represented 
by a datatype enclosed with an explicit unfold. Variant types are mapped to datatypes whose 
constructors are the respective labels. Session types are materialized by protocol definitions:
each session type originates a protocol emerging from the conversion of a sequential composition of 
types into a sequence of (algebraic) types with kind $\kindp$. 
This conversion is governed by function $\cfsttoprot{\cdot}$ that maps $\syntax{skip}$
to the empty sequence and uses polarities to enforce the direction of communication.
In the session typed fragment, recursive types and choices are represented by protocol definitions.
Type names are considered generative and we relax \algst to admit 
overloaded constructor names that can be reused across protocol definitions.

It remains to define suitable adapters in the places where CFST relies
on type equivalence of equirecursive types. To this end we generate
isomorphisms as witnesses of equivalence proofs using the rule system
proposed by \citet{DBLP:journals/corr/abs-2203-12877}. 
The tricky part of constructing such an isomorphism is its extension
into session types. 
For simplicity, we consider just one side of the isomorphism
$\mathsf{iso} : \cfsttoalgst{T} \rightarrow \cfsttoalgst{U}$ 
where $T \simeq U$ are session types:
\begin{lstlisting}
iso t = let (u, u') = new [|U|] in fork $ delta t u' |> u
\end{lstlisting} %$
This definition relies on a process $\delta : \cfsttoalgst{T} \to
\TDual{\cfsttoalgst{U}} \to \TUnit$ that consumes two (dual) channel
ends and implements a forwarder that compensates for the differences
between $T$ and $U$. The witness processes $\delta$ are 
defined for each pair of equivalent session types, according 
to the rules. We sketch some witness 
processes in \cref{tab:iso-witness-sessions} and provide a more detailed analysis of the isomorphism.

For the sake of simplicity, in all cases in \cref{tab:iso-witness-sessions}, we elide an initial administrative \lstinline|match| 
and \lstinline|select| step resulting from the definition of $\cfsttoalgst{T\colon \mathsf{S^{lin}}}$. 
The case for $T\simeq U$ where $T = \syntax{Skip}$ and $U = \syntax{Skip}$ would be 
formally defined by:
\begin{lstlisting}
delta t u = match u with {
  XU u -> let t = select XT t in 
          wait u 
          |> terminate t }
\end{lstlisting}
Process $\delta$ \lstinline|match|es with $X_U$ on channel $u$ and \lstinline|select|s $X_T$ on channel $t$, as prescribed by the
protocol definitions resulting from $\cfsttoalgst{\syntax{Skip}}$ (see~\cref{fig:cfst-vs-algst-f}). According
to the definition of $\cfsttoalgst{\syntax{Skip}}$ and \rulenameexpMatch, we are left with type $\TEndW$ 
on $u$ and $\TEndT$ on $t$, so we \lstinline|wait| on $u$ and \lstinline|terminate| on $t$. 

In general, $\delta$ inserts the explicit unfolds imposed by the translation to the isorecursive setting.
We showcase our approach for the equivalence $Y \simeq U$ where type $Y$ is defined by $T$, $Y \doteq T$.
The definition of $\delta$ follows the equivalence rule of \citet{DBLP:journals/corr/abs-2203-12877}
as shown below.
A witness process $\delta$ for the embedding $\mathsf{iso} : \cfsttoalgst{Y} \rightarrow \cfsttoalgst{U}$  
builds on a process $\delta_1$ witnessing $T \simeq U$ and is defined by:

\begin{minipage}{0.6\textwidth}
\begin{lstlisting}
delta y u = match u with {
   XU u -> let y = select XY y in 
           let t = select UnfoldY y in
           delta1 t u }
\end{lstlisting}
\end{minipage}
\begin{minipage}{0.4\textwidth}
	$$\inferrule{Y\doteq T \\ T\,\mathsf{contr}\\\isodecl{\delta_1} T \simeq U}{\isodecl{\delta} Y \simeq U}$$
\end{minipage}
As a result of the definition of $\cfsttoalgst{Y\colon \mathsf{S^{lin}}}$, 
process $\delta$ \lstinline|match|es with $X_U$ on channel $u$ and \lstinline|select|s $X_Y$ on channel $y$, as prescribed by the
protocol definitions resulting from $\cfsttoalgst{\syntax{Y}}$ (see~\cref{fig:cfst-vs-algst-f}). Then, 
 \lstinline|select|s the explicit $\termc{\mathsf{Unfold}_Y}$ introduced by $\cfsttoprot{Y}$ and 
uses $\delta_1$ to consume the leftovers at $t$ and $u$.
In the case $X \simeq U$ where $X \doteq T$, we are provided with a witness $\delta_1$ for $T \simeq U$.
We \lstinline|select| the explicit $\termc{\mathsf{Unfold}_X}$ introduced by $\cfsttoprot{X}$ and then 
use $\delta_1$ to consume the leftovers at $t$ and $u$.

For the case of $!T; V \simeq{} !U;W$, we are provided not only with a witness $\delta_1$ to consume the remainder
of channels $v$ and $w$, but also with a function $\theta \colon \cfsttoalgst{T} \rightarrow \cfsttoalgst{U}$
that witnesses the isomorphism resulting from $T\simeq U$. The functions witnessing the isomorphism 
between types in the functional fragment are presented in~\cref{tab:iso-cfst-to-algst-f}.
Given $\delta_1$ and $\theta$, after the initial administrative \lstinline|match| and \lstinline|select|,
 process $\delta$ \lstinline|receive|s $u$ on channel $w$, 
transforms $u$ with $\theta^{-1}$ and sends the result through channel $v$. Process
$\delta_1$ consumes the remainder of $v$ and $w$.

To witness the isomorphism corresponding to $!T; V \simeq{} !U$ where $V$ is a 
terminated session (see \citet{DBLP:journals/corr/abs-2203-12877} for further details), we are given
a process $\sigma$ that consumes the explicit unfolds resulting from the isorecursive interpretation of $V$. The 
\emph{consumer} processes $\sigma$ for terminated sessions are presented in \cref{tab:iso-is-terminated}.

The other cases in \cref{tab:iso-cfst-to-algst-f,tab:iso-witness-sessions,tab:iso-is-terminated} follow
a similar reasoning. The omitted cases for monomorphic types also fit the same approach. For the polymorphic 
fragment, we would need to consider a \emph{reader monad} where one could store a map from type
variables to the processes that consume the corresponding channels or transform the 
corresponding values, depending on whether the type variables are of kind $\kinds$ or not. 
In the axiom for polymorphic (non-session typed) variables $n\simeq n$, for instance,
%\begin{mathpar}
%\axiom{}{\isodecl{\theta} n}\simeq n}
%\end{mathpar}
we could think of $\theta$ as the identity function, but we wouldn't know which type to 
assign to the bound variable: $\theta = \abse{x}{?}{x}$. Whereas to consume two channels 
typed with a polymorphic session-typed variable, we would need to know how to consume both
endpoints.
%
%$\inferrule{\isodecl{\sigma} T\checkmark}{\isodecl{\delta} n;T \simeq n}$
%
Processes to consume or transform values typed with polymorphic variables would need to 
be provided to $\delta$. For the translation of types in the polymorphic fragment, we would need to collect the free
variables occurrig in the type and use a parametrized $\datak$ or $\protocolk$ definition:
\begin{mathpar}
\cfsttoalgst{X\colon \kappa} = \typec{X\ \overline{\alpha}}
	\text{ for } X \doteq T \text{ and } \kappa \neq \mathsf{S^{lin}}
	\text{ where } \datak\ \typec{X\ \overline{\alpha}}~\blk{=}~\typec{\termc{\mathsf{Unfold}_X}\ \cfsttoalgst{T}} \text{ and }\FVL{T} = \overline\alpha
	\and 
	\cfsttoprot{X\colon \mathsf{S^{lin}}} = \typec{X\ \overline{\alpha}}
	\text{ for } X \doteq T
	\text{ where }
	\protocolk\ \typec{X\ \overline{\alpha}}~\blk{=}~\typec{\termc{\mathsf{Unfold}_X}\ \cfsttoprot{T}} 
	\text{ and }\FVL{T} = \overline\alpha
\end{mathpar}
Although we believe that polymorphic types
do not pose any problem to this analysis, we decided to keep the notation simple and
focus on the monomorphic fragment.

In this analysis we only provided witnesses for the isomorphism in one direction.
% For the
% other direction of the isomorphism, we could follow a similar approach for defining 
% $\mathsf{iso}^{-1} : \cfsttoalgst{U} \rightarrow \cfsttoalgst{T}$ 
% and then prove that $\mathsf{iso}$ (presented in \cref*{sec:cfst-vs-algst}) and $\mathsf{iso}^{-1}$ compose giving the identity.
For the
other direction of the isomorphism, we could follow a similar approach for defining 
$\mathsf{iso}^{-1} : \cfsttoalgst{U} \rightarrow \cfsttoalgst{T}$ 
and then prove that $\mathsf{iso}$ and $\mathsf{iso}^{-1}$ compose giving the identity.

\input{table-iso-cfst-to-algst}
%%% Local Variables:
%%% mode: latex
%%% TeX-master: "main"
%%% End:

%% file: fig-cfst-vs-algst.tex
\begin{figure}[t!]
  \declrel{CFST to AlgST: Translation of the functional fragment}{$\cfsttoalgst{T\colon \mathsf{T^{lin}}}=\typec T \colon \kindt$}
 
    \begin{mathpar}
    \cfsttoalgst{\syntax{unit}} = \TUnit \and
    \cfsttoalgst{T \rightarrow U} = \TFun[]{\cfsttoalgst{T}}{\cfsttoalgst{U}} \and
    \cfsttoalgst{T \times U} = \TPair{\cfsttoalgst{T}}{\cfsttoalgst{U}}\\
%    \cfsttoalgst{\forall\alpha.T} = \TForall{\alpha}{\kindt}{\cfsttoalgst{T}}\quad
%	\cfsttoalgst{\alpha} = \typec{\alpha} \\
	\cfsttoalgst{\langle \ell \colon T_\ell\rangle_{\ell\in L}} = \typec{X_L }\text{ where }\datak\ \typec{X_L}~\blk{=}
	~\typec{\{{\termc{\ell}\ \cfsttoalgst{T_\ell}}\}_{\ell\in L} }
\\
	\cfsttoalgst{X\colon \kappa} = \typec{X}
	\text{ for } X \doteq T \text{ and } \kappa \neq \mathsf{S^{lin}}
	\text{ where }
	\datak\ \typec{X}~\blk{=}~\typec{\termc{\mathsf{Unfold}_X}\ \cfsttoalgst{T}}\\
	\cfsttoalgst{T\colon \mathsf{S^{lin}}} = \TOut{X_T}{\TEndT} \text{ where }\protocolk\ \typec{X_T}~\blk{=}~\typec{\termc{X_T}\ \cfsttoprot T}\text{ and $X_T$ is a fresh name.}
%% old version, with polymorphic types
%		\datak\ \typec{X\ \overline{\alpha}}~\blk{=}~\typec{\termc{\mathsf{Unfold}_X}\ \cfsttoalgst{T}} \text{ and }\FVL{T} = \overline\alpha\\
%	\cfsttoalgst{\langle \ell_1 \colon T_1, \ldots, \ell_n \colon T_n\rangle} = \TPApp {X_{\langle\ell_1\cdots \ell_n\rangle}}{\overline{\alpha}} \text{ where }\datak\ \typec{X_{\langle\ell_1\cdots \ell_n\rangle}\ \overline{\alpha}}~\blk{=}~\typec{\termc{\ell_1}\ \cfsttoalgst{T_1} \mid \ldots \mid \termc{\ell_n}\ \cfsttoalgst{T_n}} 
%	\text{ and }\cup_{i=1}^n \FVL{T_i} = \overline\alpha
%\\
%	\cfsttoalgst{T\colon \mathsf{S^{lin}}} = \TOut{X_T}{\TEndT} \text{ where }\protocolk\ \typec{X_T}~\blk{=}~\typec{\termc{X_T}\ \cfsttoprot T}\text{ and $X_T$ is a fresh name.}
    \end{mathpar}
    
    \declrel{CFST to AlgST:Translation of the session typed fragment}{$\cfsttoprot{T\colon \mathsf{S^{lin}}}=\typec{\overline {T\colon \kindp}}$}
    
    \begin{mathpar}
    \cfsttoprot{\syntax{skip}} = \emptyseq \and \cfsttoprot{! T} = \cfsttoalgst{T} \and
    \cfsttoprot{? T} = \TMinus {\cfsttoalgst{T}}\and
    \cfsttoprot{T; U} = \cfsttoprot{T} \cfsttoprot{U}
    	\\
	\cfsttoprot{\oplus\{\ell\colon T_\ell\}_{\ell\in L}} = \typec{X_L}
        \text{ where } \protocolk\ \typec{X_L}~\blk{=} ~\typec{\{{\termc{\ell}\ \cfsttoprot{T_\ell}}\}_{\ell\in L}}
	\\
	\cfsttoprot{\&\{\ell\colon T_\ell\}_{\ell\in L}} = \TMinus
        {\typec{X_L}}
        \text{ where } \protocolk\ \typec{X_L}~\blk{=}
        ~\typec{\{{\termc{\ell}\ \cfsttoprot{\syntax{dual}\,T_\ell}}\}_{\ell\in L}}
    \\
	\cfsttoprot{X\colon \mathsf{S^{lin}}} = \typec{X}
	\text{ for } X \doteq T
	\text{ where }
	\protocolk\ \typec{X}~\blk{=}~\typec{\termc{\mathsf{Unfold}_X}\ \cfsttoprot{T}} 
	%
	%~\typec{\termc{\ell_1}\ \cfsttoprot{T_1} \mid \ldots \mid \termc{\ell_n}\ \cfsttoprot{T_n}} 
    \end{mathpar}
\caption{Translation of linear context-free session types to algebraic session types}
\label{fig:cfst-vs-algst-f}
\end{figure}

%%% Local Variables:
%%% mode: latex
%%% TeX-master: "main"
%%% End:

%% file: table-iso-cfst-to-algst.tex
\begin{table}
\begin{tabular}{| c | l |}
\hline
Equivalence rule & Isomorphism witness\\
$T\simeq U \colon \mathsf{S^{lin}}$&$\cfsttoalgst{T} \to
\TDual{\cfsttoalgst{U}} \to \TUnit$\\\hline
%% E-Skip
$\axiom{}{\isodecl{\delta} \mathsf{Skip}\simeq \mathsf{Skip}}$ & 
\begin{lstlisting}
delta t u = wait u 
        |> terminate t  
\end{lstlisting}
\\\hline
 %% E-IdL
$\inferrule{X\doteq T \\ T\,\mathsf{contr}\\\isodecl{\delta_1} T \simeq U}{\isodecl{\delta} X \simeq U}$ & 
\begin{lstlisting}
delta x u = let t = select UnfoldX x in 
        delta1 t u
\end{lstlisting}
 \\\hline
 %% E-IDSeqL
$\inferrule{X\doteq T \\ T\,\mathsf{contr}\\\isodecl{\delta_1} T;V \simeq U}{\isodecl{\delta} X;V \simeq U}$ & 
\begin{lstlisting}
delta x u = let t = select UnfoldX x in 
        delta1 t u
\end{lstlisting}
 \\\hline
 %% E-MsgSeq2
$\inferrule{\isodecl{\theta} T \simeq U\\\isodecl{\delta_1} V \simeq W}{\isodecl{\delta} !T; V \simeq !U;W}$ & 
\begin{lstlisting}
delta v w = let (u, w) = receive w in 
        let t = theta^-1 u 
        let v = send t v in 
        delta1 v w
\end{lstlisting}
 \\\hline
 %% E-MsgSeq1L
$\inferrule{\isodecl{\theta} T \simeq U\\\isodecl{\sigma} V \checkmark}{\isodecl{\delta} !T; V \simeq !U}$ & 
\begin{lstlisting}
delta t u = let (u', u) = receive u in 
        let t' = theta^-1 u'
        let v = send t' t in
        let v = sigma v in
        wait u 
        |> terminate v
\end{lstlisting}
 \\\hline
%% E-Msg
$\inferrule{\isodecl{\theta} T \simeq U }{\isodecl{\delta} !T \simeq !U}$ & 
\begin{lstlisting}
delta t u = let (u', u) = receive u in
        let t' = theta^-1 u' in
        let t = send t' t in
        wait u 
        |> terminate t
\end{lstlisting}
 \\\hline
%% E-Choice
$\inferrule{\isodecl{\delta_\ell} T_\ell \simeq U_\ell\\ (\forall \ell \in L)}{\isodecl{\delta} \oplus\{\ell\colon T_\ell\}_{\ell \in L} \simeq \oplus\{\ell\colon U_\ell\}_{\ell \in L}}$ & 
\begin{lstlisting}
delta t u = match u with {
          k u -> let t = select k t in 
                 deltak t u }
\end{lstlisting}
 \\\hline
%% E-ChoiceSeqL
$\inferrule{\isodecl{\delta_1} \oplus\{\ell\colon T_\ell;U\}_{\ell \in L} \simeq V}{\isodecl{\delta} \oplus\{\ell\colon T_\ell\}_{\ell \in L};U \simeq V}$ & 
\begin{lstlisting}
delta t u = delta1 t u
\end{lstlisting}
 \\\hline
 %% E-SkipSeqL
$\inferrule{\isodecl{\delta_1} T \simeq U}{\isodecl{\delta} \mathsf{Skip}; T \simeq U}$ & 
\begin{lstlisting}
delta t u = delta1 t u
\end{lstlisting}
 \\\hline
% %% E-IndexSeq1L
%$\inferrule{\isodecl{\sigma} T\checkmark}{\isodecl{\delta} n;T \simeq n}$ & 
%\begin{lstlisting}
%delta t u = ?
%\end{lstlisting}
% \\\hline
    %%
\end{tabular}
\caption{Processes that consume two dual session endpoints.}
\label{tab:iso-witness-sessions}
\end{table}

\begin{table}
\begin{tabular}{| c | l |}
\hline
Equivalence rule & Isomorphism witness\\
$T\simeq U \colon \kappa$ with $\kappa\neq \mathsf{S^{lin}}$&$\theta \colon \cfsttoalgst{T} \rightarrow \cfsttoalgst{U}$\\\hline
%% E-Unit
$\axiom{}{\isodecl{\theta} ()_{\mathsf{lin}}\simeq ()_{\mathsf{lin}}}$ & 
\begin{lstlisting}
theta = lambda x:Unit.x
\end{lstlisting}
\\\hline
%%% E-Index
%$\axiom{}{\isodecl{\theta} n}\simeq n}$ & 
%\begin{lstlisting}
%theta = lambda x:Unit.x
%\end{lstlisting}
%\\\hline
%% E-Arrow
$\inferrule{\isodecl{\theta_1} T \simeq U\\\isodecl{\theta_2} V \simeq W}{\isodecl{\theta} T\rightarrow V \simeq U \rightarrow W}$ & 
\begin{lstlisting}
theta t = lambda x:[|U|].theta2 $ t $ theta^-1 x
\end{lstlisting}
 \\\hline
%% E-Rcd
$\inferrule{\isodecl{\theta_1} T \simeq U\\\isodecl{\theta_2} V \simeq W}{\isodecl{\theta} T\otimes V \simeq U \otimes W}$&
\begin{lstlisting}
theta t = <theta1 $ fst t,theta2 $ snd t> 
\end{lstlisting}
 \\\hline
  %% E-IdL
$\inferrule{X\doteq T \\ T\,\mathsf{contr}\\\isodecl{\theta_1} T \simeq U}{\isodecl{\theta} X \simeq U}$ & 
\begin{lstlisting}
theta x = let t = select UnfoldX x in 
      let u = theta1 t in
      u
\end{lstlisting}
%  \\\hline
%  %% E-Quant
% $\inferrule{\isodecl{\theta_1} T \simeq U}{\isodecl{\theta} \forall T \simeq \forall U}$&
% \begin{lstlisting}
% theta t = Lambda alpha: T. theta1 $ t alpha
% \end{lstlisting}
 \\\hline
\end{tabular}
\caption{Transformation of non-session typed expressions.}
\label{tab:iso-cfst-to-algst-f}
\end{table}

\begin{table}
\begin{tabular}{| c | l |}
\hline
Is terminated predicate & Witness\\
$T \colon \mathsf{S^{lin}}$&$\sigma \colon \cfsttoalgst{T} \rightarrow \TEndT$\\\hline
%% E-Unit
$\axiom{}{\isodecl{\sigma} \mathsf{Skip}\checkmark}$ & 
\begin{lstlisting}
sigma t = t
\end{lstlisting}
\\\hline
%% E-Arrow
$\inferrule{\isodecl{\sigma_1} T \checkmark\\\isodecl{\sigma_2} U \checkmark}{\isodecl{\sigma} T;U \checkmark}$ & 
\begin{lstlisting}
sigma t = let t' = sigma1 t in
      sigma2 t'
\end{lstlisting}
 \\\hline
%% E-Rcd
$\inferrule{T\doteq U\\ \isodecl{\sigma_1} T \checkmark}{\isodecl{\sigma} X \checkmark}$&
\begin{lstlisting}
sigma x = let t = select UnfoldX x in
      sigma1 t
\end{lstlisting}
 \\\hline
\end{tabular}
\caption{Processes that consume the explicit unfolds created by the translation.}
\label{tab:iso-is-terminated}
\end{table}

%%% Local Variables:
%%% mode: latex
%%% TeX-master: "main"
%%% End: